\documentclass[12pt]{amsart}
\usepackage[dvips]{color}
\usepackage{amsmath}
\usepackage{amsxtra}
\usepackage{amscd}
\usepackage{amsthm}
\usepackage{amsfonts}
\usepackage{amssymb}
\usepackage{eucal}
\usepackage{epsfig}
\usepackage{graphics}
\def\({\left(}
\def\){\right)}
\newcommand{\Rdr}{R_\mathrm{dress}}

\newcommand{\gammabs}{\scalebox{.8}{\boldmath$\gamma$}}
\newcommand{\mubs}{\scalebox{.8}{\boldmath$\mu$}}
\newcommand{\betabs}{\scalebox{.8}{\boldmath$\beta$}}

\newcommand{\mub}{\mbox{\boldmath$\mu$}}

\newcommand{\vphis}{\scalebox{.7}{\boldmath$\varphi$}}
\newcommand{\vphi}{\mbox{\boldmath$\varphi$}}
\newcommand{\betab}{\mbox{\boldmath$\beta$}}
\newcommand{\gammab}{\mbox{\boldmath$\gamma$}}

\newcommand{\cb}{\mathbf{c}}
\newcommand{\bb}{\mathbf{b}}

\newcommand{\tb}{\mathbf{t}}

\newcommand{\xin}{\frac{1-\nu}{\nu}}
\newcommand{\xinn}{{\textstyle \frac{1-\nu}{\nu}}}

\newcommand{\nn}{\nonumber}
\newcommand{\bea}{\begin{eqnarray}}
\newcommand{\ena}{\end{eqnarray}}
\def\bel{\begin{eqnarray}}
\def\enl{\end{eqnarray}}
\newcommand{\be}{\begin{eqnarray*}}
\newcommand{\en}{\end{eqnarray*}}
\newcommand{\ba}{\begin{array}}
\newcommand{\ea}{\end{array}}

\newcommand{\C}{{\mathbb C}}
\newcommand{\Z}{{\mathbb Z}}

\newcommand{\slt}{\mathfrak{sl}_2}

\newcommand{\res}{{\rm res}}

\newcommand{\Tr}{{\rm Tr}}

\newenvironment{tenumerate}{
  \begin{enumerate}
  
  }{\end{enumerate}}
\newcommand{\bi}{\begin{tenumerate}}
\newcommand{\ei}{\end{tenumerate}}
\newcommand{\isoto}[1][]%
{{\mathop{\buildrel{\sim}\over\longrightarrow}\limits_{#1}}}

\def\[{\left[}
\def\]{\right]}
\newcommand{\la}{\lambda}

\newcommand{\al}{\alpha}

\newcommand{\s}{\sigma}
\newcommand{\z}{\zeta}

\numberwithin{equation}{section}
\newtheorem{thm}{Theorem}[section]
\newtheorem{prop}[thm]{Proposition}

\newtheorem{definition}[thm]{Definition}

\def\half{\textstyle{\frac  1 2}}

\newcommand{\lb}{\mathbf{l}}

\def\bi{\mathbf{i}}

\def\Io{I_\mathrm{odd}(m)}
\newcommand{\zbz}{z,\bar{z}}
\definecolor{5/18}{rgb}{0.9,0,0.7}
\definecolor{5/19}{rgb}{0.9,0,0.7}
\definecolor{5/19comment}{rgb}{0.9,0.7,0}
\begin{document}
\begin{title}[Fermionic structure in sine-Gordon model]
{Fermionic structure in the sine-Gordon model:
form factors and null-vectors}
\end{title}
\author{M.~Jimbo, T.~Miwa and  F.~Smirnov}
\address{MJ: Department of Mathematics, 
Rikkyo University, Toshima-ku, Tokyo 171-8501, Japan}
\email{jimbomm@rikkyo.ac.jp}
\address{TM: Department of 
Mathematics, Graduate School of Science,
Kyoto University, Kyoto 606-8502, 
Japan}\email{tmiwa@math.kyoto-u.ac.jp}
\address{
FS\footnote
{Membre du CNRS}: 
Hamilton Mathematical Institute and School of Mathematics,
Trinity College, Dublin 2, Ireland
\newline
Laboratoire de Physique Th{\'e}orique et
Hautes Energies, Universit{\'e} Pierre et Marie Curie,
Tour 13, 4$^{\rm er}$ {\'e}tage, 4 Place Jussieu
75252 Paris Cedex 05, France}\email{smirnov@lpthe.jussieu.fr}

\begin{abstract}
The form factor bootstrap in integrable quantum field theory 
allows one to capture local fields in terms of 
infinite sequences of Laurent polynomials called `towers'. 
For the sine-Gordon model, towers are systematically described
by fermions introduced some time ago 
by Babelon, Bernard and Smirnov.
Recently the authors developed a new method for evaluating 
one-point functions of descendant fields, using yet another fermions
which act on the space of local fields.
The goal of this paper is to establish that these two fermions 
are one and the same object. This opens up a way for answering the 
longstanding question about how to identify
precisely towers and local fields. 
\end{abstract}

\maketitle

\section{Introduction}

The famous
sine-Gordon (sG) model is described by the action
\begin{align}
\mathcal{A}^\mathrm{sG}=\int 
\Bigl[\frac{1}{16 \pi} (\partial _\mu\vphi (x))^2+
\frac{\mub ^2}{\sin\pi\beta ^2} 2\cos(\beta\vphi(x))\Bigr]
d^2x\label{action}\,.
\end{align}
In this paper we use 
the parameter
$$
\nu=1-\beta^2\,,
$$
following the convention 
in our previous works \cite{HGSIV,OP,HGSV}
\footnote{The parameter $\nu$ is related to
$\xi=\xi^{\mathrm{FS}}$
in \cite{book} and  
$\xi=\xi^{\mathrm{SL}}$ in 
\cite{Lukalpha} by
${\textstyle\frac{1-\nu}\nu}=\xi ^\mathrm{FS}/\pi=\xi ^\mathrm{SL} $.}  . 

In our opinion, 
the sG model
is an ideal playground for developing new methods of Integrable Quantum
Field Theory (IQFT). 
On the one hand,
this model is sufficiently
complicated. Its spectrum contains,
together with usual particles and their bound states 
(breathers),
 topologically non-trivial particles: solitons.
On the other hand,
the model is intimately related 
with the simplest non-trivial quantum affine algebra
$U_q(\widehat{\mathfrak{sl}}_2)$,  
so, its study does
not involve purely technical complications coming from 
considering 
quantum groups of higher rank.
The latter property is closely related to the fact that 
in the classical case the (quasi)-periodic solutions
of the sG equation are related to hyper-elliptic Riemann surfaces, 
which represent the 
simplest, but still  non-trivial example of algebraic curves. 
All that should be clearly understood by
a researcher who is interested in serious investigation of the sG model. 

During the period of rapid development  of IQFT,
 several important results were obtained for the sG model.
They include the discovery of the exact S-matrix \cite{zamS}, formulation
of axioms for the form factors and solving equations constituting these axioms
\cite{smi86,smikir,book}. 
Around the same period there appeared
the work \cite{alzam} which is very important for us. This paper
investigated the equivalence
of the form factor bootstrap to the Operator Product Expansion (OPE) 
appearing in the context of Perturbed Conformal Field Theory (PCFT) 
combined with the knowledge of the
one-point functions.

Let us discuss  the form factors in the sG model. 
In the original works \cite{smi86,book} they were
constructed for operators most relevant to physics:
energy-momentum tensor, topological current and
disorder operators. 
With the rapid development of  Conformal Field Theory (CFT) which started with
the famous paper \cite{BPZ},
the following interesting question arose: Find
the form factors of all
the local operators which are described in the ultraviolet limit by CFT. 
The papers \cite{smi90,ReshSmi}
were important in understanding this problem 
together with the paper \cite{alzam} cited before.
Due to the latter paper it became clear that 
the ideal object for this study is not the space
of descendants of the degenerate fields, but 
rather the space of descendants of the generic primary field
\begin{align}
\Phi _\al(\zbz)=e^{i\al
\ \frac {\nu}{2\sqrt{1-\nu}}
\ \vphis (\zbz)}\,,
\label{Phial}   
\end{align}
with arbitrary $\al$. 
The
normalisation coefficient $\frac {\nu}{2\sqrt{1-\nu}}$ is introduced for  convenience.
The point is that considering generic 
$\nu$
and $\alpha$ we avoid resonances, and the correspondence
between the sG operators and their ultra-violet CFT counterparts becomes
one-to-one. 
According to
this logic, the first task is to compute 
the form factors 
of the exponential operators $\Phi _\al(0)$. 
This was done in the
paper \cite{Lukalpha}. This paper required 
the knowledge of zero-particle form factors (one-point functions)
in infinite volume found in \cite{LukZam}. 
It should be said also that the method of \cite{Lukalpha} originates
in the algebraic study of 
correlation functions and form factors  
for 
lattice models \cite{JM}.

It has been said that the sG model possesses
the advantage of being sufficiently complicated 
while avoiding unnecessary difficulties of purely algebraic nature.
Here we want to explain what the words 
``sufficiently complicated" mean. 
There
are models of IQFT with much simpler, diagonal, 
S-matrices (sinh-Gordon model, Lee-Yang model, $Z_N$ model, 
and many others). 
For these models the form factor bootstrap considerably simplifies.
However, to our mind it simplifies too much 
making obscure the mathematical structure
of the solution. When the answer is simple 
it may allow different accidental representations.
The sG model is sufficiently rigid for form factors.
Irrespectively of the method used to derive them,  
one obtains essentially the same formulae 
given by a certain integral transformation.
 The integrals involved in this transformation can be understood as quantum deformation
of the hyper-elliptic Abelian integrals \cite{smiabel}. 
The analogy with the classical case was crucial for realising this fact.
While quite useful,
this analogy is not so straightforward.
In classical mathematics, when
the solution to a differential equation is given by 
an integral transformation,
different solutions are parametrised by different contours. 
After quantisation the differential equations are replaced by 
a certain Riemann-Hilbert problem (see Section \ref{axioms}),
and the contours
are replaced by polynomials which one can insert under the integral. 
These polynomials label different local operators.
In this paper we consider only soliton form factors  for the
local operators which do not change the topological charge. So, the number
of particles is even, say $2n$.
More precisely, 
with each $n$ there is associated 
a polynomial entering the integral formula for $2n$-particle form factors,
and these polynomials are mutually related by a certain recurrence
relation. 
Hence a local operator is represented by an infinite sequence of such 
polynomials which we call a tower.
The precise definition is given in  Section \ref{FormulasforFF}.

We would like to make one more historical comment. 
When the method of counting local fields became clear, the following important
problem arose. 
It is well known that for special values of $\al$,
which correspond to degenerate primary fields, some descendants
vanish. Conventionally these vanishing descendants are called 
null vectors. On the other hand the number of towers corresponding to
descendants of a primary field is independent of $\al$. 
The only solution to this apparent contradiction may consist 
in vanishing of 
the integrals defining the form factors in some cases. 
This is indeed the case.
There are different reasons
for integrals to vanish, the most important being the  Riemann bilinear identity for 
quantum Abelian integrals \cite{smiriemann}.

Returning to the general case, let us  fix  a generic
$\al$ and explain how the space of towers is organised.
A convenient 
language for that was introduced in the paper \cite{BBS}. 
We slightly generalise the results
of this paper and change the notation. 
The local integrals of motion act on the local operators by commutators.
The form factors obtained by this action are easy to compute, so, 
we shall ignore these descendants by
local integrals of motion.
The real problem is to describe the quotient space.
To this end,
following \cite{BBS} we introduce the fermionic creation operators
$\psi ^*_{2j-1}$, $\chi ^*_{2j-1}$, 
$\bar\psi ^*_{2j-1}$, $\bar\chi ^*_{2j-1}$ 
acting on towers, $j=1,2,\cdots$.
The operators $\psi ^*_{2j-1}$, 
$\chi ^*_{2j-1}$ correspond to the right chirality, and the operators 
$\bar\psi ^*_{2j-1}$, $\bar\chi ^*_{2j-1}$ to the left chirality. 
We call these fermions Babelon-Bernard-Smirnov (BBS) fermions. 

Let  $M^{(\star)}_0=\{M^{(n)}\}_{n=0}^\infty$ 
denote the tower corresponding to 
the primary field $\Phi _\al$. 
The space of towers $L^{(\star)}_{\mathcal{O}_\al}$ corresponding to
local operators $\mathcal{O}_\al$ is obtained acting on
$M^{(\star)}_0$ by integrals of motion and BBS fermions.
The latter  must satisfy the restriction
\begin{align}
\#(\psi ^*)+\#(\bar\psi ^*)=\#(\chi ^*)+\#(\bar\chi ^*)\,.\label{charge0}
\end{align}
Later we shall say that the towers  satisfying \eqref{charge0} have charge $0$.
It is also useful to introduce the weight of a
local operator $\mathcal{O}_\al$ by
$$
m={\textstyle\frac 1 2}(\#(\psi ^*)-\#(\chi ^*)+\#(\bar\chi ^*)-\#(\bar\psi ^*))\,.
$$
The local operators of weight $m$ correspond to 
$\Phi _{\al +2m\frac{1-\nu}{\nu}}$,
and its Virasoro descendants.

As we mentioned already,  there is a
one-to-one correspondence between the operators in 
the sG model
and in the corresponding ultraviolet CFT.  
There is certain arbitrariness in choosing the latter. We
prefer to split $ 2\cos(\beta\vphi(x))$ in \eqref{action} into the sum $e^{-i\beta\vphis(x)}+
e^{i\beta\vphis(x)}$ giving the first term to the 
CFT action and considering the second one
as the perturbation. The CFT in question is nothing but 
the complex Liouville model. The fields $\Phi _{m\frac {1-\nu}\nu}$
are degenerate. For 
$m\ge0$ they correspond to the fields from the first row of the Kac table. So, certain descendants
of these fields must vanish. This circumstance was the main subject of \cite{BBS}. Namely, it was shown 
that the Riemann bilinear identity and some additional simpler properties of quantum Abelian integrals imply
certain relations between the fermionic descendants for the degenerate fields. This will be discussed in 
Section \ref{null}. 
Taking these null vectors into account,  
we arrive at the correct number of local operators.
However an important question 
 was left unanswered in \cite{BBS}: 
to identify precisely 
the descendants by BBS fermions and the usual Virasoro descendants.
We shall address this issue in this paper. 

Being unable to solve this problem,
 the authors of \cite{BBS} concentrated on the classical
limit showing that the description of the null vectors in terms of the fermions provides a new formulation
of the classical hierarchy. Actually, they considered 
only right chiral descendants, so, the hierarchy
in question was that of the
Korteweg-de-Vries (KdV) equation. 
Another  way of counting 
local operators was 
explained from a representation theory viewpoint in \cite{Fothers}.

Now we would like to discuss a seemingly  completely different subject. 
Several years ago, the present authors together
with H. Boos and Y. Takeyama started a joint work on 
correlation functions for the  XXZ spin chain. 
We have been motivated by a strong feeling
that the formulae known at the time were quite unsatisfactory. As a result, we found a
fermionic description of the space of quasi-local operators.

We believe our fermionic construction  to be important, so, we would like to explain its intuitive meaning.
Generally speaking, 
our understanding of quantum field theory is limited because
very few exact results are known about 
models with interaction. Our intuition relies too heavily 
on 
free fields, and this can be sometimes misleading.
To give an example in
the context of the XXZ spin chain,
the free model is the XX spin chain arising at
a particular value of the coupling constant
where the model becomes
equivalent to the lattice Dirac fermion. 
In this case, 
one diagonalises the Hamiltonian introducing the creation-annihilation
operators  by Fourier transform.
We parametrise
the 
corresponding momentum as
$ip=\log \frac {1-\z^2}{1+\z^2}$.
On the other hand,
the same $\z$-dependent
creation operators can be used for constructing quasi-local operators 
from a
given one: we just take the (anti)-commutators and develop the result around the point $\z^2=1$.
So, the same construction with Fourier transform serves 
two different goals: diagonalising the Hamiltonian, and 
describing the space of quasi-local operators.
It is important to understand that
these two procedures are completely different 
for 
models with interaction.
Moreover, before finding the fermionic description of 
the space of quasi-local fields
we even did not know that the second procedure makes sense. 
Let us be more explicit about this point.

The diagonilisation of the Hamiltonian for integrable models 
in general,  and the
XXZ spin chain in particular,
 is achieved by the
Bethe Ansatz which is best understood in its
algebraic formulation (ABA) \cite{FST}. The ABA
can be viewed as  a highly non-linear
analogue of Fourier transform
in the space of states. 
It
does not explicitly  
introduce the creation-annihilation operators, and 
cannot be used for creating the local fields 
by adjoint action as it was possible for the
free-fermion case. The real achievement of 
our works \cite{HGS,HGSII} is that we were
able to find a  set of creation operators 
$\bb^*(\z)$, $\cb ^*(\z)$ which act
on the space of quasi-local operators 
creating this entire space from the 
``primary fields"
by development around $\z ^2=1$. 
It may be said that, comparing to ABA, 
we have introduced another non-linear Fourier transform
in the space of operators. 
In doing that we used the same algebraic structures as 
the one used in ABA (quantum groups \cite{Drinfeld,Jimbo}
essentially) but in a more sophisticated way, 
including in particular the methods developed by \cite{BLZIII}.

Our fermionic operators have two notable features.
First, they indeed act on the space of quasi-local 
operators, i.e., they respect locality.
This was proved in the paper \cite{HGSII}.
Second, the partition function of 
the the equivalent 6-vertex model, 
formulated on a cylinder
with an insertion of a quasi-local operator created by the fermions, 
is expressed in terms of a single function $\omega (\z,\xi)$. 
This was proved in the paper \cite{HGSIII}.
We shall refer to the
compact direction on the cylinder 
as the Matsubara direction. 

Our next goal was to take the scaling limit in order to arrive at the $c<1$ CFT.
After the transformation to fermions has been done this scaling limit is simple.
In a certain sense we consider the fermionic construction for the lattice model as 
an existence theorem: the local operators are parameterised by parameters $\z$,
and the partition function with insertion is expressed in terms of 
$\omega (\z,\xi)$.
To consider the scaling limit it suffices to describe it for $\omega (\z,\xi)$.
Here the TBA-like equations for $\omega (\z,\xi)$ \cite{BG} are very useful.
(We use the term TBA in a  
broad sense. Actually the techniques
 used in \cite{BG,HGSIV} is 
that of \cite{DestriDeVega}.)
Comparing 
the scaling limit of $\omega (\z,\xi)$ with the CFT three-point functions,  
we find 
the relation between the description of local operators by the fermions and
the one in terms of the Virasoro algebra.

For the application to the sG model, we use
an inhomogeneous XXZ spin chain, and obtain
in the scaling limit fermions $\betab^*_{2j-1}$, $\gammab ^*_{2j-1}$, 
$\bar\betab^*_{2j-1}$, $\bar\gammab ^*_{2j-1}$, 
which we call Boos-Jimbo-Miwa-Smirnov (BJMS) fermions. 
This construction is applied to  
solving the longstanding problem of computing the one-point functions for
the sG model on the plane \cite{OP} and on the cylinder \cite{HGSV}. 
To be precise, in addition to the BJMS fermions,  
certain fermionic screening operators are used to create 
the primary fields 
$\Phi _{\al+2m\frac{1-\nu}{\nu}}$ and their descendants 
from $\Phi _\al$. 
We shall not discuss them as they will be irrelevant 
for the purpose of this paper.
Let us emphasise one more time that the relation between 
the BJMS fermionic descendants and the Virasoro 
descendants can be computed. This has been done
in the quotient space by the action of the 
local integrals of motion up to level 6 in \cite{HGSIV} and
on the level 8 in \cite{Boos}. 
For studying the one-point functions it is sufficient to work
in the quotient space. However, 
it is not quite sufficient for the goal of the present paper, we shall
comment on this point soon.

Though the BBS fermions and the BJMS fermions have been introduced
by different methods and for different reasons, there 
is a certain similarity between the two.
This similarity motivated
us to investigate the situation closely. 
On one hand we
have the form factor formulae 
for the descendants written in terms of the
BBS fermions. On the other hand we have local 
operators created by the
BJMS fermions. 
Quite generally,
if we insert such
a local operator on the cylinder
and take any eigenvectors of the Matsubara transfer-matrix 
as asymptotic conditions, then
the partition function can be expressed in terms of 
a single
function $\omega(\z,\xi)$. 
In particular, we can put an
excited state to the left and 
the ground state to the right. 
Then the infinite volume limit
in the Matsubara direction can be performed. 
It is clear that the result is nothing but a
form factor. 
This provides us with the possibility for comparison.
To our great surprise,
the BBS and the BJMS fermions
are completely equivalent. 
In the multi-index notation (see \eqref{multiindex})
the statement is this:
\begin{align}
&\mathrm{If}\quad\quad\  
\mathcal{O}_\al=\betab^*_{I^+}\bar\betab^*_{\bar I^+}
\bar\gammab^*_{\bar I^-}\gammab^*_{I^-}
\Phi _\al\,,
\label{If}\\
&\mathrm{then}\quad L^{(\star)}_{\mathcal{O}_\al}
=\mub^{\frac 1 \nu(|I^-|+|I^+|+|\bar I^-|+|\bar I^+|)}
\psi^*_{I^+}\bar\psi^*_{\bar I^+}\bar\chi^*_{\bar I^-}
\chi^*_{I^-}M^{(\star)}_0\,,
\label{Then}
\end{align}
where $|I|$ denotes the sum of the entries
of the multi-index $I$.
To be precise we have to add some fermionic screening operators in \eqref{If} for $m\ne 0$, but in the infinite
volume they are irrelevant as explained in Section \ref{BJMSfermions}. For us the precise identification
\eqref{If}-\eqref{Then} came absolutely unexpected. It demonstrates a remarkable self-consistency of the sG model:
taking two complicated problems, that of computing the form factors of descendants and that of computing the
one-point functions on the cylinder (at finite temperature) and going to the very bottom of them we find the same
fermionic structure. 

The profit from the identification \eqref{If}-\eqref{Then} is twofold. 

First, since the BJMS descendants 
can be quantitatively  related to the Virasoro descendants,
the form factors of the latter can be computed.
There
is one technical obstacle here: up to now we were able to 
identify the BJMS and 
the
Virasoro descendants only
modulo the action of the local integrals of motion. 
This was quite sufficient for 
one-point functions,
but for form factors we would like to have the complete answer. 
The technical problem which one needs
to solve for this goal is explained at the end of  Section \ref{6vertex}. 

Second, in \cite{BBS} the null-vectors are found in terms of 
the BBS fermions. 
So, using
\eqref{If}-\eqref{Then} we can identify them for the
BJMS fermions. 
To keep the present paper within a reasonable size, we
leave the detailed study of the null-vectors to a separate publication. 
Here we shall consider only the
chiral null-vectors for the field $\Phi _{(2m-1)\frac {1-\nu}\nu}$,
$m\ge 1$ which corresponds in the
CFT language to $\Phi _{1,2m}$.
Let us give the simplest example
which  is the singular vector on level $2m$. This singular vector
is written as
\begin{align}\betab^*_1\gammab ^*_{2m-1}\Phi _{1,2m}\,.\label{example}\end{align}
At the moment comparison can be made only modulo 
the local integrals of motion, but even
with this simplification the Virasoro counterpart 
of \eqref{example}
looks really horrific. Using the formulae
of \cite{HGSIV, Boos} we find perfect agreement up to level $8$.
The fermionic formulae are also simple for other null vectors which are 
the descendants of the singular vector  in the Virasoro language.  
We think that this nice simplicity 
is another evidence of the universality of the fermionic picture.

The plan of the paper is as follows. In Section 2 we briefly 
review
the form factors axioms in application
to the sG model. 
In Section 3 we discuss certain integrals which 
play a basic role for the description of 
the sG form factors. 
In Section 4 we present the formulae for the form factors. 
Section 5 introduces the BBS fermions.
 In Section 6 we briefly discuss the BJMS
fermions in the sG case. 
The origin of the fermionic description 
is the 6-vertex model as explained in Section 7.
In Section 8 we discuss the infinite 
volume limit in the Matsubara direction. 
In Section 9 we present the
main technical achievement of this 
paper, namely the computation of the function $\omega (\z,\xi)$ for infinite volume limit
in the
Matsubara direction. 
The equivalence of 
the
BBS and  the BJMS fermions is established in Section 10. 
The null vectors for $\Phi _{1,2m}$ 
in terms of the
BBS fermions are discussed in Section 11.
Section 12 is devoted to comparison of these null vectors in 
the
fermionic and the
Virasoro descriptions.

\section{Form factors axioms}\label{axioms}

We are interested in form factors
of the exponential fields 
$\Phi _\al(0)$ \eqref{Phial}
and their descendants. 
In this paper 
we
define form factors to be 
matrix elements of a local operator taken
between excited states on the left
and the vacuum state on the right,
thus changing the convention of\cite{book} where the opposite matrix elements were mostly studied.
They
are simply related by the crossing symmetry. 
We shall use the notation $\mathcal{O}_\al$ for
descendants. These operators do not carry the topological charge,
so, their form factors
are non-trivial only in the case of an
equal number of solitons and anti-solitons.

We do not consider the breather
form factors since they can be obtained as residues of soliton ones. 
So, we have $2n$ particles  ($n$ solitons and $n$ anti-solitons)  with rapidities $\beta _1,
\cdots \beta _{2n}$. The form facotrs
$f_{\mathcal{O}_\al}(\beta _1,\cdots ,\beta _{2n})$ are meromorphic functions of these
rapidities. They are {vectors} from the zero weight subspace of the
space $(\mathbb{C}^2)^{\otimes 2n}$. The standard 
basis in the $j$-th tensor component is denoted as $e^{\pm}_j$.

The form factors are subject to three axioms. Formulating these
axioms we follow the conventions of \cite{book}, namely, if two rapidities
interchange we assume that the corresponding 
tensor components are permutated at the same time.
\vskip .2cm
\noindent
{\it Symmetry axiom.}
\begin{align}
&S_{j,j+1}(\beta _j-\beta _{j+1})f_{\mathcal{O}_\al}(\beta_1,\cdots,\beta _j,\beta _{j+1},\cdots ,\beta_{2n})
\label{axiom1}
\\&
=
f_{\mathcal{O}_\al}(\beta_1,\cdots,\beta _{j+1},\beta _j,\cdots ,\beta_{2n})\,,
\nn
\end{align}
where $S_{i,j}$ is the  soliton 
S-matrix \cite{zamS}. Its
explicit formula will be
given later \eqref{Smatrix}.
\noindent
{\it Riemann-Hilbert problem axiom.}
\begin{align}
&f_{\mathcal{O}_\al}(\beta_1,\cdots,\beta_{2n-1},\beta_{2n}+2\pi i)=e^{-\frac {\pi i \nu}
{1-\nu}\al\sigma ^3_{2n}}
f_{\mathcal{O}_\al}(\beta_{2n},\beta_1,\cdots,\cdots ,\beta_{2n-1})\,.\label{axiom2}
\end{align}
\vskip .2cm
\noindent
{\it Residue axiom.}
\begin{align}
&2\pi i\ \mathrm{res}_{\beta _{2n}=\beta _{2n-1}+\pi i}f_{\mathcal{O}_\al}(\beta_1,\cdots,\beta_{2n-2},
\beta_{2n-1},\beta_{2n})=
\label{axiom3}\\
&\Bigl(1-e^{-\frac {\pi i \nu}{1-\nu}\al\sigma ^3_{2n}}
S_{2n-1,1}(\beta _{2n-1}-\beta _{1})\cdots S_{2n-1,2n-2}(\beta _{2n-1}-\beta _{2n-2})
\Bigr)\nn\\
&\times f_{\mathcal{O}_\al}(\beta_1,\cdots,\beta_{2n-2})\otimes s _{2n-1,2n}\,,\nn
\end{align}
where $s_{i,j}=e^+_i\otimes e^-_j+e^-_i\otimes e^+_j$.

We change the standard basis of the tensor product 
to a new basis 
$w^{\epsilon_1,\cdots ,\epsilon_{2n}}(\beta _1,\cdots ,\beta _{2n})$
described in \cite{book}, and express the form factors as
\begin{align}
&f_{\mathcal{O}_\al}(\beta _1,\cdots ,\beta _{2n})
=
Z(\beta_1,\cdots,\beta _{2n})
\phantom{gggggggggggggggggg}\label{FF}
\end{align}
\begin{align*}
&\times
{\sum_{\epsilon_1,\cdots,\epsilon_{2n}=\pm}
w^{\epsilon_1,\cdots ,\epsilon_{2n}}}(\beta _1,\cdots ,\beta _{2n})
\frac{e^{\frac \nu {2(1-\nu)}\bigl(\sum\limits_{p\in I^-}\beta _{p}-\sum\limits_{p\in I^+}\beta_{p}+n\pi i   \bigr)}}{\prod
\limits_{p\in I^-,q\in I^+}\sinh\frac {\nu} {1-\nu}(\beta_{p}-\beta_{q})}
\ \cdot\ \mathcal{F} _{\mathcal{O}_\al,n}(\beta _{I^-}|\beta _{I+})\,,\nn 
\end{align*}
where $I^\pm=\{j\mid 1\le j\le 2n,\ \epsilon_j=\pm\}$, and the sum over 
$\epsilon_j$'s 
is  such that $\sharp(I^+)=\sharp(I^-)$. 
We have
introduced 
an overall multiplier
$$Z(\beta_1,\cdots,\beta _{2n})=\frac{c^n}{n!}\prod\limits_{i<j}\zeta (\beta _i-\beta _j)
\cdot e^{\frac {1-2\nu} {2(1-\nu)}n \sum\limits_{j=1}^{2n}\beta _j}\,.$$
The formula for the function $\zeta(\beta)$ can be found in \cite{book},
$$c=\frac {\nu}{2(1-\nu)\pi ^2\zeta(-\pi i)}\,.$$
We have set also
$$\beta _I=\{\beta _{i_1},\cdots , \beta _{i_n}\},\quad \mathrm{if}\quad
I=\{{i_1},\cdots , {i_n}\}\,.$$
We use this notation only for symmetric functions, so, the ordering of the
indices $i_1,\cdots,i_n$
is irrelevant. 
The main property of the new basis and of the
function $\zeta(\beta)$ is that
\begin{align}
&\zeta(\beta _{i}-\beta _{i+1})S_{i,i+1}(\beta _i-\beta _{i+1})
w^{\epsilon _1,\cdots,\epsilon_i,\epsilon _{i+1},\cdots ,\epsilon
 _{2n}}(\beta_1,\cdots,\beta_i,\beta _{i+1},\cdots ,\beta_{2n})
\nn\\
&=\zeta(\beta _{i+1}-\beta _{i})w^{\epsilon _1,\cdots,\epsilon _{i+1},\epsilon_i,\cdots ,\epsilon _{2n}}(\beta_1,\cdots,\beta_{i+1},\beta_i,
\cdots ,\beta_{2n})\,.\nn
\end{align}
Due to this property 
the first axiom is satisfied provided the essential part 
$\mathcal{F} _{\mathcal{O}_\al,n}(\beta _{I^-}|\beta_{I^+})$  
of the form factor 
is symmetric separately in $\beta_{I^-}$ and $\beta _{I^+}$.
It is well known
that for this function one can write an integral representation.
The integrals involved in this representation are quite remarkable, 
and we shall discuss them in detail
in the next section.

Before closing this section a remark is in order 
concerning the phase in the right hand side of \eqref{axiom2}. 
With solitons are associated some quasi-local fields, which interpolate 
in- and out-states when time goes to $\mp \infty$. 
The phase in \eqref{axiom2} specifies the locality property of a given 
field $\mathcal{O}_\al$ with respect to these interpolating fields.  
The BBS fermions which will be discussed in section \ref{BBSfermions}
act on a primary field $\Phi_\al$ and create fields which share the same
locality property (i.e. the same phase) as for $\Phi_\al$.

\section{Integrals}\label{integrals}

Consider the function
\begin{align}
&\chi(\sigma |\beta _1,\cdots ,\beta _{2n})=\prod\limits _{j=1}^{2n}\chi (\sigma ,\beta _j)\,,\nn\\
&\chi (\sigma,\beta )=\half
e^{-\frac  1{2(1-\nu)}(\sigma +\beta-\frac{\pi i} 2)}
\varphi (\sigma -\beta +
{\textstyle \frac{\pi i} 2})\,,\nn
\end{align}
where $\varphi (\sigma)$ is defined in \cite{book}. We shall not list
explicitly the properties of 
$\varphi(\sigma)$, since they 
can be read from 
those of the
function $\chi(\sigma |\beta _1,\cdots ,\beta _{2n})$
which we are going to give. In what follows we use the symbols
\begin{align*}
S=e^{\sigma},\ \ B_j=e^{\beta _j},\ \ Q=e^{\pi i \frac{1-\nu}\nu},\ \ A=e^{\pi i \al}\,\\
\ \mathfrak{s}=e^{\frac{2\nu}{1-\nu} \sigma},\ \ 
\mathfrak{b}_j=e^{\frac{2\nu}{1-\nu} \beta _j}, \ \ \mathfrak{q}=e^{\pi i \frac{1}{1-\nu}},
\ \ a=e^{\pi i\frac{\nu}{1-\nu}\al}\,.
\end{align*}

The function $\chi(\sigma |\beta _1,\cdots ,\beta _{2n})$ is a meromorphic function of
$\sigma$. For real $\beta_j$'s it does not have singularities for $0>\mathrm{Im}(\sigma)>-\pi$.
It has the following asymptotic
behaviour for $\sigma\to\pm \infty$:
\begin{align}
&\chi (\sigma |\beta_1,\cdots ,\beta _{2n})\simeq_{\sigma\to\infty}
e^{-2n\frac 1{1-\nu}\sigma} x^+(\mathfrak{s}|\mathfrak{b}_1,\cdots ,\mathfrak{b} _{2n})X^+(S|B_1,\cdots ,B _{2n})\,,\nn\\
&\chi (\sigma |\beta_1,\cdots ,\beta _{2n})\simeq_{\sigma\to-\infty}
x^-(\mathfrak{s}|\mathfrak{b}_1,\cdots ,\mathfrak{b} _{2n})X^-(S|B_1,\cdots ,B _{2n})\,,\nn
\end{align}
where
\begin{align}
&x^+(\mathfrak{s}|\mathfrak{b}_1,\cdots ,\mathfrak{b} _{2n})=1+\sum\limits_{k=1}^\infty x^+_k(\mathfrak{b}_1,\cdots ,\mathfrak{b} _{2n})
\mathfrak{s}^{-k}\,,\nn\\
&X^+(S|B_1,\cdots ,B_{2n})=
1+\sum\limits_{k=1}^\infty X^+_k(B_1,\cdots ,B _{2n})
S^{-k}\,,\nn\\
&x^-(\mathfrak{s}|\mathfrak{b}_1,\cdots ,\mathfrak{b} _{2n})=
\mathfrak{q}^n
\prod_{j=1}^{2n}
\mathfrak{b}_j^{-\frac 1 2}\Bigl(
1+\sum\limits_{k=1}^\infty x^-_k(\mathfrak{b}_1,\cdots ,\mathfrak{b} _{2n})
\mathfrak{s}^{k}\Bigr)\,,\nn\\
&X^-(S|B_1,\cdots ,B_{2n})=
\prod_{j=1}^{2n} B_j^{-1}
\Bigl(1+\sum\limits_{k=1}^\infty X^-_k(B_1,\cdots ,B _{2n})
S^{k}\Bigr)\,.\nn
\end{align}
The functions 
$x^+_k$, $X^+_k$ (resp. $x^-_k$, $X^-_k$)
are symmetric (resp. symmetric Laurent) polynomials of their arguments.
They can be inductively computed from the functional equations
\begin{align}
&\chi (\sigma +2\pi i )p(\mathfrak{s}\mathfrak{q}^4) = \chi(\sigma)p(\mathfrak{s}\mathfrak{q}^2)\label{2piiden}\\
&\chi (\sigma +\xinn\pi i)P(SQ) = \chi(\sigma)P(-S)\,,\label{xiiden}
\end{align}
where
\begin{align}
P(S)=\prod\limits _{j=1}^{2n}(S-B_j)\,,\quad p(\mathfrak{s})=
\prod\limits _{j=1}^{2n}(\mathfrak{s}-\mathfrak{b}_j)\,.\label{Pp}
\end{align}
The normalisation of $\chi (\sigma,\beta)$ is defined by the
identity
\begin{align}\chi(\sigma,\beta)\chi (\sigma,\beta +\pi i)=\frac 1{(\mathfrak{s}-\mathfrak{b})
(S^2-B^2)}\,,\label{normchi}\end{align}
which is convenient  for computing the residues \eqref{axiom3}.

Consider now the integral
\begin{align}
I_\al(\beta _1,\cdots,\beta_{2n})=\int\limits _{\mathbb{R}-i0}
\chi (\sigma |\beta_1,\cdots ,\beta _{2n})
e^{\frac{\nu\al}{1-\nu}\s}d\sigma\,.
\label{integral}
\end{align}
From the asymptotic behaviour of 
$\chi(\sigma|\beta_1,\cdots ,\beta _{2n})$,
it is clear that this integral converges for
\begin{align}
0<\mathrm{Re}(\al)<\frac{2n}{\nu}\,.\label{r1}
\end{align}
In what follows we shall often omit the dependence 
of these functions
on 
$\beta_1,\ldots,\beta_{2n}$,
abbreviating them
to $I_\al$, $\chi (\sigma)$, $x^\pm(\mathfrak{s})$ and $X^\pm(S)$. 

We want to continue the integral $I_\al$ to the entire complex plane of $\al$. 
To this end for any $k\in \mathbb{Z}$
let us introduce Laurent polynomials of the form
\begin{align*}
m^{(k)}(\mathfrak{s})=\sum_{j=-k+1}^{2n-k}m^{(k)}_j\mathfrak{s}^j,\quad
n^{(k)}(\mathfrak{s})=
\begin{cases}
\sum_{j=-k+1}^0n^{(k)}_j\mathfrak{s}^j&\text{ if $k\geq1$;}\\[3pt]
\sum_{j=0}^{-k+1}n^{(k)}_j\mathfrak{s}^j&\text{ if $k\leq0$,}
\end{cases}
\end{align*}
which satisfy the identity
\begin{align}
p(\mathfrak{s}\mathfrak{q}^{-2})=m^{(k)}
(\mathfrak{s})+a^{-2}p(\mathfrak{s})n^{(k)}
(\mathfrak{s}\mathfrak{q}^{-4})-
p(\mathfrak{s}\mathfrak{q}^{-2})n^{(k)}
(\mathfrak{s})\,,\label{mandn}
\end{align}
where $p(\mathfrak{s})$ is defined in \eqref{Pp}.
For example
\begin{align}
&n^{(1)}(\mathfrak{s})=\frac 1 {a^{-2}\mathfrak{q}^{4n}-1}\,,\label{n1}\\
&n^{(2)}(\mathfrak{s})=\frac 1 {a^{-2}\mathfrak{q}^{4n}-1}
\(1+\frac{\mathfrak{q}^{4n}(1-\mathfrak{q}^2)}
{(a^{-2}\mathfrak{q}^{4n+4}-1)}a^{-2}\sigma _1(\mathfrak{b}_1,\cdots,\mathfrak{b}_{2n})
\ \mathfrak{s}^{-1}\)\,.\label{n2}
\end{align}
Here and after,  
$\sigma_l(\mathfrak{b}_1,\cdots\mathfrak{b}_{2n})$ stands for the $l$-th
elementary symmetric polynomial in $\mathfrak{b}_j$'s.
Note that
\begin{align}
{\rm res}_{a^2=q^{4n}}n^{(1)}(\mathfrak{s})\frac{da^2}{a^2}&=-1,\label{RESIDUES}\\
{\rm res}_{a^2=q^{4n+4}}n^{(2)}(\mathfrak{s})\frac{da^2}{a^2}&=-x^+_1\mathfrak{s}^{-1}.
\label{RESIDUES2}
\end{align}
Using the functional equation \eqref{2piiden} we can transform the integral defined in
the original region \eqref{r1} to
\begin{align}
&I_\al
=\int\limits _{\mathbb{R}-i0}
\chi (\sigma
)e^{\frac{\nu\al}{1-\nu}\sigma}
\frac{m^{(1)}(\mathfrak{s})}{p(\mathfrak{s}\mathfrak{q}^{-2})}d\sigma+
\int\limits _{\Gamma}
\chi (\sigma 
)e^{\frac{\nu\al}{1-\nu}\sigma}n^{(1)}(\mathfrak{s})d\sigma\,,\label{I1}
\end{align}
where the contour $\Gamma$ contains all the poles of the integrand
in the strip $0>\mathrm{Im}(\sigma)\geq-2\pi$. It
naturally splits into a sum of contours $\Gamma _j$ as is shown
on {\it fig. 1}.
\vskip .5cm
\hskip 1.5cm\includegraphics[height=8cm]{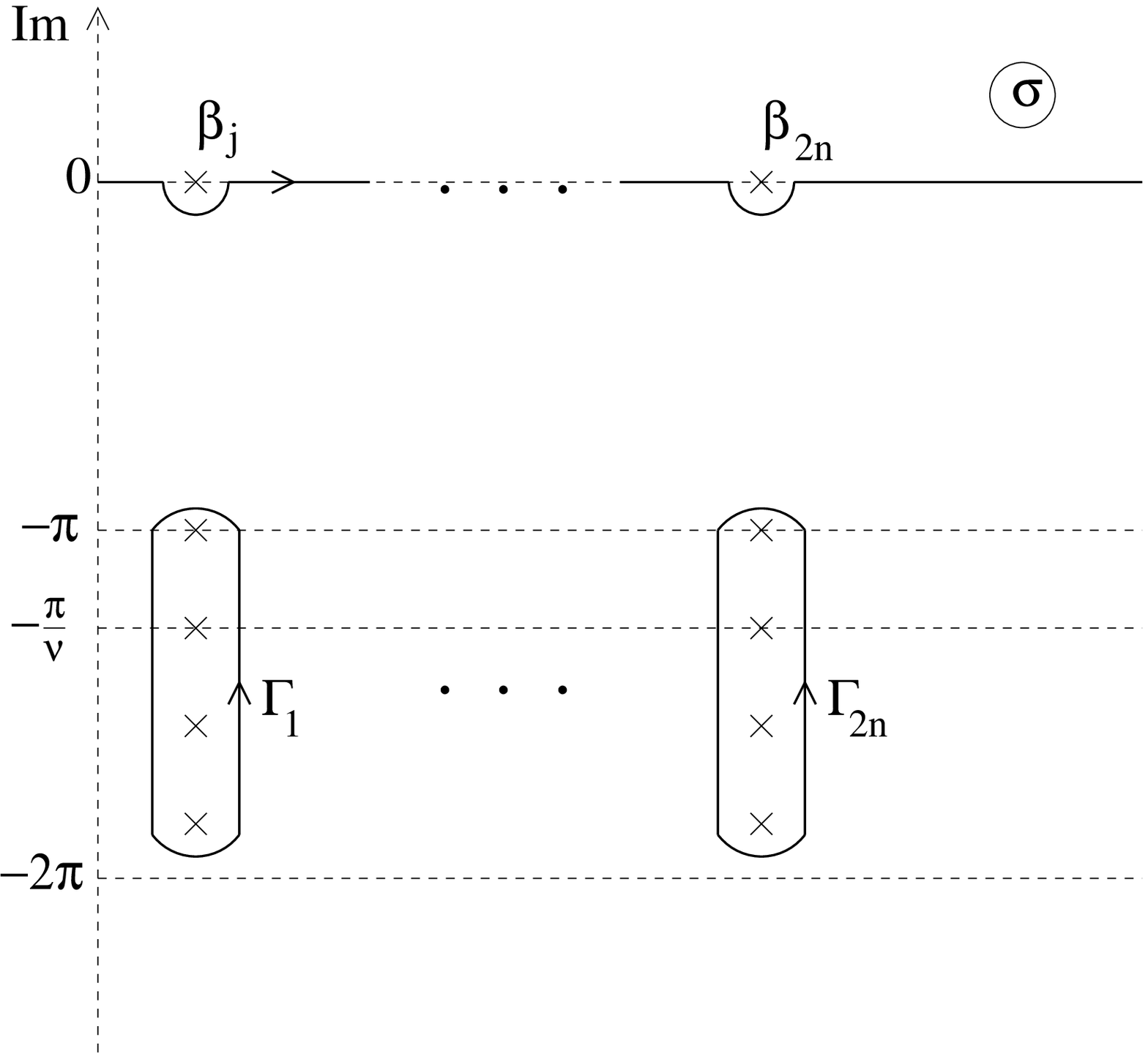}

\noindent
{\it Fig.1: Contours for analytic continuation of $I_\al$. \\
The original contour $\Gamma$ is split into 
a sum of the contours $\Gamma_j$. Each $\Gamma_j$ surrounds a series of
poles between the strip $-\pi \ge \mathrm{Im}(\sigma)>-2\pi$
with the same real part as $\beta_j$. 
}
\vskip .3cm

The equivalence with
the original definition is established by transforming $\Gamma$ to
$(-\mathbb{R}-i0)\cup
(\mathbb{R}-2\pi i -i0)$. But the right hand side of
\eqref{I1} is analytical in $0<\mathrm{Re}(\al)<
\frac{2n}{ \nu}+2$. So, we have managed to
continue $I_\al$
analytically to  the strip $\frac{2n}{ \nu}\le \mathrm{Re}(\al)<\frac{2n}{ \nu}+2$.
From the formula \eqref{n1} we see that the result of the analytical
continuation has simple poles at the points $2n+(2 n+l)\xin$ 
with $l\ge 0$.
We shall compute the residues at these poles later. 

Let us proceed. Another representation for $I_\al$ is
possible for $2<\mathrm{Re}(\al)<\frac{2n}\nu+2$:
\begin{align}
&I_\al
=\int\limits _{\mathbb{R}-i0}
\chi (\sigma 
)
e^{\frac{\nu\al}{1-\nu}\s}
\frac{m^{(2)}(\mathfrak{s})}{p(\mathfrak{s}\mathfrak{q}^{-2})}d\sigma+
\int\limits _{\Gamma}
\chi (\sigma
)e^{\frac{\nu\al}{1-\nu}\s}n^{(2)}(\mathfrak{s})
d\sigma\,.\label{I2}
\end{align}
Indeed, the difference between \eqref{I2} and \eqref{I1} is
\begin{align}
&\int\limits _{\mathbb{R}-i0}
\chi (\sigma
 )e^{\frac{\nu\al}{1-\nu}\sigma}
\frac
{m^{(2)}(\mathfrak{s})-m^{(1)}(\mathfrak{s})}
{p(\mathfrak{s}\mathfrak{q}^{-2})}d\sigma
+
\int\limits _{\Gamma}
\chi (\sigma
 )e^{\frac{\nu\al}{1-\nu}\sigma}
\(n^{(2)}(\mathfrak{s})-n^{(1)}(\mathfrak{s})\)d\sigma\,.\label{zero}
\end{align}
For $2<\mathrm{Re}(\al)<\frac{2n}\nu+2$ the integrals are well 
defined. Moreover, by deforming the contour 
as before, it is easy to show that \eqref{zero} is equal to zero.
The right hand side of \eqref{I2} is well defined in $2<\mathrm{Re}(\al)<\frac{2n}\nu+4$.
So, we have continued $I_\al$
into the strip $\frac{2n}\nu+2\le\mathrm{Re}(\al)<\frac{2n}\nu+4$. The series of poles at 
$2n+(2 n+l)\xin$ continues, and new poles at $2(n+1)+2(n + l)\xin$ 
with $l\ge 0$ appear. 

It is clear now how to continue $I_\al$
further. It is equally clear how to continue it
to non-positive $\mathrm{Re}(\al)$. The final result is that $I_\al
$ is a meromorphic function in $\mathbb{C}$ with simple poles at
\begin{align}
\al=2(n+m)+(2n+l)\xinn , \quad l,m\ge 0;\qquad
\al =-2m-\xinn l ,\quad l,m\ge 0\,.\nn
\end{align}

Let us compute the residues. Consider $I_{\al}$
in the strip $0<\mathrm{Re}(\al)<2+2n\frac 1 \nu$.
In this strip we have poles at
the points $\al=2n+(2n+l)\xin$, $l\ge 0$, because at these points
$m^{(1)}(\mathfrak{s})$ and $n^{(1)}(\mathfrak{s})$ have singularities.
Obviously,
$$\res_{\al =2n+(2n+l)\xin}\Bigl(m^{(1)}(\mathfrak{s})+a^{-2}
p(\mathfrak{s})
n^{(1)}(\mathfrak{s}\mathfrak{q}^{-4})
-p(\mathfrak{s}\mathfrak{q}^{-2})
n^{(1)}(\mathfrak{s})\Bigr)=0\,.$$
Then, using this equality with \eqref{RESIDUES} in \eqref{I1} and
transforming the contour $\Gamma$ to 
$-\bigl(-(-\infty,\Lambda)\bigr)\cup (\Lambda -2\pi i,
\Lambda)\cup (-\infty-2\pi i,\Lambda -2\pi i)$ 
we find
\begin{align*}
\res_{\al =2n+(2n+l)\xin}I_{\al}\frac{da^2}{a^2}
&=-\lim_{\Lambda\to \infty}\int\limits _{\Lambda-2\pi i}^{\Lambda}
\chi (\sigma)e^{2n(1+\frac \nu{1-\nu})\sigma} S^ld\sigma
=-2\pi iX^+_l\,.
\end{align*}
Next, consider $I_{\al}$ in the strip $2<\mathrm{Re}(\al)<4+2n\frac 1 \nu$.
Using \eqref{I2}, we compute the residues at the
two series 
$\al =2n+(2n+l)\xin$ and $\al =2n+2+(2n+l)\xin$
of poles of $n^{(2)}(\mathfrak{s})$ given by \eqref{n2}.
For the former the result is the same because the term 
containing $\mathfrak s$ disappears in the limit
$\Lambda\rightarrow\infty$. For the latter using \eqref{RESIDUES2} we obtain
\begin{align*}
\res_{\al =2n+2+(2n+l)\xin}I_{\al}\frac{da^2}{a^2}&=-2\pi iX^+_lx^+_1\,.
\end{align*}
Continuing along the same lines we come to
the following result.
\begin{prop}
The residues of $I_{\al}
$ at $\al=2n+2m+(2n+l)\xin$ for
$\l,m\ge 0$ are given by
\begin{align}
&
\mathop{\res}_{\al=2(n+m)+(2n+l)\xin}
I_{\al}d\al
=-{\textstyle\frac {1-\nu}{\nu}}
\mathop{\res}_{S=\infty}\(S^l X^+(S)\frac{dS}S\)
\mathop{\res}_{\mathfrak{s}=\infty}\(
\mathfrak{s}^m x^+(\mathfrak{s})\frac{d\mathfrak{s}}{\mathfrak{s}}\)\,.\label{reslm+}
\end{align}
Similarly, the residues of $I_{\al}
$ at $\al=-2m-l\xin$ for
$\l,m\ge 0$ are given by
\begin{align}
&
\mathop{\res}_{\al=-2m-l\xin}
I_{\al}d\al
={\textstyle\frac {1-\nu}{\nu}} 
\mathop{\res}_{S=0}\(S^{-l} X^-(S)\frac{dS}S\)
\mathop{\res}_{\mathfrak{s}=0}\(
\mathfrak{s}^{-m} x^-(\mathfrak{s})\frac{d\mathfrak{s}}{\mathfrak{s}}\)\,.\label{reslm-}
\end{align}
\end{prop}
Now we give the main definition.

\begin{definition}
Consider two Laurent polynomials $\ell(\mathfrak{s})$ and $L(S)$. 
We define their pairing
$(\ell,L)_\al $ 
by the following two requirements:
\newline 1. The pairing is bilinear.
\newline 2. If $\ell(\mathfrak{s})=\mathfrak{s}^m$, $L(S)=S^l$ then
\begin{align}
(\ell,L)_\al=I_{\al+2m+\xin l}\,.\label{MAINDEF}
\end{align}
\end{definition}
The polynomials $\ell$ and $L$ are interpreted 
as cycles and forms in the classical limit. 
Actually there are two possibilities.
In the limit $\nu\to 1$, $\ell$ 
describes cycles and $L$ describes one-forms. 
In the limit $\nu\to 0$, $\ell$
describes one-forms and $L$ describes cycles.
The situation is very much dual, 
and there is no preferred choice to call one a cycle and the other a
form. So we will call $\ell$ a 
$\mathfrak{q}$-deformed form and $L$
a $Q$-deformed form.
The next two propositions
describe exact forms.
\begin{prop}
For any Laurent polynomial $z(\mathfrak{s})$ we  define 
\begin{align}
D_a[z](\mathfrak{s})=a^{-2}p(\mathfrak{s})z(\mathfrak{s})-
p(\mathfrak{s}\mathfrak{q}^2)
z(\mathfrak{s}\mathfrak{q}^{4})\,,
\label{qexact}
\end{align}
and call it a $\mathfrak{q}$-exact form.
For any Laurent polynomial $L(S)$ we have
\begin{align}
(D_a[z],L)_\al =0\,.\label{qexact=0}
\end{align}
\end{prop}
\begin{proof}
It is sufficient to consider the case $L(S)=S^l$. 
The procedure described above gives for any $\ell(\mathfrak{s})$
\begin{align}
(\ell,L)_\alpha=
\int\limits _{\mathbb{R}-i0}
\chi (\sigma 
)e^{(\frac {\al\nu}{1-\nu}+l)\sigma}
\frac{m(\mathfrak{s})}{p(\mathfrak{s}\mathfrak{q}^{-2})}d\sigma+
\int\limits _{\Gamma}
\chi (\sigma 
)n(\mathfrak{s})e^{(\frac {\al\nu}{1-\nu}+l)\sigma}d\sigma\,,\label{propq}
\end{align}
where
\begin{align}
p(\mathfrak{s}\mathfrak{q}^{-2})\ell(\mathfrak{s})=m
(\mathfrak{s})+a^{-2}p(\mathfrak{s})n
(\mathfrak{s}\mathfrak{q}^{-4})-
p(\mathfrak{s}\mathfrak{q}^{-2})n
(\mathfrak{s})\,,\label{lmnpropq}
\end{align}
and $m(\mathfrak{s})$ is chosen in such a way that the integral
in \eqref{propq} converges. 
This is possible to do with a comfortable margin.
If we take $\ell=D_a[z]$ 
it is easy to see that \eqref{lmnpropq} is satisfied with
\begin{align}
m(\mathfrak{s})=0, \quad n(\mathfrak{s})=p(\mathfrak{s}
\mathfrak{q}^2)z(\mathfrak{s}\mathfrak{q}^{4})\,.\label{AMBIGUITY}
\end{align}
Then the right hand side of \eqref{propq} vanishes since $p(\mathfrak{s}
\mathfrak{q}^2)$ cancels the singularities of $\chi (\sigma 
)$
inside $\Gamma$.
\end{proof}
\begin{prop}\label{propQexact}
For any Laurent polynomial $Z(S)$ we  define 
\begin{align}
D_A[Z](S)=Z(S)P(S)-AZ(SQ)P(-S)\label{Qexact}
\end{align}
and call it a $Q$-exact form.
For any Laurent polynomial $\ell(\mathfrak{s})$ we have
\begin{align}
(\ell,D_A[Z])_\al =0\,.\label{Qexact=0}
\end{align}
\end{prop}
\begin{proof}
It would be sufficient to consider a simple case when
there is no need for regularisation of the integral, 
and then continue analytically.
However, 
in what follows it is
more instructive to have a direct proof for any $\al$.
Figure 2
illustrates the proof.
\vskip .5cm
\hskip 5cm\includegraphics[height=8cm]{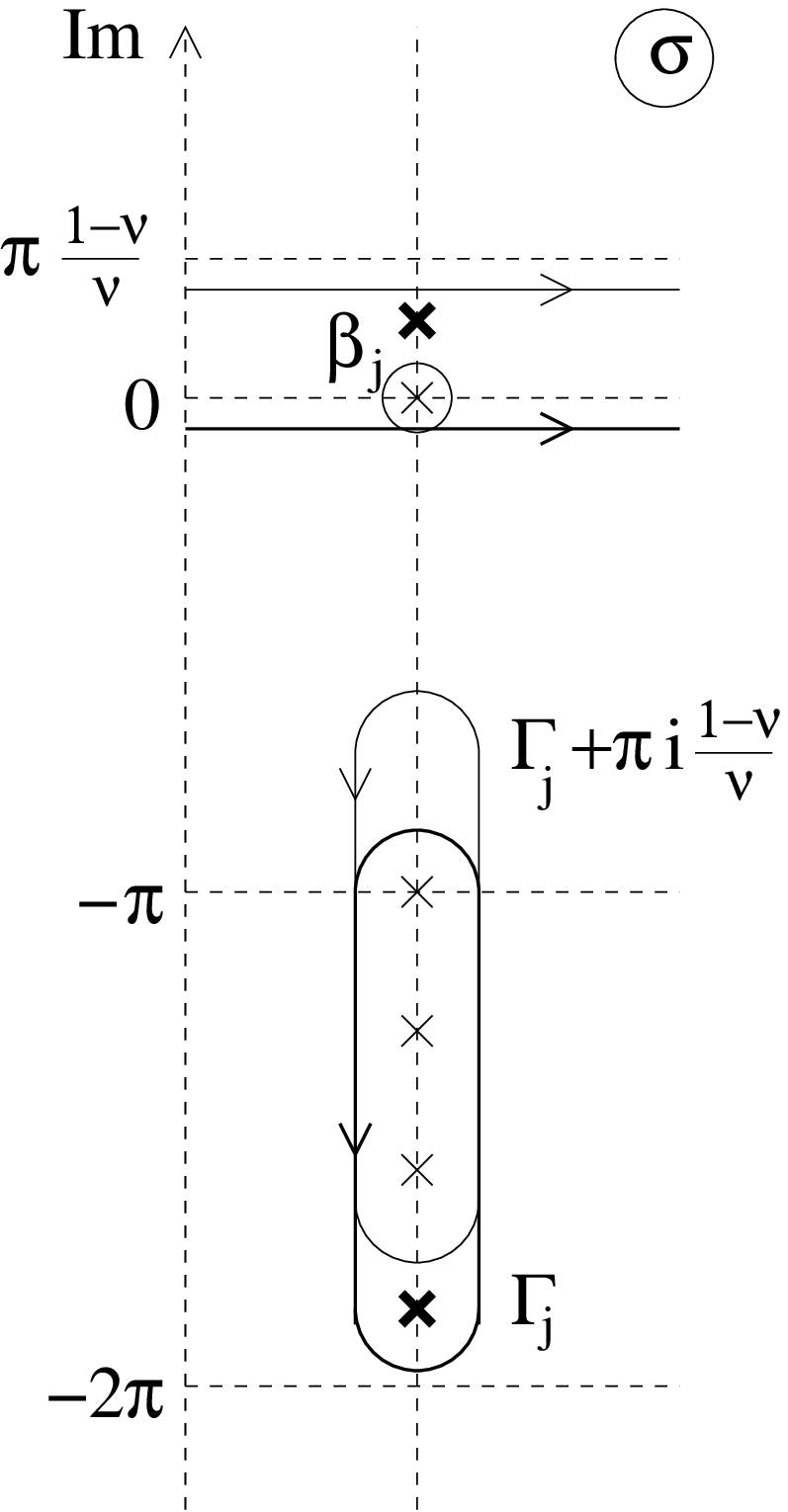}

\noindent
{\it Fig.2: Possible poles arising in the proof of Proposition \ref{propQexact}
\\
Poles canceled by the factor $P(S)$ are represented by ordinary crosses
inside circles. The residues at the 
remaining poles depicted by boldface crosses 
cancel each other.
}
\vskip .3cm
From the definition we have
\begin{align}
(\ell,D_A[Z])_\al=
\int\limits _{\mathbb{R}-i0}
\chi (\sigma)
e^{\frac{\nu\al}{1-\nu}\s}D_A[Z](S)
\frac{m(\mathfrak{s})}{p(\mathfrak{s}\mathfrak{q}^{-2})}d\sigma+
\int\limits _{\Gamma}
\chi (\sigma)
e^{\frac{\nu\al}{1-\nu}\s}D_A[Z](S)n(\mathfrak{s})d\sigma\,,\label{propQ}
\end{align}
where
$m(\mathfrak{s})$ and $n(\mathfrak{s})$
are polynomials satisfying \eqref{lmnpropq}
and are
chosen by the requirement of convergence. Actually they
may be different for different monomials in $Z(S)$, but this does not
matter for the following computation. 

Using the functional equation  \eqref{xiiden} one easily finds that
\begin{align*}
&(\ell,D_A[Z])_\al=
\(\int\limits _{\mathbb{R}-i0}-\int\limits _{\mathbb{R}+\pi i\xin-i0}\)
\chi (\sigma )e^{\frac{\nu\al}{1-\nu}\sigma}P(S)Z(S)
\frac{m(\mathfrak{s})}{p(\mathfrak{s}\mathfrak{q}^{-2})}d\sigma\\&
\quad\quad\quad\quad+
\(\int\limits _{\Gamma}-\int\limits _{\Gamma+\pi i\xin}\)
\chi (\sigma)
e^{\frac{\nu\al}{1-\nu}\sigma}P(S)Z(S)n(\mathfrak{s})d\sigma\,.
\end{align*}
Poles arise from $\chi(\s)$ and $1/p(\mathfrak{sq}^{-2})$.
We ask which poles are inside the difference of contours
$(\mathbb{R}-i0)-(\mathbb{R}+\pi i\xin-i0)$, or $\Gamma-(\Gamma+\pi i\xin)$. 
The poles of the first integrand at $\sigma =\beta _j$ are canceled by $P(S)$.
So, the only remaining poles between $\mathbb{R}$ and $\mathbb{R}+\pi i\xin$ are
at $\sigma =\beta _j+\pi i\bigl(1 -\xin \bigl[\frac \nu {1-\nu}\bigr]\bigr)$. Similarly,
the poles inside $\Gamma -(\Gamma +i\xin)$ are situated 
only at the points $\sigma =\beta _j-\pi i\bigl(1 +\xin \bigl[\frac \nu {1-\nu}\bigr]\bigr)$.
On {\it fig. 2} the pole canceled by $P(S)$ is in the circle, and the
remaining poles are fat. 
It follows from the functional equation \eqref{2piiden}
and the definition \eqref{lmnpropq}
that the residues at these remaining poles cancel each other.

\end{proof}



The last definition which we would like to give here concerns the 
multiple integrals. Consider $k$ Laurent
polynomials $\ell_1(\mathfrak{s}),\cdots \ell_k(\mathfrak{s})$ and
$k$ Laurent polynomials $L_1(S),\cdots ,L_k(S)$, and define the antisymmetric 
Laurent polynomials
of $k$ variables:
$$\ell^{(k)}(\mathfrak{s}_1,\cdots \mathfrak{s}_k)=\bigl(\ell_1\wedge\cdots\wedge \ell_k
\bigr)(\mathfrak{s}_1,\cdots \mathfrak{s}_k),
\quad L^{(k)}(S_1,\cdots ,S_k)=\bigl(L_1\wedge\cdots \wedge L_k)(S_1,\cdots ,S_k)\,.
$$
We define
\begin{align}
(\ell^{(k)},L^{(k)})_\al=
\det\left(\, (\ell_i,L_j)_\al\, \right)_{i,j=1,\cdots ,k}\,.\label{kPairing}
\end{align}
Then this pairing is generalised for all 
antisymmetric Laurent polynomials of $k$
variables $\ell^{(k)}$, $L^{(k)}$ by bilinearity.

\section{Formulae for form factors} \label{FormulasforFF}

Now we are ready to write down the 
formulae for form factors \eqref{FF}.  
Let us first prepare a couple of symbols. 

For a partition $I^-\sqcup I^+=\{1,\cdots,2n\}$
such that $\sharp(I^-)=\sharp(I^+)$, 
define the polynomials
$$
p_{I^-}(\mathfrak{s})=\prod\limits _{j\in I^-}(\mathfrak{s}-\mathfrak{b}_j),
\quad p_{I^+}(\mathfrak{s})=\prod\limits _{j\in
I^+}(\mathfrak{s}-\mathfrak{b}_j)\,,
$$
so that we have $p(\mathfrak{s})=p_{I^+}(\mathfrak{s})p_{I^-}(\mathfrak{s})$.
We set also
$$
p_{I^\pm,i}(\mathfrak{s})=
\left[ \mathfrak{s}^{i-n}  p_{I^\pm}(\mathfrak{s}) \right]_\ge\,,
$$
where $[~]_{\ge}$ signifies the polynomial part.
Essentially following \cite{book} we define 
\begin{align}
&\ell^{(n)}_{I^-\sqcup I^+}(\mathfrak{s}_1,\cdots ,\mathfrak{s}_n)
=\(\ell_{I^-\sqcup I^+,0}
\wedge \cdots \wedge \ell_{I^-\sqcup I^+,n-1}\)
(\mathfrak{s}_1,\cdots ,\mathfrak{s}_n)\,,\label{lli}
\\
&\ell_{I^-\sqcup I^+,i}(\mathfrak{s})
=a^{-1}\bigl\{p_{I^-}(\mathfrak{s})
\(p_{I^+,i}(\mathfrak{s})-p_{I^+,i}(\mathfrak{s}\mathfrak{q}^2)\)
\label{li}
\\
&\qquad+\mathfrak{q}^{2(i-n)}p_{I^+}(\mathfrak{s}\mathfrak{q}^2)\(p_{I^-,i}(\mathfrak{s})
-a^2p_{I^-,i}(\mathfrak{s}\mathfrak{q}^2)\)\bigr\}\,.
\nn
\end{align}
Formula \eqref{li} can also be rewritten as 
\begin{align}
&c_{I^-\sqcup I^+}(t,s):=\sum_{i=0}^{n-1}(\mathfrak{q}^2\mathfrak{t})^{n-i}
\ell_{I^-\sqcup I^+,i}(\mathfrak{s})
\label{GENERATING}\\
&=
\frac{a\mathfrak{t}}{\mathfrak{t}-\mathfrak{q}^2\mathfrak{s}}p(
\mathfrak{q}^2\mathfrak{s})
-\frac{a^{-1}\mathfrak{q}^2\mathfrak{t}}
{\mathfrak{q}^2\mathfrak{t}-\mathfrak{s}}p(\mathfrak{s})\nn\\
&+p_{I^+}(\mathfrak{q}^2\mathfrak{t})p_{I^-}(\mathfrak{s})\(
\frac{a^{-1}\mathfrak{q}^2\mathfrak{t}}{
\mathfrak{q}^2\mathfrak{t}-s}-\frac{a^{-1}\mathfrak{t}}{\mathfrak{t}-\mathfrak{s}}\)+
p_{I^-}(\mathfrak{t})p_{I^+}(\mathfrak{q}^2\mathfrak{s})\(\frac{a^{-1}\mathfrak{t}}{\mathfrak{t}-\mathfrak{s}}-\frac{a\mathfrak{t}}{\mathfrak{t}-
\mathfrak{q}^2\mathfrak{s}}\)\,.\nn
\end{align}

Now consider the ansatz \eqref{FF}.  
As it has  already been mentioned,
the symmetry axiom \eqref{axiom1} is satisfied
automatically if 
$\mathcal{F}_{\mathcal{O}_\al,n}(\beta _{I^-}|\beta _{I^+})$
is symmetric separately in $\beta _{I^-}$ 
and in $\beta _{I^+}$.
Furthermore, the Riemann-Hilbert problem axiom \eqref{axiom2} is satisfied 
if we set 
\begin{align}
\mathcal{F} _{\mathcal{O}_\al,n}(\beta _{I^-}|\beta _{I^+})=e^{\frac{\nu}{2(1-\nu)}(1-\al)(\sum\limits_{j=1}^{2n}\beta _j-
\pi i n)}\cdot
(\ell^{(n)}_{I^-\sqcup I^+}\ ,\ L^{(n)})_\al\,,
\label{Phialpha}
\end{align}
where $L^{(n)}=L^{(n)}(S_1,\cdots,S_n|B_1,\cdots,B_{2n})$ is 
an arbitrary Laurent polynomial 
which is anti-symmetric in $S_i$'s and symmetric in $B_j$'s. 
This statement is proved by a direct computation 
similar to the one given in \cite{book}.

While these two axioms concern
form factors with a fixed number of
particles, the third, the residue axiom, relates the Laurent polynomials 
$L^{(n)}$ with different $n$. 
Using the procedure of computing the residues familiar from \cite{book},  
one can reduce the residue axiom further to a simple 
system of recurrence relations given below (see \eqref{rec}).
For later use let us formulate it in 
a slightly more general setting. 

Let $c$ be an integer, and consider a sequence of Laurent polynomials
\begin{align*}
L^{(\star)}=
\{L^{(l,n)}(S_1,\cdots,S_l|B_1,\cdots,B_{2n})\}_{l,n\ge0\atop l-n=c}
\end{align*}
which are anti-symmetric in $S_i$'s and symmetric in $B_j$'s. 
We refer to $L^{(l,n)}$ as the $2n$-particle component of
$L^{(\star)}$, and $c$ as the charge.

\begin{definition}
We say that $L^{(\star)}$ is a tower of charge $c$ if 
\begin{align}
&L^{(l,n)}(S_1,\cdots,S_{l-1},B|
B_1,\cdots ,B_{2n-2},B,-B)
\label{rec}\\
&=
(-1)^cB\prod\limits_{p=1}^{l-1}(B^2-S^2_p)\cdot
L^{(l-1,n-1)}(S_1,\cdots, S_{l-1}|
B_1,\cdots ,B_{2n-2})\,
\nn
\end{align} 
holds for all $l,n\ge1$ with $l-n=c$. 
\end{definition}
In the case $c=0$, we also write $L^{(n,n)}$ as $L^{(n)}$.

The statement is, the residue axiom is satisfied if 
the sequence $\{L^{(n)}\}_{n=0}^\infty$ in 
\eqref{Phialpha}
is a tower of charge $0$. 
In other words, a tower of charge $0$, 
or more precisely its residue class modulo $Q$-exact forms,   
represents a descendant field $\mathcal{O}_\alpha$. 

The most basic example of a tower of charge $0$ is given by 
the polynomials which do not depend on the parameters $B_1,\cdots , B_{2n}$:
\begin{align}
M^{(n)}_{0}(S_1,\cdots ,S_n)=\langle\Phi _\al\rangle\cdot S\wedge S^3\wedge \cdots\wedge S^{2n-1}\,.
\label{FERMIZONE}
\end{align}
This tower, denoted by $M^{(\star)}_0$, 
is identified with the primary field $\Phi _\al$, whose vacuum expectation value 
$\langle\Phi _\al\rangle$ has been found in \cite{LukZam}.

\section{BBS fermions.}\label{BBSfermions}

We are now in a position to introduce fermions which create
towers out of the tower $M_0^{(\star)}$ for the primary field. 
In all formulas discussed below, 
we shall suppress the dependence on the parameters $B_j$. 
Since we deal with towers, however, 
their number $2n$ can vary from place to place. 
For that reason, in this section alone,  
we indicate the $n$-dependence by a suffix, e.g.,  
$P_n(S)=\prod_{j=1}^{2n}(S-B_j)$.

First we note that on the space of all towers there is an obvious 
action of the local integrals of motion. 
This is because the recursion relation \eqref{rec} 
is unaffected by the multiplication 
$L^{(l,n)}\mapsto f(I,\bar{I})L^{(l,n)}$,
where $f(I,\bar{I})$ is an arbitrary polynomial in 
\begin{align}
I_{2j-1,n}=\sum_{k=1}^{2n} B_k^{2j-1},
\quad 
\bar{I}_{2j-1,n}=\sum_{k=1}^{2n} B_k^{-(2j-1)}\,.
\label{integralsofm}
\end{align}
In the construction of towers 
we shall fully make use of the freedom \eqref{integralsofm}. 
We remark that the eigenvalues \eqref{integralsofm} arise in the expansion 
\begin{align}
\sqrt{\frac{P_n(-Z)}{P_n(Z)}}=
\begin{cases}
e^{X_n(Z)} & (Z\to\infty), \\
e^{\bar{X}_n(Z)} & (Z\to0), \\
\end{cases}
\label{sqrt-P}
\end{align}
where 
\begin{align}
X_n(Z)=\sum\limits _{j=1}^{\infty}\frac 1 {2j-1}Z^{-2j+1}I_{2j-1,n}\,,
\quad
\overline{X}_n(Z)=
\sum\limits _{j=1}^{\infty}\frac 1 {2j-1}Z^{2j-1}\bar{I}_{2j-1,n}\,.
\label{X-barX}
\end{align}

Now, following \cite{BBS}, let us introduce the polynomial
\begin{align}
C_n(S_1,S_2)=\frac 1 {4\nu} S_1
\sum_{\epsilon_1,\epsilon_2=\pm}\frac{P_n(\epsilon_1 S_1)
P_n(\epsilon_2 S_2)}{\epsilon_1 S_1+\epsilon_2 S_2}\,,
\label{defC}
\end{align}
where the overall coefficient is introduced for future convenience.
This polynomial is characterised by the following requirements:
\begin{align}
&\text{$C_n(S_1,S_2)$ is odd in $S_1$ and even in $S_2$},
\label{C-1}\\
&\mathrm{deg}_{S_1}C_n(S_1,S_2) = 2n-1,
\quad  \mathrm{deg}_{S_2}C_n(S_1,S_2)= 2n-2\,,
\label{C-2}\\
&C_n(S_1,S_2)|_{B_{2n-1}=B,B_{2n}=-B}=(S_1^2-B^2)(S_2^2-B^2)C_{n-1}(S_1,S_2),
\label{C-3}
\end{align}
with $C_1(S_1,S_2)=-\frac 1 \nu(B_1+B_2)B_1B_2 S_1$. We note also that the total homogeneous degree of
$C_n(S_1,S_2)$ in $S_i$'s and $B_j$'s is $4n$.

Due to the definition through the pairing 
\eqref{Phialpha}, 
adding $Q$-exact forms to components of a tower
$L^{(l,n)}$ does not change the form factors. 
A simple way of fixing this freedom is to restrict their degrees by
\begin{align}
0\le \mathrm{deg}_{S_i}L^{(l,n)}(S_1,\cdots,S_l)\le 2n-1\,.
\label{deg-rest}
\end{align}
If $\alpha$ is generic, any polynomial 
can be brought to this form by adding $Q$-exact forms. 

Let $\mathcal{T}_{0,c}$ denote the space of all towers of charge $c$
satisfying \eqref{deg-rest}. 
We introduce the action of fermions
\begin{align*}
\psi^*_0(Z):\mathcal{T}_{0,c}\longrightarrow \mathcal{T}_{0,c+1},
\quad 
\chi^*_0(Z):\mathcal{T}_{0,c}\longrightarrow \mathcal{T}_{0,c-1}\,,
\end{align*}
by defining
\begin{align}
&\bigl(\psi_0^*(Z)L^{(\star)}\bigr)^{(l+1,n)}(S_0,\cdots,S_l)
\label{psi0}\\
&\quad=
\frac{1}{P_n(-Z)}\,
\frac{1}{l!}\,\text{Skew}_{S_0,\cdots,S_l}C_n(Z,S_0)L^{(l,n)}(S_1,\cdots,S_l)\,,
\nn\\
&\bigl(\chi_0^*(Z)L^{(\star)}\bigr)^{(l-1,n)}(S_1,\cdots,S_{l-1})
\label{chi0}\\
&
\quad=\frac{1}{P_n(-Z)}\,
\frac{1}{2}\,
\left(L^{(l,n)}(Z,S_1,\cdots,S_{l-1})
-L^{(l,n)}(-Z,S_1,\cdots,S_{l-1})
\right)\,.\nn
\end{align}
Here $L^{(\star)}\in \mathcal{T}_{0,c}$, and
\begin{align*}
\mathrm{Skew }f(x_1,\cdots,x_l)=
\sum_{\sigma\in \mathfrak{S}_l}\mathrm{sgn}\sigma\,
f(x_{\sigma(1)},\cdots,x_{\sigma{(l)}})
\end{align*}
stands for skew symmetrisation. 
Because of \eqref{C-2} and \eqref{C-3}, 
$\psi_0^*(Z)$, $\chi^*_0(Z)$ 
preserve the degree condition \eqref{deg-rest} and 
the recurrence relation \eqref{rec}. 
It is also easy to see that these operators mutually anticommute. 

Choosing representatives of restricted 
degrees \eqref{deg-rest} is convenient
to prove the completeness. Namely, following \cite{Fothers} 
it should be possible to prove that
$\psi ^*_0(Z)$, $\chi^*_0(Z)$ create the 
complete set of solutions to the recurrence relations
over the ring of $I_{2j-1}$, $\bar{I}_{2j-1}$.

Now 
we are confronted with the task of identifying these towers with local operators. 
Notice that the
coefficients entering the form factor axioms \eqref{axiom1}--\eqref{axiom3} 
are periodic in $\alpha$ of period $2(1-\nu)/\nu$.
Hence for all integers $m$ the axioms are the same for 
the fields $\Phi _{\al+2m\frac{1-\nu}\nu}$ and their descendants.
One has to be able to identify the towers corresponding to  all these fields.
Consider the simplest case:
the tower of 
\begin{align}
M^{(n)}_m=\langle\Phi _{\al +2m\frac{1-\nu}\nu}\rangle
\cdot S^{2m+1}\wedge S^{2m+2}\wedge\cdots\wedge S^{2n+2m-1}
\prod\limits_{j=1}^{2n}B_j^{-m}\,
\label{al+2mxi}
\end{align}
obviously corresponds to the form factors of the primary field
$
\Phi _{\al +2m\frac{1-\nu}\nu}(0)\,.
$
How does this fit into our description in terms of polynomials satisfying the restriction
\eqref{deg-rest}? Obviously, for $m>0$ 
(resp. $m<0$) the Laurent polynomial $M^{(n)}_m$ 
contains degrees higher than $2n-1$ 
(resp. lower than $0$). 
They have to be reduced using $Q$-exact
forms. We need some general description 
of this procedure which looks completely random for the moment.

The solution to the above problem was found in \cite{BBS} by considering the classical limit.
In the classical case 
the variables $S_1,\cdots ,S_n$  turn into the separated variables. 
With our conventions the local classical observables are represented by polynomials 
$L^{(n)}_{\mathcal{O}_\al}$ of odd degrees in all $S_j$. So, the idea is
to bring the towers created by $\psi ^*_0(Z)$, $\chi ^*_0(Z)$ to this
form using $Q$-exact forms. 
We shall see that a
very transparent structure will emerge. 


In order to make the fermions odd in $Z$,  
it suffices to multiply them by \eqref{sqrt-P} 
(that is, by making use of the integrals of motion).
To make their action odd also in $S$, we modify 
 $C_n(S_1,S_2)$ by adding $Q$-exact forms in $S_2$. 
Namely, we introduce formal series $C_{\pm,n}(S_1,S_2)$ such that 
\begin{align}
&C_{\pm,n}(S_1,S_2)\equiv 
C_n(S_1,S_2)
\quad (S_1^{\pm1}\to\infty,\
\mathrm{mod}\ \text{$Q$-exact forms in $S_2$})\,, 
\label{C+C-}
\end{align}
and which are odd in both variables. 
This can be achieved by setting
\begin{align*}
&C_{\pm,n}(S_1,S_2)=
\frac 1 2\sum_{\epsilon_1,\epsilon_2=\pm}
\epsilon_1\epsilon_2
P_n(\epsilon_1 S_1)P_n(\epsilon_2 S_2)
\tau_\pm(\epsilon_2 S_2/\epsilon_1 S_1)\,,
\end{align*}
where
\begin{align}
&\tau_+(x)=-\sum_{l=0}^{\infty}(-x)^l{\textstyle\frac i {2\nu}}
\cot{\textstyle\frac \pi 2}(\al +{\textstyle\frac l\nu})
\,,\label{tau+}\\
&\tau_-(x)=\sum^{-1}_{l=-\infty}(-x)^l{\textstyle\frac i {2\nu}}
\cot{\textstyle\frac \pi 2}(\al +{\textstyle\frac l\nu})\,.
\label{tau-}
\end{align}
Expand $C_{\pm,n}(Z,S)$ in $Z^{\mp1}$. Their coefficients of $Z^{2j-1}$
where $j<0$ or $j>n$ are $Q$-exact forms in $S$, 
and the remaining terms read as
\begin{align}
&C_{+,n}(Z,S)\equiv\sum\limits _{j=1}^nZ^{2n-2j+1}p_{2j-1}(S)
\quad(\mathrm{mod}\text{ $Q$-exact forms in $S$})\,,\label{Pexpand1}
\\
&p_{2j-1}(S)=\sum\limits _{k=1}^{n+j}S^{2k-1}p_{2j-1,2k-1}\,,
\quad p_{2j-1,2n+2j-1}
={\textstyle\frac i \nu}\cdot\cot{\textstyle\frac\pi2(\al+\frac{2j-1}\nu)},
\nn
\end{align}
\begin{align}
&C_{-,n}(Z,S)\equiv\sum\limits _{j=1}^nZ^{2j-1}\bar{p}_{2j-1}(S)
\quad(\mathrm{mod}\text{ $Q$-exact forms in $S$})\,,
\label{Pexpand2}
\end{align}
\begin{align*}
&\bar{p}_{2j-1}(S)=\sum\limits _{k=1}^{n+j}S^{2n-2k+1}\bar{p}_{2j-1,2k-1}\,
\quad\bar{p}_{2j-1,2n+2j-1}
=
\sigma_{2n}(B)^2
\cdot
{\textstyle\frac i \nu}\cot{\textstyle\frac\pi2(\al-\frac{2j-1}\nu)}.
\end{align*}

After these modifications we define new fermions as formal series at 
$Z^{\pm 1}=\infty$:
\begin{align}
&\bigl(\psi^{'*}(Z)L^{(\star)}\bigr)^{(l+1,n)}(S_0,\cdots,S_l)
\label{psip}\\
&\quad=
\frac{1}{\sqrt{P_n(Z)P_n(-Z)}}\,
\frac{1}{l!}
\text{Skew}_{S_0,\cdots,S_l}C_{+,n}(Z,S_0)L^{(l,n)}(S_1,\cdots,S_l)\,,
\nn\\
&\bigl(\bar{\psi}^{'*}(Z)L^{(\star)}\bigr)^{(l+1,n)}(S_0,\cdots,S_l)
\label{psim}\\
&\quad=
\frac{1}{\sqrt{P_n(Z)P_n(-Z)}}\,
\frac{1}{l!}
\text{Skew}_{S_0,\cdots,S_l}C_{-,n}(Z,S_0)L^{(l,n)}(S_1,\cdots,S_l)\,,
\nn
\end{align}
where 
the right hand sides 
are understood as power series expansions 
in $Z^{\mp1}$.
Define also $\chi^{'*}(Z)$, $\bar{\chi}^{'*}(Z)$ by 
the same formula \eqref{chi0} as for $\chi^*_0(Z)$, replacing 
\begin{align*}
\frac{1}{P_n(-Z)}\longrightarrow \frac{1}{\sqrt{P_n(Z)P_n(-Z)}}\,
\end{align*}
and Taylor expanding in $Z^{\mp1}$. 

Finally we define $\psi^*(Z),\ \bar{\psi}^*(Z)$,
 $\chi^*(Z),\ \bar{\chi}^*(Z)$ by a Bogolubov transform
\begin{align}
&\psi^*(Z)=e^\Xi\psi^{'*}(Z)e^{-\Xi},
\quad 
\bar{\psi}^*(Z)=e^\Xi\bar{\psi}^{'*}(Z)e^{-\Xi},
\label{def-psi1}\\
&
\chi^*(Z)=e^\Xi\chi^{'*}(Z)e^{-\Xi},
\quad 
\bar{\chi}^*(Z)=e^\Xi\bar{\chi}^{'*}(Z)e^{-\Xi},
\label{def-psi2}
\end{align}
where $\Xi$ acting on the $2n$-particle component is 
\begin{align*}
&\Xi_n=\oint\!\!\oint\frac{dZ}{2\pi i Z} \frac{dX}{2\pi i X}
\left(\widehat{C}_{+,n}(Z,X)\psi(Z)\bar{\chi}(X)
+\widehat{C}_{-,n}(Z,X)\bar{\psi}(Z){\chi}(X)
\right)\,,
\\
&\widehat{C}_{\pm,n}(Z,X)=
\frac{1}{\sqrt{P_n(Z)P_n(-Z)}\sqrt{P_n(X)P_n(-X)}}\,C_{\pm,n}(Z,X)\,.
\end{align*}
In the above we have introduced formally the annihilation operators
$\psi(Z)$, $\bar{\psi}(Z)$, $\chi(Z)$, $\bar{\chi}(Z)$, which are  
canonically conjugate to 
$\psi^{'*}(Z)$, $\bar{\psi}^{'*}(Z)$, 
$\chi^{'*}(Z)$, $\bar{\chi}^{'*}(Z)$, respectively 
and annihilate $M^{(\star)}_0$.
Obviously, the same operators are canonically conjugated to 
$\psi^{*}(Z)$, $\bar{\psi}^{*}(Z)$, 
$\chi^{*}(Z)$, $\bar{\chi}^{*}(Z)$. That is why we did not put prime in the notation. 
Note, for example, that
\begin{align*}
\psi^*(Z)=\psi^{'*}(Z)-\oint\!\!\frac{dX}{2\pi i X}\widehat{C}_{+,n}(Z,X)\bar{\chi}(X).
\end{align*}
The modification 
\eqref{def-psi1}, \eqref{def-psi2}
is useful for separating the role of the two chiralities, 
as it will be explained shortly. 

With the definition of $\psi^*(Z),\bar{\psi}^*(Z)$,
$\chi^*(Z),\bar{\chi}^*(Z)$ given above, 
let us write down their action on the tower $M^{(\star)}_0$. 
Suppose $k-k'=l-n$. 
\begin{align}
&\bigl(\psi ^*(Z_1)\cdots \psi ^*(Z_p)
\bar{\psi} ^*(Z_{p+1})\cdots \bar{\psi} ^*(Z_{k})
\label{genferm1} \\
&\quad\times \bar{\chi }^*(X_{k'})\cdots 
\bar{\chi }^*(X_{q+1})\chi ^*(X_q)\cdots \chi ^*(X_1)
M_0^{(\star)}\bigr)^{(l,n)}(S_1,\cdots, S_l)\nn\\
&=\langle\Phi _\al\rangle(-)^{kk'}
\frac{1}{\prod\limits_{j=1}^k\sqrt{P_n(Z_j)P_n(-Z_j)}\ \ 
\prod\limits_{j=1}^{k'}\sqrt{P_n(X_j)P_n(-X_j)}}
\cdot\left|\ 
\begin{matrix}
\mathcal{A}&\mathcal{B}\\ \mathcal{C}&\mathcal{D}
\end{matrix}
\ \right|\,,\nn
\end{align}
where $\mathcal{A}$, $\mathcal{B}$, $\mathcal{C}$ and $\mathcal{D}$
are respectively 
$k\times k'$, $k\times l$, $n\times k'$ and $n\times l$ 
matrices:
\begin{align}
\mathcal{A}=
\begin{pmatrix}
0&\cdots &0&C _+(Z_1,X_{q+1})&\cdots&C _+(Z_1,X_{k'})\\
\vdots&\   &\vdots &\vdots&\ &\vdots \\
0&\cdots &0&C _+(Z_p,X_{q+1})&\cdots&C _+(Z_p,X_{k'})\\
C _-(Z_{p+1},X_1)&\cdots&C _-(Z_{p+1},X_{q})&0&\cdots &0\\
\vdots&\   &\vdots &\vdots&\ &\vdots \\
C _-(Z_{k},X_1)&\cdots&C _-(Z_{k},X_{q})&0&\cdots &0\\
\end{pmatrix}\,.\nn
\end{align}
\begin{align}
\mathcal{B}=\begin{pmatrix}
C _+(Z_1,S_1)&\cdots &C _+(Z_1,S_l)
\\
\vdots&\   &\vdots \\
C _+(Z_p,S_1)&\cdots &C _+(Z_p,S_l)\\
C _-(Z_{p+1},S_1)&\cdots&C _-(Z_{p+1},S_l)\\
\vdots&\   &\vdots \\
C _-(Z_{k},S_1)&\cdots&C _-(Z_{k},S_l)
\end{pmatrix}\,,\nn
\end{align}
\begin{align}
\mathcal{C}=\begin{pmatrix}
X_{1} &\cdots&X_{k'}\\ 
\vdots &\ &\vdots\\
X^{2n-1}_{1} &\cdots &X^{2n-1}_{k'} 
\end{pmatrix}\,,
\quad
\mathcal{D}=\begin{pmatrix}
S_1 &\cdots&S_l\\
\vdots &\ &\vdots\\
S^{2n-1}_1 &\cdots &S^{2n-1}_l 
\end{pmatrix}\,.\nn
\end{align}
The Fourier modes are introduced by 
\begin{align}
&\psi ^*(Z)=\sum_{j=1}^\infty Z^{-2j+1}\psi ^*_{2j-1},
\quad
\chi ^*(X)=\sum_{j=1}^\infty X^{-2j+1}\chi ^*_{2j-1}\,,
\label{EXPANSION1}\\
&\bar{\psi }^*(Z)=\sum_{j=1}^\infty Z^{2j-1}\bar{\psi} ^*_{2j-1},
\quad
\bar{\chi} ^*(X)=\sum_{j=1}^\infty X^{2j-1}\bar{\chi} ^*_{2j-1}\,.
\label{EXPANSION2}
\end{align}
A similar determinant formula is obtained for
$\psi^*_0(Z_1)\cdots\psi^*_0(Z_k)\chi^*_0(X_{k'})\cdots\chi^*_0(X_{1})M^{(\star)}_0$,
for which the $\mathcal A$-part is $0$.
For completeness we give the Fourier decomposition for the annihilation operators
\begin{align*}
&
\psi(Z)=\sum_{j=1}^{\infty} Z^{2j-1}\psi_{2j-1}\,,\qquad
\bar\psi(Z)=\sum_{j=1}^{\infty} Z^{-2j+1}\bar\psi_{2j-1},\\
&\chi (Z)=\sum_{j=1}^{\infty} Z^{2j-1}\chi _{2j-1}\,,\qquad
\bar{\chi} (Z)=\sum_{j=1}^{\infty} Z^{-2j+1}\bar{\chi }_{2j-1},
\end{align*}
and the non-vanishing commutation relations
$$
[\psi _{2j-1},\psi ^*_{2k-1}]_+=
[\chi _{2j-1},\chi^* _{2k-1}]_+=
[\bar{\psi }_{2j-1},\bar{\psi }_{2k-1}^*]_+=
[\bar{\chi } _{2j-1},\bar{\chi } _{2k-1}^*]_+=\delta _{j,k}\,.
$$


Let us explain the formula \eqref{genferm1} using physicist's
terminology,
choosing $k=k'$ for simplicity. 
Suppose we create a new tower starting from the primary one
by application of several operators $\psi ^*_{2j-1}$,
$\chi ^*_{2j-1}$,
$\bar{\psi} ^*_{2j-1}$,
$\bar{\chi} ^*_{2j-1}$. Let us take $n\gg j$
for all indices $j$ of the operators involved. Then the polynomial
$
M^{(n)}_{0}(S_1,\cdots ,S_n)=\langle \Phi _{\al}\rangle S\wedge S^3\wedge \cdots\wedge S^{2n-1}
$
can be considered as Fermi zone with two ends situated at $S$ and $S^{2n-1}$.
Roughly speaking
the operator $\chi ^*_{2j-1}$ creates a hole at $S^{2n-2j+1}$, $\bar\chi ^*_{2j-1}$ a hole at $S^{2j-1}$,
$\psi ^*_{2j-1}$  a particle at $S^{2n+2j-1}$, and $\bar\psi ^*_{2j-1}$ a particle at $S^{-2j+1}$.

The tower for the primary field is represented by the pure wedge product 
\eqref{FERMIZONE}. 
Those for the descendants are linear combinations of pure 
wedge products with symmetric Laurent polynomials  in $B_j$ as coefficients.
We discuss which pure wedge products appear in the linear combination. 
Suppose $a_j,b_j\in2\Z_{>0}-1$, and we pick up
\begin{align}
(\psi^*_{a_1}\cdots\psi^*_{a_p}\bar\psi^*_{a_{p+1}}\cdots\bar\psi^*_{a_k}
\bar\chi^*_{b_k}\cdots\bar\chi^*_{b_{q+1}}\chi^*_{b_q}\cdots\chi^*_{b_1}
M_0^{(\star)}\bigr)^{(n,n)}(S_1,\cdots, S_n)\label{FOURIER}
\end{align}
from \eqref{genferm1} taking the coefficient of
$\prod_{j=1}^pZ_j^{-a_j}\prod_{j=p+1}^kZ_j^{a_j}\prod_{j=1}^qX_j^{-b_j}\prod_{j=q+1}^kX_j^{b_j}$.
We assume that $a_1>\ldots>a_p$, $a_{p+1}<\ldots<a_k$, $b_1<\ldots<b_q$
and $b_{q+1}>\ldots>b_k$. Moreover we assume that $n$ is large enough so that
$2n-b_q>b_{q+1}$, too. Note that
$m=p-q$ is the excess of particles at the right end of the Fermi zone,
and $-m$ is the excess of particles at the left end. We call $m$ a weight.
A basic example of a tower of weight $m$ is \eqref{al+2mxi}.

The procedure of taking the coefficient is as follows. 
First we take the coefficient of $Z_j^{-a_j}$ $(1\leq j\leq p)$ 
and $Z_j^{a_j}$ $(p+1\leq j\leq k)$
in the first $k$ rows. Second we take the coefficient of 
$X_j^{-b_j}$ $(1\leq j\leq q)$ and
$X_j^{b_j}$ $(q+1\leq j\leq k)$ in the first $k$ columns. 
For example, take $1\leq j\leq p$ and consider
the coefficient of $Z_j^{-a_j}$ in 
$C_{+,n}(Z_j,X)/\sqrt{P_n(Z_j)P_n(-Z_j)}$. 
We expand this for $Z_j\rightarrow\infty$.
The leading term is obtained from the approximation
\begin{align*}
\sqrt{P_n(Z_j)P_n(-Z_j)}\sim Z_j^{2n}\,,
\end{align*}
the coefficient being $p_{a_j}(X)$. 
The sub-leading terms are multiples of $p_{2l-1}(X)$ where $2l-1<a_j$.
Therefore, if we consider a partial-ordering of 
pure wedge products by the natural ordering of $a_j$, 
the sub-leading terms has coefficients of lower order.  
It is the same for $p+1\leq j\leq k$ and the column expansions.
Collecting the coefficients in the leading order we obtain a matrix, whose
$(\mathcal A,\mathcal B)$ part is given by 
\begin{align}
\left(
\begin{matrix}
\text{\huge0}&
\begin{matrix}
p_{a_1,b_{q+1}}&\cdots&p_{a_1,b_k}\\
\vdots&&\vdots&\\
p_{a_p,b_{q+1}}&\cdots&p_{a_p,b_k}
\end{matrix}
&
\begin{matrix}
p_{a_1}(S_1)&\cdots&p_{a_1}(S_n)\\
\vdots&&\vdots&\\
p_{a_p}(S_1)&\cdots&p_{a_p}(S_n)\\
\end{matrix}\\
\\
\begin{matrix}
\bar{p}_{a_{p+1},b_1}&\cdots&\bar{p}_{a_{p+1},b_q}\\
\vdots&&\vdots&\\
\bar{p}_{a_k,b_1}&\cdots&\bar{p}_{a_k,b_q}
\end{matrix}&
\text{\huge0}&
\begin{matrix}
\bar{p}_{a_{p+1}}(S_1)&\cdots&\bar{p}_{a_{p+1}}(S_n)\\
\vdots&&\vdots&\\
\bar{p}_{a_k}(S_1)&\cdots&\bar{p}_{a_k}(S_n)\\
\end{matrix}\\
\end{matrix}
\right)\label{AB}
\end{align}
Consider $q+1\leq j\leq k$ and set $b_j=2l-1$.
The $l$-th row in the $(\mathcal C,\mathcal D)$ part gives rise to
\begin{align}
(
\raise2pt\hbox{$
\begin{matrix}
\buildrel{\text{$1$-st}}\over0&\cdots
\buildrel{\text{$q$-th}}\over0
&\buildrel{\text{$(q+1)$-st}}\over0&\cdots
&0&\buildrel{\text{$j$-th}}\over1&0&\cdots&
&\buildrel{\text{$k$-th}}\over0
\end{matrix}$}
\
\begin{matrix}
S_1^{b_j}&\cdots&S_n^{b_j}
\end{matrix}
)\label{ROW1}
\end{align}
Similarly, for $1\leq j\leq q$, set $2n-b_j=2l-1$.
The $l$-th row in the $(\mathcal C,\mathcal D)$ part gives rise to
\begin{align}
(\raise2pt\hbox{$
\begin{matrix}
\buildrel{\text{$1$-st}}\over0&\cdots
&0&\buildrel{\text{$j$-th}}\over1&0
&\cdots&\buildrel{\text{$q$-th}}\over0
&\buildrel{\text{$(q+1)$-st}}\over0&\cdots
&\buildrel{\text{$k$-th}}\over0
\end{matrix}
$}\
\begin{matrix}
S_1^{2n-b_j}&\cdots&S_n^{2n-b_j}
\end{matrix}
)\label{ROW2}
\end{align}
If $2l-1\not\in\{b_1,\ldots,b_q,2n-b_{q+1},\ldots,2n-b_k\}$
the $l$-th row in the $(\mathcal C,\mathcal D)$ part gives rise to
\begin{align*}
(
\raise2pt\hbox{$
\begin{matrix}
\buildrel{\text{$1$-st}}\over0&\cdots&
\buildrel{\text{$q$-th}}\over0
&\buildrel{\text{$(q+1)$-st}}\over0&\cdots&\buildrel{\text{$k$-th}}\over0
\end{matrix}
$}
\
\begin{matrix}
S_1^{2l-1}&\cdots&S_n^{2l-1}
\end{matrix}
)
\end{align*}
Subtracting certain multiples of rows \eqref{ROW1} ,\eqref{ROW2} 
from the $(\mathcal A,\mathcal B)$ part \eqref{AB}, we change the latter to
\begin{align}
\left(
\begin{matrix}
\text{\huge0}
&\text{\huge0}&
\Bigl(p_{a_j}(S_k)-\sum_{l=q+1}^kp_{a_j,b_l}S_k^{b_l}
\Bigr)_{j=1,\ldots,p\atop k=1,\ldots,n}
\\\text{\huge0}
&\text{\huge0}&
\Bigl(\bar p_{a_j}(S_k)-\sum_{l=1}^q\bar p_{a_j,b_l}S_k^{2n-b_l}
\Bigr)_{j=1,\ldots,p\atop k=1,\ldots,n}
\end{matrix}
\right)\label{PARTICLE}
\end{align}
In physicist's terminology \eqref{ROW1} implies the existence 
of a hole at $S^{b_j}$, and
\eqref{ROW2} hole at $S^{2n-b_j}$. On the other hand we obtain a row from
\eqref{PARTICLE} by expanding it in the same power in $S_k$:
\begin{align*}
p_{a_j,2l-1}
(
\raise2pt\hbox{$
\begin{matrix}
\buildrel{\text{$1$-st}}\over0&\cdots&\buildrel{\text{$k$-th}}\over0
\end{matrix}
$}
\
\begin{matrix}
S_1^{2l-1}&\cdots&S_n^{2l-1}
\end{matrix}
)\quad(1\leq j\leq p;2l-1\not\in\{b_{q+1},\ldots,b_n\})
\end{align*}
or
\begin{align*}
\bar p_{a_j,2l-1}
(
\raise2pt\hbox{$
\begin{matrix}
\buildrel{\text{$1$-st}}\over0&\cdots&\buildrel{\text{$k$-th}}\over0
\end{matrix}
$}
\
\begin{matrix}
S_1^{2n-2l+1}&\cdots&S_n^{2n-2l+1}
\end{matrix}
)\quad(p+1\leq j\leq k;2l-1\not\in\{b_1,\ldots,b_q\})
\end{align*}
The former creates a particle at $S^{2l-1}$ if $2l-1>2n$, and fills the hole at $S^{2n-b_j}$
if $2l-1=2n-b_j$ for some $1\leq j\leq q$. The latter creates a particle at $S^{2n-2l+1}$
if $2l-1>2n$, and fills the hole at $S^{b_j}$ if $2n-2l+1=b_j$ for some $q+1\leq j\leq k$.
It is important that the change caused by filling a hole by a particle occurs only if both particle and hole are
near the right end, or both near the left end.
Therefore, the weight $m$ is invariant by this change.
In other words, one can define the weight $m$ of a tower
without ambiguity. 
This is the effect of the Bogolubov transform 
\eqref{def-psi1}, \eqref{def-psi2} which 
introduces the non-trivial block $\mathcal A$. 

From \eqref{Pexpand1}, \eqref{Pexpand2} it is easy to see that, 
up to a sign, the leading term in \eqref{FOURIER} is given by
\begin{align}
&\langle\Phi _\al\rangle({\textstyle\frac i\nu})^k\prod_{j=1}^{2n}B_j^{q-p}\prod_{j=1}^p\cot{\textstyle\frac\pi2(\al+\frac{a_j}\nu})
\prod_{j=p+1}^k\cot{\textstyle\frac\pi2(\al-\frac{a_j}\nu})\cdot
S^{-a_k}\wedge\cdots\wedge S^{-a_{p+1}}\label{LEADING}\\
&\quad\wedge S\wedge\cdots(S^{b_k})^\vee\cdots(S^{b_{q+1}})^\vee\cdots
(S^{2n-b_q})^\vee\cdots(S^{2n-b_1})^\vee\cdots\wedge S^{2n-1}\nn\\
&\qquad\wedge S^{2n+a_p}\wedge\cdots\wedge S^{2n+a_1}.\nn
\end{align}
Here $(S^b)^\vee$ represents a hole. 

Our goal is to identify this fermionic picture 
with the content of local fields in the sG model. 
Here we do it qualitatively, leaving  
the quantitative identification to later sections. 


Returning to the primary fields we find
\begin{align}
&\psi^*_1\cdots\psi^*_{2m-1}\bar{\chi}_{2m-1}^*\cdots\bar{\chi}^*_1
\ M_0^{(\star)}\label{al+2mxi1}\\&=\frac{\langle\Phi _\al\rangle}
{\langle\Phi _{\al+2m\frac {1-\nu}\nu}\rangle}\cdot({\textstyle \frac {i}\nu})^m
\prod_{j=1}^m\cot{\textstyle \frac \pi {2\nu}}(\al\nu+(2j-1))
\ M_m^{(\star)}\,,
\nn
\\
&\bar{\psi}^*_1\cdots\bar{\psi}^*_{2m-1}\chi_{2m-1}^*\cdots\chi^*_1
\ M_0^{(\star)}\label{al+2mxi2}\\&=\frac{\langle\Phi _\al\rangle}
{\langle\Phi _{\al-2m\frac {1-\nu}\nu}\rangle}\cdot
({\textstyle \frac {i}\nu})^{m}
\prod_{j=1}^{m}
\cot{\textstyle \frac \pi {2\nu}}(\al\nu-(2j-1))
\ M_{-m}^{(\star)}\,.
\nn
\end{align}
In general, it is now clear that
the space of descendants of $M^{(\star)}_0$ by fermions,
together with the action of the local integrals of motion,   
has the same character as that of the space of fields in CFT:
$$
\bigoplus\limits _{m=-\infty}^{\infty}\mathcal{V}_{\al+2\xin m}
\otimes \overline{\mathcal{V}}_{\al+2\xin m}\,,
$$
where $\mathcal{V}_\al$ denotes the Virasoro Verma
module with central charge $c=1-\frac{6\nu^2}{1-\nu}$ and highest weight 
$\Delta_\al=\frac{\nu^2}{4(1-\nu)}\al(\al-2)$.

We remark that often it is convenient 
to change the definition slightly and work with fermions 
given in \eqref{genferm1} wherein 
the block $\mathcal{B}$ is replaced by 
\begin{align}
\mathcal{B}\ \rightarrow\ \begin{pmatrix}
C (Z_1,S_1)&\cdots &C (Z_1,S_l)
\\
\vdots&\   &\vdots \\
C (Z_{k},S_1)&\cdots&C (Z_{k},S_l)
\end{pmatrix}\,.
\label{modify-fermion}
\end{align}
The resulting fermions have mode expansions 
\eqref{EXPANSION1}, \eqref{EXPANSION2} in odd degrees. 
Acting on the primary field, the Fourier components 
generate towers which coincide with \eqref{genferm1}
modulo $Q$-exact forms, 
but satisfying the restricted degree condition \eqref{deg-rest}.
It should be noted, however, that
we cannot replace $C_{\pm}$ 
to $C$ in the block $\mathcal{A}$ because their 
second arguments are not the integration variables.

\section{BJMS fermionic description of sG model}\label{BJMSfermions}

In the paper \cite{HGSIV} we gave a  fermionic description of CFT. 
This description has an advantage of
being compatible with the integrable perturbation.
As we explained in
\cite{OP,HGSV} the logic of Perturbed Conformal Field Theory (PCFT)
\cite{alzam,FFLZZ} implies that in a
generic situation ($\nu$ and $\al$ irrational)
there is a
one-to-one correspondence between the local fields before 
and after perturbation, i.e. 
they consist of primary exponential fields $\Phi _\al$ and
their descendants created by two copies of the
Virasoro algebra with generators
$\mathbf{l}_{k}$, $\bar{\mathbf{l}}_{k}$. 
Following \cite{HGSIV} we consider 
the case 
$$0<\al<2\,.$$

Equivalence to the fermionic 
description implies that the local fields are created from the primary
field by the local integrals of motion $\mathbf{i}_{2k-1}$,  $\bar{\mathbf{i}}_{2k-1}$
and the fermions $\betab ^*_{2k-1}$, $\gammab^*_{2k-1}$, 
$\bar{\betab} ^*_{2k-1}$, $\bar{\gammab} ^*_{2k-1}$. It is explained in
\cite{HGSIV,Boos} how to recalculate the result of acting with the fermions
in terms of the usual Virasoro descendants. Unfortunately, the techniques
of \cite{HGSIV,Boos}  allow us to do that only modulo the action of the local
integrals of motion. 
At the end of the next section we give some explanation to this point.
This was not a problem for the papers \cite{OP,HGSIV}
where the one-point functions in the sG model were calculated, because
they vanish on the descendants created by the local integrals of motion.
But in the context of the present paper this is 
an unpleasant restriction.
We hope that the technical difficulties behind this problem will be resolved
in the future. The origin of these difficulties will be explained in the next section.
So, for the moment we work with the space
$$\mathcal{V}^\mathrm{quo}_\al\otimes
\overline{\mathcal{V}}^\mathrm{quo}_\al\,,$$
where
$$
\mathcal{V}^\mathrm{quo}_\al=
\mathcal{V}_\al\ /\ \sum_{k=1}^\infty\mathbf{i}_{2k-1}\mathcal{V}_\al,
\quad \overline{\mathcal{V}}^\mathrm{quo}_\al=
\overline{\mathcal{V}}_\al\ /\ \sum_{k=1}^\infty
\bar{\mathbf{i}}_{2k-1}\overline{\mathcal{V}}_\al\,.
$$
Let us emphasise one more time that the operators 
$\mathbf{i}_{2k-1}$,  $\bar{\mathbf{i}}_{2k-1}$, $\betab ^*_{2k-1}$, $\gammab^*_{2k-1}$, 
$\bar{\betab} ^*_{2k-1}$, $\bar{\gammab} ^*_{2k-1}$ are defined  on
$\mathcal{V}_\al\otimes
\overline{\mathcal{V}}_\al$ and the restriction to 
$\mathcal{V}^\mathrm{quo}_\al\otimes
\overline{\mathcal{V}}^\mathrm{quo}_\al$ is due to purely technical reasons. 

To be precise the
space $\mathcal{V}^\mathrm{quo}_\al\otimes
\overline{\mathcal{V}}^\mathrm{quo}_\al$
allows a fermionic basis:
\begin{align}
\betab^* _{I^+}\bar{\betab} ^*_{\bar{I}^+}
\bar{\gammab }^*_{\bar{I}^-}\gammab^*_{I^-}
\Phi _{\al}(0)\,,\label{basis}
\end{align}
where $\#(I^+)=\#(I^-)$, $\#(\bar I^+)=\#(\bar I^-)$. We set generally
\begin{align*}
&\betab _I^*=\betab^*_{a_1}\cdots \betab _{a_p}^*,\quad\gammab _I^*=\gammab^*_{a_p}\cdots \gammab_{a_1}^*,
\quad
\bar\betab _I^*=\bar\betab^*_{a_1}\cdots \bar\betab _{a_p}^*,\quad \bar\gammab _I^*=\bar\gammab^*_{a_p}\cdots \bar\gammab_{a_1}^*, \\
&\mathrm{for}\ I=\{a_1,\cdots ,a_p\},\quad a_1<a_2<\cdots <a_p\,.
\end{align*}

 On the other hand the space
$\mathcal{V}^\mathrm{quo}_\al\otimes
\overline{\mathcal{V}}^\mathrm{quo}_\al$ can be realised as result
of acting on the primary field $\Phi _{\al}$ by  even generators of two
chiral Virasoro algebras $\mathbf{l}_{-2k}$, $\bar{\mathbf{l}}_{-2k}$.
The formulae relating
the fermionic basis with the Virasoro basis
can be found up to level 8 in \cite{HGSIV,Boos}. 
The fermionic basis was used in \cite{OP,HGSV} because it allows {one} to
compute the one-point functions. 

We use the same 
convention for multi-indices as in \cite{HGSIV,HGSV}:
\begin{align}
&\mathrm{if} \quad J=\{j_1,\cdots ,j_p\}\quad
\mathrm{then}\quad aJ+b=\{aj_1+b,\cdots ,aj_p+b\}\,,\nn\\
&I(m)=\{1,2,\cdots ,m\},\quad I_\mathrm{odd}(m)=2I(m)-1\,.\nn
\end{align}

The operators  $\betab ^*_{2k-1}$, $\gammab^*_{2k-1}$ and 
$\bar{\betab} ^*_{2k-1}$, $\bar{\gammab} ^*_{2k-1}$ are combined into 
the generating functions 
\begin{align}
 &\betab ^*(\la)=\sum_{k=1}^\infty\la ^{-\frac{2k-1}\nu} \betab ^*_{2k-1},\quad
 \gammab^*(\la)=\sum_{k=1}^\infty\la ^{-\frac{2k-1}\nu} \gammab^*_{2k-1}\,,\label{betaCFT}\\
 &\bar{\betab} ^*(\la)=\sum_{k=1}^\infty\la ^{\frac{2k-1}\nu}\bar{\betab} ^*_{2k-1},
 \quad\quad
\bar{\gammab} ^*(\la)=\sum_{k=1}^\infty\la ^{\frac{2k-1}\nu} \bar{\gammab} ^*_{2k-1}\,.
\nn
\end{align}
In a weak sense these series describe the asymptotics of holomorphic functions
at  $\la\to\infty$ and $\la\to 0$ for right and left 
chiralities respectively. 
When expanded at the opposite points,
the same holomorphic functions
create another interesting set of operators 
(actually a slight modification is needed \cite{HGSIV,HGSV}).
This allows us to define the operators
\begin{align}
\gammab ^*_{\mathrm{screen}}(\la)=
\sum_{k=1}^\infty\la ^{2k-\al}\gammab ^*_{\mathrm{screen}, k},\quad
\bar{\betab}_{\mathrm{screen}}^*(\la)
=\sum_{k=1}^\infty\la ^{\al-2k}\bar{\betab}^*_{\mathrm{screen},k}\,.
\label{screenCFT}
\end{align}

We define as usual 
\begin{align}\gammab ^*_{\mathrm{screen}, I(m)}=
\gammab ^*_{\mathrm{screen}, m}\cdots\gammab ^*_{\mathrm{screen}, 1}\,,\quad
\bar{\betab}^*_{\mathrm{screen},I(m)}=
\bar{\betab}^*_{\mathrm{screen},1}\cdots \bar{\betab}^*_{\mathrm{screen},m}\,.\label{mult2}\end{align}
Then the $m$-fold screened primary field is by definition
\begin{align}
\Phi _{\al}^{(m)}(0)=
{i^m}\mub^{2m}\prod\limits _{j=1}^{m}\cot {\textstyle \frac{\pi\nu}2(2j-\al)}
\times
\bar{\betab}^*_{\mathrm{screen},I(m)}\ \gammab ^*_{\mathrm{screen}, I(m)}
\Phi _{\al}(0)\,,
\label{phim}
\end{align}
where the multiplier in the right hand side is introduced for 
convenience. 

It is explained in \cite{HGSV} that
$\mathcal{V}^\mathrm{quo}_{\al+2m\frac {1-\nu}\nu}\otimes
\overline{\mathcal{V}}^\mathrm{quo}_
{\al+2m\frac {1-\nu}\nu}$  for $m > 0$ can be
constructed as
\begin{align}
&\betab^*_{I^+}\bar\betab^*_{\bar{I}^+}
\bar\gammab^*_{\bar{I}^-}\gammab^*_{I^-}
\Phi_{\alpha+2m\frac{1-\nu}{\nu}}(0)\label{IDENDESC}\\
 &\cong C_m(\al)\betab^*_{I^+ +2m}\bar\betab^*_{\bar{I}^+-2m}
\bar\gammab^*_{\bar{I}^- +2m}\gammab^*_{I^--2m}\betab^*_{\Io}\bar\gammab^*_{\Io}
\Phi^{(m)}_{\alpha}(0),\nn
\end{align}
where for negative indices we set
\begin{align}
&\gammab^*_{-a}=-t_a(\al)\betab_a,\quad \bar{\betab}^*_{-a}=
-t_a(2-\al)\bar{\gammab }_a\,,\quad t_a(\al)={\textstyle\frac i\nu}\cot
{\textstyle\frac\pi{2\nu}}(\al\nu+a)\,,
\label{gammanegative}
\end{align}
the only non-trivial anticommutation relations are 
$$[\betab_a,\betab^*_b]_+=\delta_{a,b}\,,
\quad [\bar{\gammab }_a,\bar{\gammab }^*_b]_+=\delta_{a,b}\,,$$
and
$$\betab_a\Phi^{(m)}_{\alpha}(0)=0\,,\quad \bar\gammab_a\Phi^{(m)}_{\alpha}(0)=0\,.$$

In particular,
\begin{align}
\Phi _{\al +2m\frac {1-\nu}\nu}(0)\ {\cong}\ C_m(\al)\betab^*_{\Io}
\bar{\gammab}^*_{\Io}
\Phi ^{(m)}_{\al}(0)\,,
\label{shift-primary}
\end{align}
where (see \cite{HGSV})
\begin{align}
C_m(\al)=\frac {\langle \Phi _{\al+m\xin}\rangle}
{\langle \Phi _{\al}\rangle}{(-i\nu)^m}\prod_{j=1}^m\tan {\textstyle \frac \pi {2\nu}}(\al\nu+2j-1)
\,.\label{Cm}
\end{align}
The fermionic screening {operators $\bar{\betab}^*_{{\rm screen}, 2j-1}$,
$ \gammab ^*_{{\rm screen}, 2j-1}$} play
an
important role in the CFT computations. However, in the sG case
when we consider the one-point function on the cylinder
\cite{HGSV} their contribution completely factors out. This results in the
possibility of the identification in the weak sense:
\begin{align}
\Phi _\al^{(m)}(0)=\hskip -.4cm\raisebox{9pt}{${}_
 \mathrm{w}$}\ \Phi _\al(0)\,.\label{Phi=Phi}
\end{align}
This is a general fact for the sG model. We shall comment on it later.

Considering the formulae \eqref{shift-primary}, \eqref{Cm}, \eqref{Phi=Phi} we
find an amazing similarity  to the formula \eqref{al+2mxi}. 
This similarity was one of starting points for our present research.

Recall that up to now we considered only the case $0<\al<2$. According to 
\eqref{IDENDESC} it is actually sufficient to consider the fundamental domain
\begin{align}
0<\al<2{\textstyle\frac {1-\nu}\nu}\,.\label{fundomain}
\end{align}
We excluded the point $\al =2\xin$ from the fundamental
domain for the following reason. An important role in the definition
of the
BJMS fermions is played by the function $\Delta ^{-1}_\z\psi _0(\z,\al)$
discussed in the next section and its asymptotics \eqref{asDinv}.
For $\al =2\xin$ the asymptotic
formula at $\z\to 0$ breaks down,
and a 
logarithmic term appears
from 
the $j=2$ term in
the first sum
and the $j=1$ term in
the second sum. 
So, the case $\al =2\xin$ requires a
special treatment.

\section{Six-vertex model as origin of BJMS fermionic description}\label{6vertex}

We introduce fermions for $c<1$ CFT and for the
sG model taking
the scaling limit of the
homogeneous and inhomogeneous six-vertex
model on the cylinder. This is explained in details in the
paper \cite{HGSV}, so, we shall be very brief here. But we would like
to repeat some facts from the paper \cite{HGSIII} which is the cornerstone 
for all of our recent researches. 

Let us consider the  inhomogeneous six-vertex model on the cylinder.
We count the {sites} in the space direction by indices $j$, and in the Matsubara
direction by $\mathbf{m}$. The model can be inhomogeneous in both
directions with the parameters of inhomogeneity $\z_j$ and $\tau_{\mathbf{m}}$.
The Boltzmann weights are combined into the $L$-operator $L_{j,\mathbf{m}}(\z_j/\tau_
{\mathbf{m}})$. Using the equivalent XXZ spin chain language we introduce the
operator $q^{2\al S(0)}\mathcal{O}$ with $\mathcal{O}$ being local and 
$S(k)=\frac 1 2 \sum_{j=-\infty}^{k}\sigma ^3_j$. {Then} we consider the partition function
as in the following figure:
\vskip .5cm
\hskip 0.1cm\includegraphics[height=5cm]{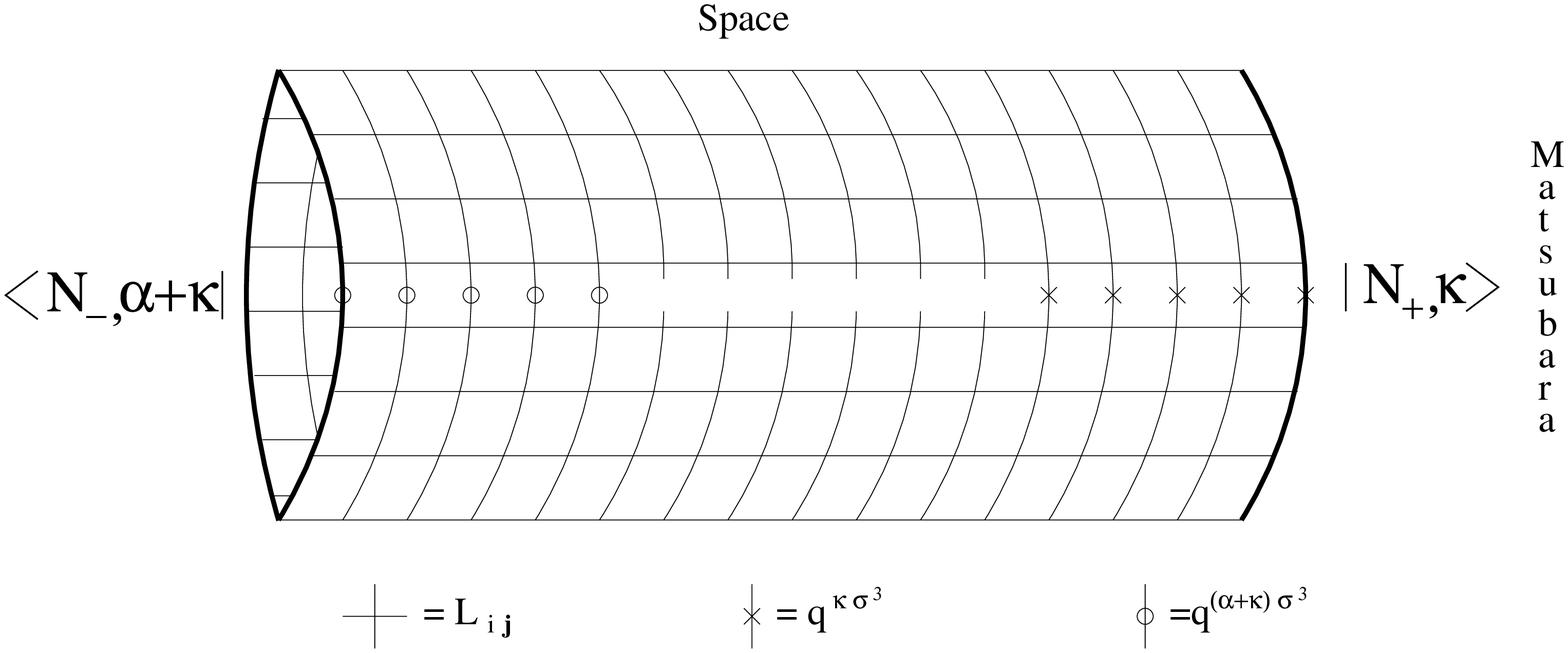}

\noindent
{\it Fig.3: Partition function of the six-vertex model on a cylinder.\\
To each site of the lattice is attached a local Boltzmann weight
$L_{j,\mathbf{m}}$. On one horizontal line 
a background field $q^{\kappa\sigma^3}$ is applied. 
On the same line, one allows in addition
a local dislocation in the middle 
accompanied by an extra field $q^{\alpha\sigma^3}$ extending to the left,  
representing the insertion of a quasi-local operator
$q^{2\al S(0)}\mathcal{O}$. 
}
\vskip .3cm
The number of sites in the Matsubara direction is denoted by
$\mathbf{n}$. The boundary conditions in the space direction
are given by the eigenvectors $\langle N_-,\kappa+\al|$ and
$|N_+,\kappa\rangle$ of the Matsubara transfer-matrices with 
twists $\kappa+\al $ and $\kappa$, 
and $N_\pm$ are the labels for different eigenvectors. In this formulation the number
of sites in the space direction is not important as far as the operator $\mathcal{O}$
fits into the picture {\it fig.3}. 
The same computation can be used in a
different situation.
Put arbitrary boundary conditions on the right end of the cylinder,
and allow it to grow to the right infinitely. Then 
the partition function will look for the 
eigenvector $|0,\kappa\rangle$ with the 
largest eigenvalue
of the Matsubara transfer-matrix (ground state).
Certainly one should not be extremely unlucky which means that the 
vector representing the boundary conditions should not be orthogonal 
{to} the ground state. 

The partition function
{\it fig.3}
is written as
$$\langle N_-,\kappa +\al|\Tr_{\rm S}
\(T_{{\rm S},\mathbf{M}}q^{2\kappa S+2\al S(0)}\mathcal{O}\)|N_+,\kappa
\rangle\,,$$
where ${\rm S}$
stands for space, $\mathbf{M}$ stands for Matsubara, $S$ is the total
spin in the space direction, $T_{{\rm S},\mathbf{M}}$ is the rectangular monodromy
matrix in the tensor product of the space and the
Matsubara Hilbert spaces. 

Let us make a digression. We discuss a connection to form factors.
Consider two quasi-local fields with opposite twists, $q^{-2\al S(0)}\mathcal{O}_1$
and $q^{2\al S(n)}\mathcal{O}_2$.
For simplicity we take $\kappa=0$ and denote corresponding ground state
by $|\mathrm{vac}\rangle$. In Section 8 we will see that in the infinite volume limit in the Matsubara direction,
the $\kappa$-dependence is dropped.
In order  to compute the correlation function
$$\langle\mathrm{vac}|{\rm Tr}_{\rm S}
\Bigl(T_{{\rm S},\mathbf{M}}q^{-2\al S(0)}\mathcal{O}_1q^{2\al S(n)}\mathcal{O}_2\Bigr)
|\mathrm{vac}\rangle$$
we have to glue two partition functions of this kind and sum over the
intermediate states as in {\it fig.4}:
\vskip .5cm
\hskip -0.1cm\includegraphics[height=3.5cm]{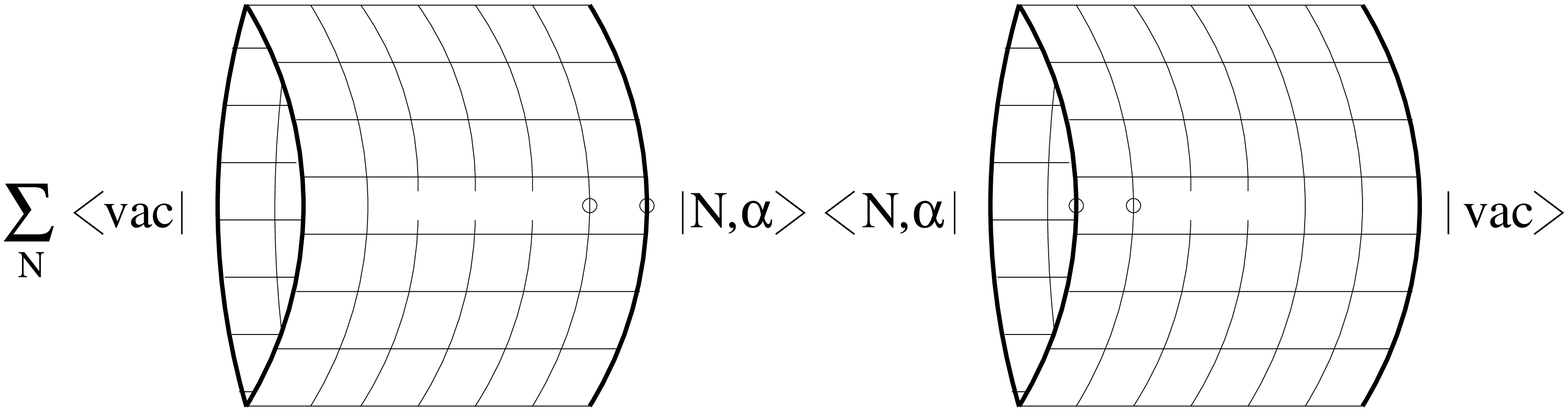}
\vskip .1cm
\noindent
{\it Fig.4: Form factor decomposition of the two-point function 
on the lattice.\\
}
\vskip .2cm
\noindent
This is nothing but the form factor decomposition in the Matsubara direction.
We will come back to this subject in Section \ref{EVOMEGA}.

In the paper \cite{HGSIII} we computed the functional
\begin{align}
Z^{\kappa}_\mathbf{n}\{q^{2\kappa S+2\al S(0)}\mathcal{O}  \}=\frac{
\langle N_-,\kappa +\al|\Tr_{S}
\(T_{S,\mathbf{M}}q^{2\kappa S+2\al S(0)}\mathcal{O}\)|N_+,\kappa
\rangle}
{
\langle N_-,\kappa +\al|\Tr_{S}
\(T_{S,\mathbf{M}}q^{2\kappa S+2\al S(0)}\)| N_+,\kappa
\rangle}\label{Z}
\,.
\end{align}
The dependence of $Z^{\kappa}_\mathbf{n}$ on the particular eigenvectors of the 
Matsubara transfer-matrices is not explicitly 
exhibited. 
Obviously \eqref{Z} is not a
form factor but rather a ratio of the
form factor of a descendant field 
to the form factor of the primary field. 
This should be remembered. 

The key objects to the computation of $Z^{\kappa}_\mathbf{n}$ were introduced in the
paper \cite{HGSII}. In this paper it was shown that the 
quasi-local operators $q^{2\al S(0)}
\mathcal{O}$ can be created from the primary field by action of one
bosonic and two fermionic operators: $\tb ^*(\z)$, $\bb ^*(\z)$, $\cb^*(\z)$. 
More precisely in the homogeneous case
the quasi-local operators $q^{2\al S(0)}\mathcal{O}$ are
created by $\tb ^*_p$, $\bb ^*_p$, $\cb^*_p$ defined through
\begin{align}
\tb ^*(\z)=\sum\limits _{p=1}^{\infty}(\z^2-1)^{p-1}\tb ^*_p\,,
\ \ \bb ^*(\z)=\sum\limits _{p=1}^{\infty}(\z^2-1)^{p-1}\bb ^*_p\,,
\ \ \cb ^*(\z)=\sum\limits _{p=1}^{\infty}(\z^2-1)^{p-1}\cb ^*_p\,.\label{tbchoh}
\end{align}
In particular $\tb^*_1/2$ is the {right} shift by one site along the lattice:
{
\begin{align}
\textstyle{\frac12}\tb^*_1\bigl(q^{2\al S(0)}\bigr)=q^{2\al S(1)}\,.\label{1site}
\end{align}
}
The completeness of this basis is proved in \cite{compl}.

The main theorem \cite{HGSIII} states that
\begin{align}
&Z^{\kappa}_\mathbf{n}\{\tb^*(\eta _1)\cdots \tb^*(\eta _m)
\bb^*(\z_1)\cdots \bb^*(\z_k)\cb ^*(\xi _k)\cdots \cb^*(\xi _1)\bigl(q^{2\al S(0)}\bigr)  \}
\label{mainIII}\\
&=\prod _{j=1}^m2\rho (\eta _j)
\cdot
\det \bigl(\omega (\z_i,\xi _j)
\bigr)_{i,j=1,\cdots k}\,,
\nn
\end{align}
where 
$$
\rho (\z)=\frac{T_\mathrm{L}(\z)}{T_\mathrm{R}(\z)}\,,$$
and $T_\mathrm{L}(\z)$
and 
$T_\mathrm{R}(\z)$
are respectively the eigenvalues 
of the left (with twist $\kappa+\al$) and the right (with twist $\kappa$) Matsubara transfer-matrices on the eigenvectors
$\langle N_-,\kappa+\al|$ and $|N_+,\kappa\rangle$.
The function $\omega (\z,\xi )$ was defined in \cite{HGSIII} through a number of
conditions as a quantum deformation of the normalised second
kind differential on a
hyper-elliptic Riemann surface. For computational purposes it is useful to consider an 
equivalent, alternative formula \cite{BG} 
which is explained below.

Consider the eigenvector  $|N_+,\kappa\rangle$.
The corresponding eigenvalue is defined by the function
$\mathfrak{a}(\z)$
which satisfies 
the DDV equation 
\begin{align}
\log \mathfrak{a}(\z)=-2\pi i\nu\kappa 
+\log \( 
\frac{d(\z)}
{a(\z)}   \)
-\int\limits _\gamma K(\z/\xi)\log \(1+ \frak{a}(\xi)\)
\frac {d\xi^2}{\xi ^2}\,,
\label{DDV}
\end{align}
where the cycle $\gamma$ goes around the zeros of the eigenvalue 
of the Baxter operator
$Q_\mathrm{R}(\z)$ (Bethe roots)
in the {\it clockwise} direction, 
as opposed to all other contours. The functions $a(\z)$ and $d(\z)$
are defined by using inhomogeneous parameters in 
the Matsubara direction.
The formulae for them can be found in \cite{HGSIII}. 
We would like to emphasise that $\frak{a}(\xi)$ depends only on the eigenvalue
of the right transfer-matrix.

For completeness let us write the Baxter equation, and
the definition of $\mathfrak{a}(\z)$
\begin{align}
&T_\mathrm{R}(\z)Q_\mathrm{R}(\z)=d(\z)Q_\mathrm{R}(\z q)+a(\z)Q_\mathrm{R}(\z q^{-1})\,, \quad
\mathfrak{a}(\z)
=\frac{d(\z)Q_\mathrm{R}(\z q)}{a(\z)Q_\mathrm{R}(\z q^{-1})}\,.\nn
\end{align}
Recall also that 
$\z ^{\kappa}Q_\mathrm{R}(\z)$ is a polynomial in $\zeta^2$.
The kernel in the integral equation is defined through
a more general one
$$
K(\z)=K(\z,0)\,,
\quad K(\z,\al)=\frac 1{2\pi i}\Delta _{\z}\psi_0 (\z,\al)\nn\,,$$
where we use the symbols
$$\Delta _\z f(\z)=f(\z q)-f(\z q^{-1})\,,\quad \psi _0(\z,\al)= \z ^{\al}\frac{1}{\z ^2-1}\,.$$
In the papers \cite{HGSIII,HGSIV} we used instead of $\psi_0(\z,\al)$ the
function 
$\psi (\z,\al)=\psi_0(\z,\al)-\frac 1 2 \z ^{\al}$. As 
it has been shown in \cite{BG} on the finite lattice, 
either of these two functions
can be used in the equations defining $\omega(\z,\xi)$ without
changing the final result. So, we can choose the one which is more appropriate
for the scaling limit. 

In {order} to define $\omega (\z,\xi)$ we have to introduce 
the operation $\delta ^-_\z$
$$\delta ^-_\z f(\z)=f(\z q)-\rho(\z)
f(\z)\,,$$
and the convolution
$$A\star B(\z,\xi)= \int\limits _\gamma A(\z,\eta)B(\eta,\xi)dm(\eta)\,,$$
with measure
\begin{align}dm(\eta )=\frac {d\eta ^2}{ \eta ^2\rho (\eta)
\(1+\frak{a}(\eta)\)}\,.\label{measure}
\end{align}
The fundamental role is played by the dressed resolvent
\begin{align}
\Rdr
- \Rdr       \star     K_{\al}
=K_{\al}\,,
\label{eqR}
\end{align}
where $K_{\al}$ stands for the integral operator with 
the kernel $K(\z/\xi,\al)$.
Introducing two more kernels
\begin{align}
\quad f_\mathrm{left} (\z,\xi)=\textstyle{\frac 1 {2\pi i}\ }\delta ^-_\z\psi_0(\z/\xi,\al),\quad 
f_\mathrm{right}(\z,\xi)=\delta ^-_\xi \psi_0 (\z/\xi,\al)\,,
\label{flr}
\end{align}
we define  $\omega(\z,\xi)$ by
\begin{align}
\textstyle{\frac 1 4}\omega (\z,\xi)=
&\(f_\mathrm{left}\star   f_\mathrm{right}+   f_\mathrm{left}\star           
\Rdr
\star f_\mathrm{right}\)(\z,\xi)-\omega _0(\z,\xi)
 \label{defomega}\end{align}
where
\begin{align}
\omega _0(\z,\xi)
=-\delta ^-_\z\delta ^-_\xi\Delta ^{-1}_\z\psi_0(\z/\xi,\al)
\,.\label{omega0}
\end{align}
The definition of the``primitive function" $\Delta^{-1}$
is ambiguous.
Using $Z=\zeta^{1/\nu}$, $H=\eta^{1/\nu}$, $X=\xi^{1/\nu}$, 
we shall define once and forever \cite{HGSIV}
\begin{align}
\Delta ^{-1}_\z\psi_0 (\z/\xi,\al)&
=\frac1{2\nu}VP\int\limits _0^{\infty}\psi_0 (\z/\eta,\al)\frac H{H+X}
\frac{d\eta ^2}{2\pi i\eta ^2}\,,
\end{align}
where the principal value is taken 
with regards to the pole {at $\eta^2=\zeta^2$.}
By this definition, we have
\begin{align}
\Delta^{-1}_\zeta\psi_0(q^{\pm1}\zeta/\xi,\al)=\pm\frac1{4\nu}\frac Z{Z-X}+
\frac1{2\nu}
\int\limits _0^{\infty}\psi_0 (\z e^{\pm i0}/\eta,\al)\frac H{H-Xe^{\pm i0}}\frac{d\eta ^2}{2\pi i\eta ^2}.
\label{ANALYTIC CONTINUATION}
\end{align}
Here the left hand side is the analytic continuation
of $\Delta ^{-1}_\z\psi_0 (\z,\al)$  in the variable $\frac1\nu\log\zeta$ from ${\rm Im}\frac1\nu\log\zeta=0$ to
${\rm Im}\frac1\nu\log(q^{\pm1}\zeta)=\pm\pi i$.

A physically
interesting situation occurs when we take the scaling limit
which implies in particular $\mathbf{n}\to\infty$. Here we consider two cases.
We shall be very brief because detailed explanations are given in \cite{HGSV}.

\vskip .2cm
\noindent
1. Chiral CFT.
\vskip .2cm
\noindent
In the homogeneous case $\z_j=1$, $\tau_{\mathbf{m}}=1$,
the scaling limit 
\begin{align}
\mathbf{n}\to\infty,\quad a\to 0,\quad \mathbf{n}a=2\pi R,\quad \z=\la (Ca)^{\nu}\,,
\label{rescCFT}\end{align}
($R$ and $\la $ are  finite)
describes the chiral conformal field theory \cite{HGSIV}. 
The important
constant $C$ is given by
$$C=\frac{\Gamma 
\(\frac {1-\nu}{2\nu}\)}{2\sqrt{\pi}\Gamma 
\(\frac {1}{2\nu}\)}\Gamma (\nu)^{\frac 1 \nu}\,.
$$
\vskip .2cm
\noindent
2. Sine-Gordon model.

\vskip .2cm
\noindent
In the inhomogeneous case 
\begin{align}\z_j=\z_0^{(-1)^j},\quad \tau _\mathbf{m}=\z _0^{(-1)^\mathbf{m}}\,,
\label{inhom}
\end{align}
the scaling limit 
\begin{align}\mathbf{n}\to\infty,\ \ a\to 0,\ \ \mathbf{n}a=2\pi R, \ \ \z_0\to \infty,
\ \ \mub=\z _0^{-1}(C a)^{-\nu}\,,\label{rescsG}\end{align}
($R$ and $\mub$ are finite) we describes  the sG model on a cylinder of radius $R$
and the coupling constant $\mub$ (see \eqref{action})..

Let us briefly recall the scaling limit of the operators $\tb^*(\z)$,
$\bb ^*(\z)$ and $\cb^*(\z)$ in the inhomogeneous case. The construction is not trivial,
so the interested reader is referred to the
papers \cite{OP,HGSV} for details.
The local operators in the inhomogeneous case
are created by coefficients of $\bb^*(\z)$, $\cb^*(\z)$ developed 
into seria around $\z^2=\z_0^2$
and $\z ^2=\z _0^{-2}$. {From} these seria,
after certain Bogolubov transformation, we obtain the
operators $\bb^{+*}(\z)$, $\cb^{+*}(\z)$, $\bb^{-*}(\z)$, $\cb^{-*}(\z)$.
Their expectation values are  the determinants \eqref{mainIII}
wherein the ``two-point correlators" given by
\begin{align}
&Z_\mathbf{n}\{\bb ^{+*}(\z)\cb ^{+*}(\xi)
{\bigl(q^{2\al S(0)}\bigr)}\}=\omega 
(\z,\xi)\,,\label{latpar}\\
&Z_\mathbf{n}\{\bb ^{+*}(\z)\cb ^{-*}(\xi)
{\bigl(q^{2\al S(0)}\bigr)}\}=\omega 
(\z,\xi)
+\omega _0(\z,\xi)
\,,\nn\\
&Z_\mathbf{n}\{\bb ^{-*}(\z)\cb ^{+*}(\xi)
{\bigl(q^{2\al S(0)}\bigr)}\}=\omega 
(\z,\xi)
+\omega _0(\z,\xi)
\,,\nn\\
&Z_\mathbf{n}\{\bb ^{-*}(\z)\cb ^{-*}(\xi)
{\bigl(q^{2\al S(0)}\bigr)}\}=\omega 
(\z,\xi)
\,.\nn
\end{align}
The scaling limit is dictated by the behaviour of the functions
$\rho $, $\omega$ and $\omega _0(\z,\xi)$.
It can be shown
(see Appendix A of \cite{HGSIV} for very similar conclusions) that
in a very general setting 
\begin{align}
&\lim_\mathrm{scaling}\log \rho (\z)\ \ \ \ \simeq \hskip -1cm\raisebox{-7pt}{${}_
 {\log\z\to \epsilon\infty}$}
\ \ \sum_{j=1}^\infty \rho _{\epsilon(2j-1)}\z ^{-\epsilon \frac{2j-1}\nu}\,,\label{asrhoomega}\\
&\lim_\mathrm{scaling}\frac 1 {\sqrt{\rho (\z)}\sqrt{\rho (\xi)}}
\ \omega (\z,\xi)
\ \ \ \  \simeq \hskip -1cm\raisebox{-14pt}{$
 {{\log\z\to \epsilon\infty}\atop{\ \log\xi\to \epsilon'\infty}}$}
\ \sum\limits_{j,k=1}^{\infty}\omega _{\epsilon (2j-1),\epsilon '(2k-1)}
\z ^{-\epsilon \frac{2j-1}\nu}\xi ^{-\epsilon' \frac{2k-1}\nu}\,,\nn
\end{align}
where $\epsilon,\epsilon '=\pm$. It remains to consider $\omega _0(\z,\xi)$.
Recall the definition \eqref{omega0}. Since we know
the asymptotics of $\rho(\z)$ the only non-trivial thing is the
asymptotics of $\Delta ^{-1}_\z\psi _0(\z,\al)$. 
For studying the asymptotics
it is convenient to rewrite this function in the form of 
a Mellin transform:
\begin{align}
&\Delta ^{-1}_\z\psi _0(\z,\al)
={\frac i4}\int\limits _{-\infty}^{\infty}\z^{2ik}\
\coth\pi(k+{\textstyle\frac{i\al} 2})\frac 1 {\sinh 2\pi\nu(k-i0)}dk\label{asDinv}
\\ &\simeq \hskip -.6cm\raisebox{-7pt}{${}_ {\z\to\infty}$}
-{\textstyle\frac i{4\nu}}\sum_{j=0}^{\infty}(-1)^j\z ^{-\frac j {\nu}}\cot {\textstyle\frac \pi{2\nu}}
(\al\nu+j) -{\textstyle\frac i2}\sum_{j=1}^{\infty}\z ^{\al-2j}
\frac 1 {\sin \pi\nu(\al-2j)}\nn\,,\\
&\simeq \hskip -.6cm\raisebox{-7pt}{${}_ {\z\to0}$}
{\textstyle\frac i{4\nu}}\sum_{j=1}^{\infty}(-1)^j\z ^{\frac j {\nu}}\cot {\textstyle\frac \pi{2\nu}}
(\al\nu-j) +{\textstyle\frac i2}\sum_{j=0}^{\infty}\z ^{\al+2j}
\frac 1 {\sin \pi\nu(\al+2j)}.\nn
\end{align}

Motivated by the
above formulae we have conjectured \cite{OP,HGSV} that the following scaling limit 
exists for the operators $\bb ^{\pm*}(\z)$, $\cb ^{\pm*}(\z)$:
\begin{align*}
&\half\bb ^{+*}(\z)\ \ \longrightarrow
\hskip -1cm {}_{{\ }_{\mathrm{scaling}}}
\ \ \betab^{+*}(\z)\ \ \simeq \hskip -.6cm\raisebox{-7pt}{${}_
 {\z\to\infty}$}
\ \ \betab^*(\mub\z)+\bigl(1+O(\{\mathbf{i}_*\})\bigr)\bar{\betab}^*_\mathrm{screen}
(\z/\mub)
\,,
\\&
\half\cb ^{+*}(\z)   \ \ \longrightarrow
\hskip -1cm {}_{{\ }_{\mathrm{scaling}}}
\ \ \gammab^{+*}(\z)\ \ \simeq \hskip -.6cm\raisebox{-7pt}{${}_
 {\z\to\infty}$}
\ \ 
\gammab^*(\mub\z)+\bigl(1+O(\{\mathbf{i}_*\})\bigr)
\bar{\gammab}^*_\mathrm{screen}
(\z/\mub)
\,,\\
&\half\bb ^{-*}(\z)\ \ \longrightarrow
\hskip -1cm {}_{{\ }_{\mathrm{scaling}}}
\ \ \betab^{-*}(\z)
\ \ \simeq \hskip -.6cm\raisebox{-7pt}{${}_
 {\z\to 0}$}
\ \ 
\bar{\betab}^*(\z/\mub)+\bigl(1+O(\{\bar{\mathbf{i}}_*\})\bigr)
{\betab}^*_\mathrm{screen}
(\mub \z)
\,,\\
&\half\cb ^{-*}(\z)\ \ \longrightarrow
\hskip -1cm {}_{{\ }_{\mathrm{scaling}}}
\ \ \gammab^{-*}(\z)
\ \ \simeq \hskip -.6cm\raisebox{-7pt}{${}_
 {\z\to 0}$}
\ \ 
\bar{\gammab}^*(\z/\mub)+
\bigl(1+O(\{\bar{\mathbf{i}}_*\})\bigr)\gammab^*_\mathrm{screen}
(\mub \z)
\,,
\end{align*}
where $O(\{{\mathbf{i}}_*\})$, $O(\{\bar{\mathbf{i}}_*\})$ stand for descendants created
by the local integrals of motion. We have to apologise for forgetting
to write these terms in \cite{HGSV}. There are two explanations for that. First, in 
\cite{HGSV} we considered the one-point functions for which the descendants by
the local integrals of motion are irrelevant. Second, most important, these terms are
absent in the asymptotics for the chiral CFT which serves to normalise the fermionic
operators. Their appearance is an artefact of mixing two chiralities, and we never use them.

The operators in the right hand side are the
unique \cite{OP,HGSV} operators
in the
sG model which provide deformation of the CFT operators
\eqref{betaCFT}, \eqref{screenCFT}. The prescription for powers of the spectral
parameter are the same: they are dictated by the scaling dimensions. Looking at
the asymptotics \eqref{asrhoomega}, \eqref{asDinv} one concludes that the pairings
of the screening operators come only  {from} $\omega_0(\z,\xi)$. Moreover,
these pairings are diagonal:  $\betab^*_{\mathrm{screen},j} $ couples only with
$\bar{\gammab}^*_{\mathrm{screen},j} $ and $\gammab^*_{\mathrm{screen},j} $ couples only with
$\bar{\betab}^*_{\mathrm{screen},j} $. 
The corresponding
pairings are easy
to compute, and
as a result we come to the conclusion that the screening
operators can be ignored provided we postulate the formula \eqref{Phi=Phi}.

An important generalisation of the above construction 
was proposed in \cite{HGSIV,HGSV}.
We have three parameters: 
$\al$, the twist $\kappa +\al$ for the left 
Matsubara transfer-matrix, and 
$\kappa$ for the right Matsubara transfer-matrix.
By introducing the screening operators in the 
lattice construction it was shown that
in the scaling limit we can achieve the emancipation of $\kappa+\al $, namely,
it can be replaced by an arbitrary parameter $\kappa '$. 
The basic formula \eqref{mainIII}
and the definitions of $\rho(\z)$ and $\omega(\z,\xi)$
remain valid in this case.
A particularly
nice situation  occurs  for $\kappa '=\kappa$ which implies $\rho(\z)=1$.
This case is physically relevant because it describes 
the one-point functions on the
cylinder for the sG model \cite{HGSV}. 
But also this case is simpler from the
technical point of view, and
it allowed the quantitative investigation 
of the equation for $\omega(\z,\xi)$ in the CFT case
\cite{HGSIV,Boos}. 
The price to pay for this simplification is that we can
identify the action of descendants created by 
fermions with those created by
the Virasoro generators only modulo 
the action of the local integrals of motion.
It would be important to remove this 
technical obstacle. In Section \ref{EVOMEGA}
we shall see the first example of an exact 
solution for $\omega(\z,\xi)$ at $\rho(\z)\ne1$.

\section{Infinite volume in Matsubara direction.}\label{infinite}

Consider the limit $R\to\infty$ which corresponds
to the sG model on the plane. 
In this case the scaling limit and the limit of the infinite volume 
for the lattice model commute. 
In other words the correct result for the functions $\rho$ and $\omega$ 
will be achieved if we consider the limit $\mathbf{n}\to \infty$
directly in the six-vertex case without the rescaling \eqref{rescCFT} or
\eqref{rescsG}.
This is so well known that we have a problem with  
making proper references. However,
there is one less known point which we would like to underline in this section. 

The results of this section is used for both left and right transfer-matrices, so, we
shall denote the twist by $\theta$, assuming that it is real.
In the limit $\mathbf{n}\to \infty$, the
zeros of $Q(\z) $ densely fill $\mathbb{R}_+$
in the plane of $\z ^2$.
Obviously the solutions to the
equation 
\begin{align} 
\mathfrak{a}(\z)+1=0\,,\label{a+1=0}
\end{align}
are the zeros either
of $Q(\z)$ or of $T(\z) $. For the ground
state all the real positive solutions to the equation \eqref{a+1=0} 
are the zeros
of $Q(\z) $. For the excited state a finite number of real positive solutions
are the
zeros of $T(\z)$. 
They are conventionally called holes, and correspond
to solitons in the sG language. 
In the sector of zero spin the number of holes is even (this is our $2n$).
It is necessary
to introduce the same number  of complex zeros of $Q(\z)$ in order
to compensate the spin. 
The positions of these complex zeros are defined by 
the positions of the 
holes through the Higher Level Bethe Ansatz (HLBA) equations
\cite{DestriLow}. The HLBA equations have $\binom{2n}{n}$
solutions which count different isotopic structures for solitons-antisolitons.
One more property of these complex solutions is that they do not
contribute to the local integrals of motion taking care only of
the isotopic structure. We shall discuss this in more details in Section \ref{equivalence}.
There may be an
additional number of complex roots
organised into strings. These strings describe breathers in the sG language, and
we shall not consider them. 

Changing the cycle, the DDV equation \eqref{DDV} reads as follows
\begin{align}
\log \frak{a}(\z)=
\log \( 
\frac{d(\z)}
{a(\z)}   \)-2\pi i\nu\theta -\sum _{h}\Phi (\z/\xi_h)+\sum_{c}\Phi (\z/\xi _c)\label{DDV1}\\
-\int\limits _{\gamma_0} K(\z/\xi)\log \(1+ \frak{a}(\xi)\)
\frac {d\xi^2}{\xi ^2}\,,
\nn
\end{align}
where 
$$\Phi(\z)=\log\(\frac{1-q^2\z^2}{1-q^{-2}\z^2}\)\,,$$
and
$\gamma_0$ goes clockwise around $\mathbb{R}_+$.
The first term in the right hand side is of order $\mathbf{n}$. This implies
that $\log \frak{a}(\z)$ is of order $\mathbf{n}$. 
Standard analysis shows that
\begin{align}&\mathrm{Re}(\log \frak{a}(\z))>0,\quad \mathrm{Im}\ \z^2<0,\label{beh}\\
&\mathrm{Re}(\log \frak{a}(\z))<0,\quad 
\mathrm{Im}\ \z^2>0\,.\nn\end{align}
Hence for large $\mathbf{n}$ with 
exponential precision in $\mathbf{n}$,
only the lower part of the contour $\gamma _0$ contributes and
we come to the linear equation:
\begin{align}
\log \frak{a}(\z)&-\int\limits _{0}^\infty K(\z/\xi)\log \frak{a}(\xi)
\frac {d\xi^2}{\xi ^2}\label{DDV2}
\\&=\log \( 
\frac{d(\z)}
{a(\z)}   \)-2\pi i\nu\theta -\sum _{h}\Phi (\z/\xi_h)+\sum_{c}\Phi (\z/\xi _c)
\,.
\nn
\end{align}
We split $\log \frak{a}(\z)$ according to four terms in the right hand side:
$$\log \frak{a}(\z)=\mathbf{n}F_\mathrm{vac}(\z)+F_\theta(\z)+
F_h(\z)+F_c(\z)\,.$$
The first term is model dependent which means that it depends on the
inhomogeneous parameters. But it does not contribute to physically relevant quantities.
For example the $S$-matrices for the
XXZ model and for the sG model coincide.
The contributions $F_h(\z)$ and $F_c(\z)$ are 
well known, so
we shall not
write the corresponding formulae. 
For $F_\theta(\z)$ one finds immediately a
$\z$-independent answer:
$$F_\theta=-\pi i\theta{\textstyle \frac{\nu}{1-\nu}}\,.$$

Let us consider the eigenvalue $T(\z)$. Due to \eqref{beh} 
with exponential precision in $\mathbf{n}$,
for $0<\arg(\z^2)<\pi$ we have
\begin{align}
&\log T(\z)=\log(a(\z))+\pi i\nu \theta-\frac 1 {2\pi i}\int
\limits_{\gamma}\Psi (\z/\xi)
d\log (1+\mathfrak{a}(\xi))\nn\\
&=\log(a(\z))+\pi i\nu \theta-\sum_h\Psi(\z/\xi_h)+\sum _c\Psi(\z/\xi _c)\nn\\
&-\frac 1 {2\pi i}\int_{0}^{\infty}\frac {d}{d\xi^2}\Psi(\z/\xi)\ \log (\mathfrak{a}(\xi))d\xi ^2
+O(e^{-c\mathbf{n}})\,,\nn
\end{align}
where
$$\Psi (\z)=\log\(\frac{\z^2q^{-2}-1}{\z^2-1}\)\,.$$
It is easy to see that
$$\pi i\nu\theta-F_\theta\cdot \frac 1 {2\pi i}\int_{0}^{\infty}\frac {d}{d\xi^2}\Psi(\z/\xi)d\xi ^2=0\,,$$
so we come to an
important conclusion: in the limit $\mathbf{n}\to \infty$ with
exponential precision in $\mathbf{n}$,
$\log T(\z)$ is independent
of 
the twist $\theta$.
Finally a standard computation yields
\begin{align}
T(\z)=T_\mathrm{vac}(\z)\prod_h\frac{\z^{\frac 1 \nu}-\xi_h^{\frac 1 \nu}}
{\z^{\frac 1 \nu}+\xi_h^{\frac 1 \nu}}\,,\label{Tfinal}
\end{align}
where $T_\mathrm{vac}(\z)$ is independent of $\theta$.
The complex zeros $\xi _c$ do not contribute. Returning to
the picture {\it fig.4}
we see that the dependence on $\al$ in the
intermediate states is dropped in the limit $\mathbf{n}\to\infty$.
The $\al$-dependence remains in the quasi-local operators
$q^{-2\al S(0)}\mathcal{O}_1$ and $q^{2\al S(n)}\mathcal{O}_2$.

On the other hand it is clear that in this limit  {\it fig.4}  represents just the form factor decomposition on the plane.
Further, in the scaling limit the 
rotational symmetry occurs and we obtain the usual form factor decomposition for the sG model.

\section{Evaluation of $\omega(\z,\xi)$ in {the} presence of solitons}\label{EVOMEGA}

We want to evaluate the function $\omega (\z,\xi)$ which corresponds
to the right part of  {\it fig.4}. In that case the $T_\mathrm{L}(\z)$ is given by \eqref{Tfinal}
and $T_\mathrm{R}(\z)=T_\mathrm{vac}(\z)$. Hence
\begin{align}
\rho(\z)=\prod_{j=1}^{2n}\frac{\z^{\frac 1 \nu}-\xi_j^{\frac 1 \nu}}
{\z^{\frac 1 \nu}+\xi_j^{\frac 1 \nu}}\,,\label{rhoinf}
\end{align}
where the number of holes is $2n$. 
We shall use other variables
$${\zeta^{\frac 1 \nu}=Z},\quad\quad \xi_j^{\frac 1 \nu}=B_j=e^{\beta _j}\,.$$
So, in our usual notation
$$\rho(\zeta)=\frac{P(Z)}{P(-Z)}\,.$$
Later we shall use also $\eta^{\frac 1 \nu}=H$, $\xi^{\frac 1 \nu}=X$.

Usually the function $\omega (\z,\xi)$ is found in two steps. First, we solve the
equation \eqref{eqR} for $R_\mathrm{dress}$ and then substitute the
result into \eqref{defomega} to find $\omega (\z,\xi)$. However, this procedure
relies heavily on the assumption of simplicity of the spectrum of $T_\mathrm{R}(\z)$.
In the case $\mathbf{n}\to \infty$ this assumption is not true, the eigenvalue
\eqref{Tfinal} corresponds {to $\binom{2n}{n}$ vectors in the sector of total spin $0$}.
That is why the {integral} operator in \eqref{eqR} is degenerate and 
$R_\mathrm{dress}$ is not well-defined. To avoid this problem we
proceed in the same way as \cite{BG} introducing the function
$G (\z,\xi)$ which satisfies the equation
\begin{align}
G(\z,\xi)=\delta ^-_\xi{\psi_0}(\z/\xi,\al)+\frac 1 {2\pi i}\int\limits_{\mathbb{R}_+e^{+i0}}
\({\psi_0}(q\z/\eta,\al)-{\psi_0}(q^{-1}\z/\eta,\al)\)G(\eta,\xi)\frac{d\eta ^2}{\eta^2\rho(\eta)}\,.
\label{eqG}
\end{align}
The shift $e^{+i0}$ in \eqref{eqG} refers to the poles of $1/\rho (\eta)$.
The term  $\delta ^-_\xi{\psi_0}(\z/\xi,\al)$ as function of $\z^2$
has poles at $\z^2=\xi^2$ and $\z^2=\xi ^2q^{2}$. Let us choose such $\xi$
that these poles do not lie on $\mathbb{R}_+$. Then the singular integral equation \eqref{eqG}
is well-defined being supplemented with the requirement that
$G(\z,\xi)$ is regular for $\z\in\mathbb{R}_+$. 

Once $G(\z,\xi)$ is found the function $\omega(\z,\xi)$ is defined by
\begin{align}
&\omega(\z,\xi)=
{\delta ^-_\z\delta ^-_\xi
\Delta ^{-1}_\z{\psi_0}(\z/\xi,\al)+
\frac 1 {2\pi i}\int\limits _{\mathbb{R}_+e^{+i0}}\delta ^-_\z{\psi_0}(\z/\eta,\al)
G(\eta,\xi)\frac {d\eta ^2}{\eta^2\rho(\eta)}}\,.\label{omega}
\end{align}
Writing down these equations we used \eqref{beh}.

The form factors were described in terms of pairings $(\ell^{(n)},L^{(n)})_\al$.
Now we want to solve similarly the equation \eqref{eqG}. We formulate the
result as Proposition, {giving a necessary explanation
for the precise definition of $G(\z,\xi)$ in the proof.}
\begin{prop}
Recall that
$$
M^{(n)}_0(S_1,\cdots,S_n)=\langle\Phi _\al\rangle\prod_{p=1}^nS_p
\prod_{p>r}(S_p^2-S_r^2).
$$
Consider the rational function
\begin{align}
R^{(n)}_{Z,X}(S_1,\cdots,S_n)=-\frac 1 {\nu}\frac {ZX}{Z^2-X^2}\frac {P(Z)}{P(-X)}
\prod_{p=1}^n\frac{X^2-S_p^2}{Z^2-S_p^2}
{
M^{(n)}_0(S_1,\ldots,S_n).
}
\label{RS}
\end{align}
Then for any {$\ell^{(n)}$} the function
\begin{align}
G(\z,\xi;\ell^{(n)})=\frac {(\ell^{(n)},R_{Z,X}^{(n)})_\al} {(\ell^{(n)},{M_0^{(n)}})_\al}\label{G}
\end{align}
solves the equation \eqref{eqG}. 
\end{prop}
\begin{proof}

First of all we have to explain how the pairing with the rational
function $R_{Z,X}(S)^{(n)}$ is understood. 
For polynomials {$\ell(\mathfrak{s})$, $L(S)$} we defined the pairing as
\begin{align}
{(\ell,L)_\al}=
\int\limits _{\mathbb{R}-i0}
\chi (\sigma 
)e^{\frac {\al\nu}{1-\nu}\sigma}
\frac{m(\mathfrak{s})}{p(\mathfrak{s}\mathfrak{q}^{-2})}L(S)d\sigma+
\int\limits _{\Gamma}
\chi (\sigma )e^{\frac {\al\nu}{1-\nu}\sigma}n(\mathfrak{s}){L(S)}d\sigma\,,\nn
\end{align}
where
\begin{align}
p(\mathfrak{s}\mathfrak{q}^{-2})\ell(\mathfrak{s})=
m(\mathfrak{s})+a^{-2}p(\mathfrak{s})
n(\mathfrak{s}\mathfrak{q}^{-4})-
p(\mathfrak{s}\mathfrak{q}^{-2})n
(\mathfrak{s})\,,\label{arb}
\end{align}
and $m(\mathfrak{s})$, $n(\mathfrak{s})$
are chosen from  the requirement of convergence. Now for every $S_p$
the rational function $R_{Z,X}^{(n)}$ has
poles at $S_p^2=Z^2$. Admitting the existence of these poles
we define the pairing $(\ell^{(n)},R_{Z,X}^{(n)})_\al$ by
drawing the contours of integration as follows.
\vskip .5cm
\hskip 2cm\includegraphics[height=8cm]{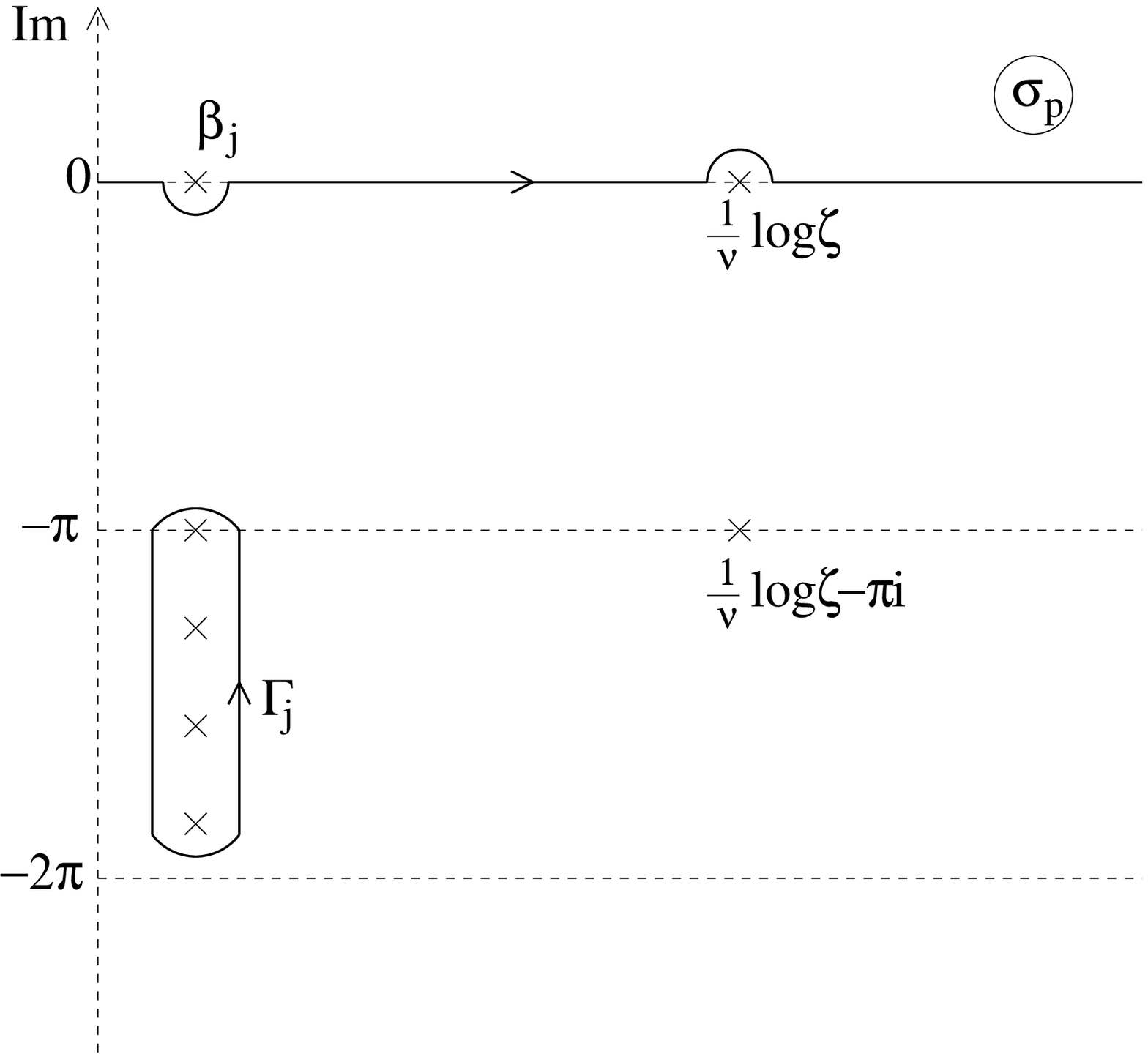}

\noindent
{\it Fig.6: Contour of integration for the pairing. \\
The presence of the poles at $S^2=Z^2$ makes the pairing dependent on
the choice of the functions $m(\mathfrak{s})$,  $n(\mathfrak{s})$. 
}
\vskip .2cm
We immediately see that there is a trouble with this definition. 
The point is that the  polynomials $m(\mathfrak{s})$ and $n(\mathfrak{s})$ are defined
from  \eqref{arb} not uniquely. In spite of this arbitrariness the
pairing was defined uniquely for polynomials {$L(S)$ }because we could {consider} the
difference between two solutions to \eqref{arb} and cancel the integrals
moving the contours. For {rational functions} with poles at $S^2=Z^2$
this is not true because moving the contours we pick up contributions
from the poles 
({\it fig. 6}). 
Thus the {pairing depends}
on the choice of $m(\mathfrak{s})$ and $n(\mathfrak{s})$.
Still we shall show that the equation \eqref{eqG} is satisfied for
any choice of $m(\mathfrak{s})$ and $n(\mathfrak{s})$ and, hence, the
{difference solves} the homogeneous equation. 

Before doing the computation of the integral in \eqref{eqG}, we give some remarks.
First, it is very helpful for understanding of what is going on {to} check the last statement in the previous paragraph
directly for $n=1$ when we have only one integral over $\sigma$.
Second, we shall see that the arbitrariness in question does not concern the 
function $\omega (\z,\xi;\ell^{(n)})$ which is the goal of our computation. The function
$G(\z,\xi;\ell^{(n)})$ is an auxiliary object. It may depend on the choice of regularisation of the integrals, but
the main object $\omega(\z,\xi;\ell^{(n)})$ does not. This will be proved at the end of the section. The last remark is regularity of $G(\z,\xi;\ell^{(n)})$ at 
$\z^2=\xi_j^2$. It clearly follows
from 
{\it fig. 6}
that the integral has singularities
at ${\zeta}=\xi _j$. They are simple 
poles, and $P(Z)$ cancels them, so, altogether
$G(\z,\xi;\ell^{(n)})$ is regular at these points.
 
We compute
\begin{align}
&\frac 1 {2\pi i}\int\limits_{\mathbb{R}_+}
\({\psi_0}(q\z/\eta,\al)-{\psi_0}(q^{-1}\z/\eta,\al)\)R^{(n)}_{H,X}(S_1,\cdots,S_n)\frac{d\eta ^2}{\eta^2\rho(\eta)}\label{compR}\\
&=
\frac 1 {2\pi i}\(\int\limits_{\mathbb{R}_+{q^{-2}}}
-\int\limits_{\mathbb{R}_+{q^2}}\)
{\psi_0}(\z/\eta,\al)R^{(n)}_{-H,X}(S_1,\cdots,S_n)
\frac{P(H)}{P(-H)}\frac{d\eta ^2}{\eta^2}\nn\\
&=
-\frac 1 {2\pi i}\(\int\limits_{\mathbb{R}_+{q^{-2}}}
-\int\limits_{\mathbb{R}_+{q^2}}\)
{\psi_0}(\z/\eta,\al)R^{(n)}_{H,X}(S_1,\cdots,S_n)
\frac{d\eta ^2}{\eta^2}\nn\\
&=R^{(n)}_{{Z},X}(S_1,\cdots,S_n)-\delta ^-_\xi{\psi_0}(\z/\xi,\al)
{M^{(n)}_0}(S_1,\cdots, S_n)
\nn\\
&-\frac{X}{P(-X)}{M^{(n)}_0}(S_1,\cdots, S_n)
\sum_{p=1}^n{S_p^{-1}\prod_{r\ne p}}\frac{X^2-S_r^2}{S_p^2-S_r^2}\nn\\
&\times 
\left\{
\begin{matrix}
\hskip -2cm{\psi_0}(\z S_p^{-\nu},\al)P(S_p)-A
{\psi_0}(\z (S_pQ)^{-\nu},\al)P(-S_p)\,,&\quad
\sigma _p\in \mathbb{R}\,,\\ \ \\
{\psi_0}(\z S_p^{-\nu}e^{-2\pi i\nu},\al)P(S_p)-A
{\psi_0}(\z (S_pQ)^{-\nu}e^{-2\pi i\nu},\al)P(-S_p)\,,
&\quad \sigma _p\in \Gamma\,,\nn
\end{matrix}
\right.
\end{align}
where we used a trivial identity:
$$\psi_0 (\z e^{\pi i},\al)=A{\psi_0} (\z,\al)\,.$$
\vskip .5cm
\hskip 2cm\includegraphics[height=8cm]{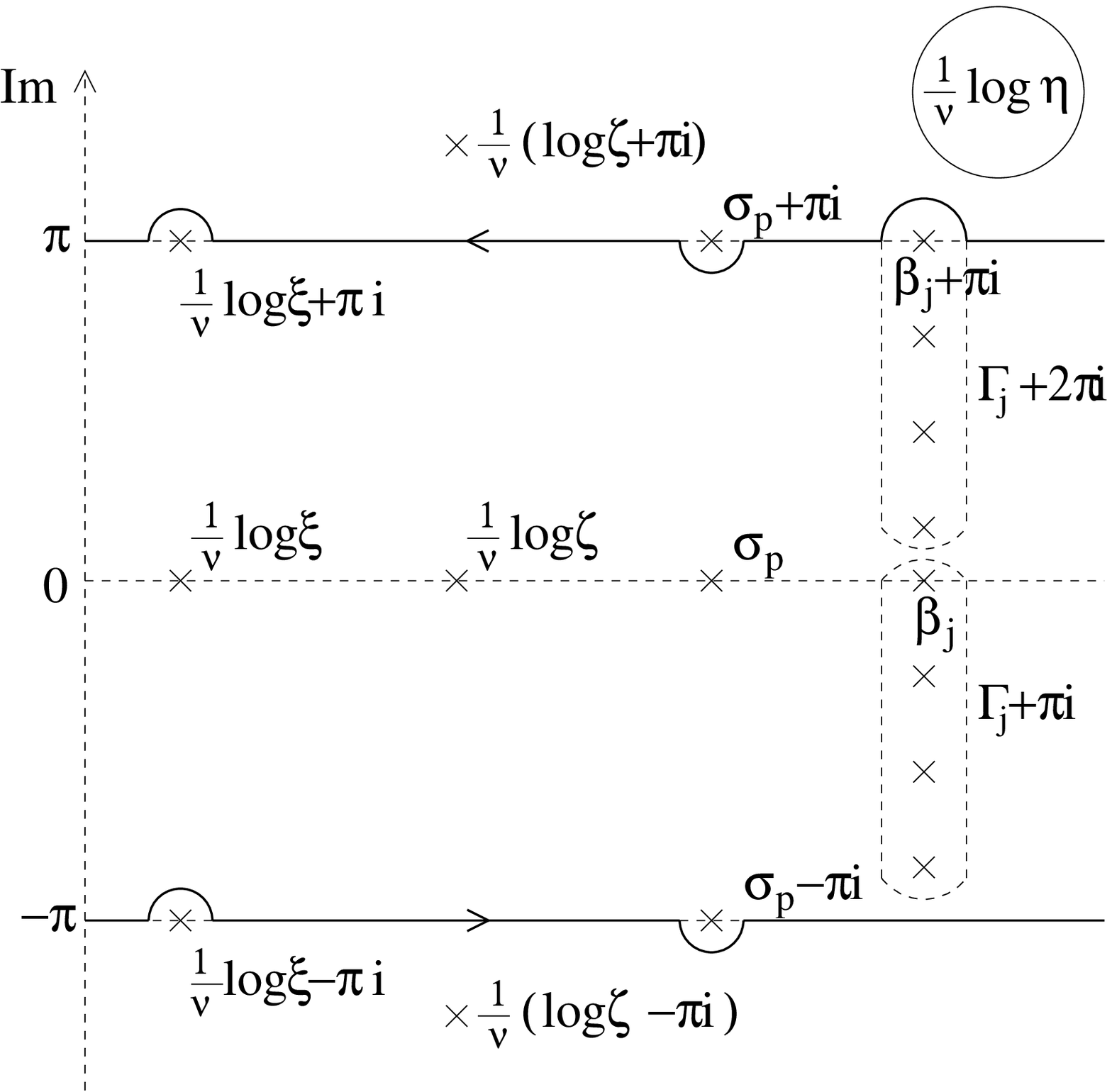}

\noindent
{\it Fig.7: Poles of the integrand in \eqref{compR}.
\\
The position of the poles of the integrand is depicted in the 
complex $(1/\nu)\log\eta$-plane.
}
\vskip .2cm
The figure {\it fig.7}
illustrates the position of poles.
The poles in sequences headed
by $\beta_j$ and $\beta_j+\pi i$ effectively occur when the integration in $\s_p\in\Gamma_j$ is performed.

We are happy with the first two terms in the right hand side 
of \eqref{compR}. As for
the last, containing the sum over $p$ terms,
we want to show that it vanishes when
the integrals over $\sigma _p$ are taken. 
Observing that in the summand for $p$ the dependence on $S_p$
before the curly bracket cancels out by $M^{(n)}_0$, one clearly sees
the situation is similar to
that considered in Proposition \ref{propQexact} for the $Q$-exact forms: both integrals
over $\mathbb{R}$ and $\Gamma$ reduce to residue at one pole like on
{\it fig. 2}. 
The difference with {the $Q$-exact forms} is that ${\psi_0}(\z S_p^{-\nu},\al)$
is not a
$2\pi i$-periodic function of $\sigma _p$. But on the other hand
there is 
a difference between
the integrands over $\mathbb{R}$ and $\Gamma$. It is easy 
to see that this difference is exactly such {that} the
residues at the fat points on {\it fig. 2} cancel. This finishes the proof.
\end{proof}
The difference of two solutions is proportional to 
\begin{align}
\mathrm{Diff}(\z,\xi)=X^{-1}{P(-X)}\Bigl(
 {(\ell_1^{(n)},R_{Z,X}^{(n)})_\al} {(\ell_2^{(n)},M^{(n)}_0)_\al}-
{(\ell_2^{(n)},R_{Z,X}^{(n)})_\al} {(\ell_1^{(n)},M^{(n)}_0)_\al}\Bigr)\,.
\label{diff}
\end{align}
This gives {a} huge but finite number of linearly
independent solutions
to the homogeneous equation. Indeed, {each} $\ell_i^{(n)}$ ($i=1,2$) is a linear
combination of $\ell_{i,1}\wedge\cdots \wedge \ell_{i,n}$. Every $\ell_{i,j}(\mathfrak{s})$
should be split into $m_{i,j}(\mathfrak{s})$ and $n_{i,j}(\mathfrak{s})$. The degree
of $m_{i,j}(\mathfrak{s})$ is restricted by the requirement of
convergence, and $n_{i,j}(\mathfrak{s})$ is defined modulo $p(\mathfrak{s}\mathfrak{q}^2)$
{(see \eqref{AMBIGUITY})}. Therefore, its degree is essentially bounded by $2n-1$. In addition
$\mathrm{Diff}(\z,\xi)$ is a polynomial of degree $n-1$ in $X^2$. 

{
Now we proceed to {the} computation of $\omega(\z,\xi)$.
Using the function $G(\eta,\xi;\ell^{(n)})$ we compute the integral in \eqref{omega}.
Since the dependence on $\eta$ and $\xi$ is solely contained in $R^{(n)}_{H,X}$,
we concentrate on the computation of the integral
\begin{align*}
\frac 1 {2\pi i}\int\limits _{\mathbb{R}_+e^{+i0}}\delta ^-_\z{\psi_0}(\z/\eta,\al)
R^{(n)}_{H,X}(S_1,\ldots,S_n)\frac {d\eta ^2}{\eta^2\rho(\eta)},
\end{align*}
keeping in mind that the integration with respect to the variables $S_1,\ldots,S_n$ is implied.}

To this end {we}
want first of all to transform somewhat {$R^{(n)}_{H,X}/\rho(\eta)$.} Do the partial fractions
\begin{align}
&R^{(n)}_{H,X}(S_1,\cdots ,S_n){/\rho(\eta)}\label{pf}
\\ &=\frac 1 {2\nu}{M_0^{(n)}(S_1,\ldots,S_n)}
\Bigl\{{-\frac{H}{H-X}+\frac{H}{H+X}\frac{P(X)}{P(-X)}}\nn\\
&+\frac{X}{P(-X)}\sum_{p=1}^n{S_p^{-1}\,}\frac {X^2-S_r^2}{S_p^2-S_r^2}
\Bigr(P(-S_p)\frac{H}{H-S_p}-P(S_p)\frac {H}{H+S_p}\Bigl)\Bigl\}.\nn
\end{align}
We call the last term in the bracket including the prefactor $\frac 1 {2\nu}M_0^{(n)}(S_1,\ldots,S_n)$
the sum term.
Let us transform {it by} using the Q-exact forms:
$$P(S_p)\frac {H}{H+S_p}\ \to\ AP(-S_p)\frac {H}{H+S_pQ}\,.$$
This is possible {unless} the point $\frac 1\nu\log\eta-\pi i$
{coincides} with one of {the} poles inside $\Gamma$. This does not happen if
$\frac 1\nu\log\eta$ is slightly above all $\beta _j$ which can be harmlessly
implied in the integral \eqref{omega}. So, our first goal is to compute the integral
\begin{align}
&{\frac1{2\pi i}}
\int\limits _{\mathbb{R}_+}\delta ^-_\z{\psi_0} (\z/\eta)\Bigl(\frac{H}{H-S_p}
-A\frac {H}{H+S_pQ}\Bigr)\frac {d\eta ^2}{\eta ^2}\label{AAAA}\\
&={\frac1{2\pi i}}
\(\int\limits _{\mathbb{R}_+}-\int\limits _{\mathbb{R}_+{e^{-2\pi i}}}\)
\delta ^-_\z{\psi_0}
 (\z/\eta)\frac{H}{H-S_p}\frac {d\eta ^2}{\eta ^2}=
{A}\frac{Z}{Z-S_pQ}-\frac{P(Z)}{P(-Z)}\frac{Z}{Z-S_p}\,.\nn
\end{align}
This computation is illustrated on the figure 
{\it fig. 8}
where the position of
{the poles and the effective poles}  is shown. The pole in the circle does not count because
of {the} multiplier $P(-S_p)$. {We used $A=e^{\pi i\al}$ twice.}
\vskip .5cm
\hskip 2cm\includegraphics[height=8cm]{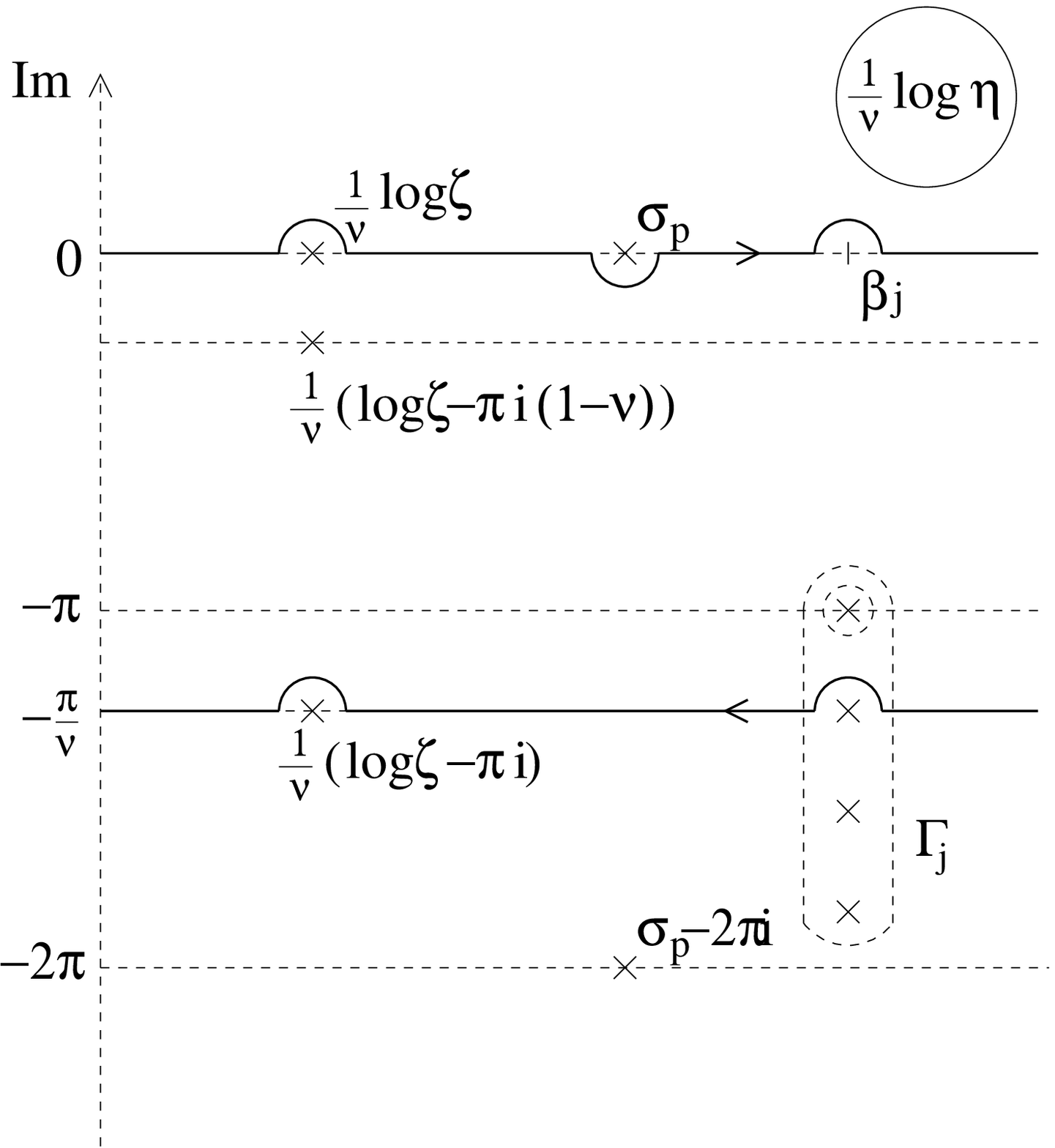}

\noindent
{\it Fig.8: Contour of integration for \eqref{AAAA}.\\
}
\vskip .2cm

We use once again the $Q$-exact forms in order to
transform
$${A}\frac{Z}{Z-S_pQ}P(-S_p)\ \to \ \frac{Z}{Z-S_p}P(S_p)\,.$$
{Thus, the sum term without the prefactor gives rise to}
\begin{align}
\frac 1 {P(-Z)P(-X)}\sum_{p=1}^n{\frac X{S_p}}\prod_{r\ne p}\frac {X^2-S_r^2}{S_p^2-S_r^2}
\cdot\frac Z {Z-S_p}\bigr(P(S_p)P(-Z)-P(-S_p)P(Z)\bigl)\,.\nn
\end{align}
Notice {that the last part is a polynomial of $Z$ and $S_p$.
We divide it into the even and odd parts (see \eqref{defC} for the definition of $C(Z,S)$):}
\begin{align}
&\frac Z {Z-S_p}\bigr(P(S_p)P(-Z)-P(-S_p)P(Z)\bigl)=-2\nu C(Z,S_p)\nn\\
&+\frac Z 2\Bigl(\frac {P(S_p)P(-Z)-P(-S_p)P(Z)}{Z-S_p}-
\frac  {P(-S_p)P(-Z)-P(S_p)P(Z)}{Z+S_p}\Bigr)\,.\nn
\end{align}
For the second term the summation can be performed by 
the interpolation formula.
We combine the result of this summation with the contribution
to $\omega (\z,\xi)$ coming from the first term in the right hand side of \eqref{pf}:
\begin{align}
&
{
-\frac 1 {2\nu}\hskip -.3cm\int\limits _{\ \mathbb{R}_+e^{+i0}}\delta ^-_\z{\psi_0} (\z/\eta)
\Bigl[\frac{H}{H-X}-\frac{H}{H+X}
\frac{P(X)}{P(-X)}\Bigr]\frac {d\eta ^2}{2\pi i\eta ^2}}\nn\\
&+ \frac  Z {4\nu P(-Z)P(-X)}
\Bigl(\frac {P(X)P(-Z)-P(-X)P(Z)}{{Z-X}}-
\frac  {P(-X)P(-Z)-P(X)P(Z)}{{Z+X}}\Bigr)\nn\\ \nn\\
&=-\delta ^-_\z\delta ^-_\xi\Delta ^{-1}_z{\psi_0} (\z/\xi,\al)\,.\nn
\end{align}
We use the analytical continuation formula \eqref{ANALYTIC CONTINUATION}.
Altogether we come to following nice expression
\begin{align}
\omega(\z,\xi;\ell^{(n)})=\frac {(\ell^{(n)}, L^{(n)}_{Z,X})_\al} {(\ell^{(n)},{ M^{(n)}_0})_\al}
\,,\label{mainomega}
\end{align}
where the polynomial $L^{(n)}_{Z,X}$ is given by 
\begin{align}
L^{(n)}_{Z,X}(S_1,\cdots, S_n)=\frac {\langle\Phi _\al\rangle} {P(-Z)P(-X)}
\left|\begin{matrix}0&C(Z,S_1) &\cdots &C(Z,S_n)\\
X &S_1&\cdots &S_n\\X^3 &S^3_1&\cdots &S^3_n\\ \vdots&\vdots&\vdots\\
X^{2n-1} &S^{2n-1}_1&\cdots &S^{2n-1}_n
\end{matrix}
\right|\,.\label{defLZX}
\end{align}
This final result does not depend on the arbitrariness of $m(\frak{s})$, $n(\frak{s})$
because $L^{(n)}_{Z,X}(S_1,\cdots, S_n)$ is a polynomial.

\section{Equivalence of BBS and BJMS fermions.}\label{equivalence}

Let us examine the result of the previous section. 
For finite $\mathbf{n}$, the function $\omega(\z,\xi)$ is associated 
with each eigenstate of the left Matsubara transfer matrix. 
In the limit $\mathbf{n}\to\infty$, 
the eigenvalues are parametrized by a set of real numbers 
 $\{\beta_1,\cdots,\beta_{2n}\}$  
and are ${{2n}\choose{n}}$-fold degenerate. 
Correspondingly, for each fixed $\beta_j$'s,
one has to have the same number of the functions 
 $\omega (\z,\xi;\ell ^{(n)})$.

In the process of solving the integral equation, however, we did not 
use any other condition on $\ell^{(n)}$  than that it is a 
skew-symmetric polynomial of degree at most $2n-1$ in each variable. 
It enters the solution as a ratio, so the solutions of the integral 
equation are parametrized by points in the projective space
$\mathbb{P}\bigl(\bigwedge^nV\bigr)$ where 
$V=\oplus_{j=0}^{2n-1}\mathbb{C}\mathfrak{s}^j$. 
Let us call $\langle\beta _1,\cdots,\beta _{2n};\ell^{(n)}|$ the actual 
eigenvectors of the Matsubara transfer matrix for
$\mathbf{n}\to\infty$. 
Then the corresponding $\ell^{(n)}$'s should be 
some special set of ${{2n}\choose{n}}$ points of the projective space.
Let us see if one can further narrow down the possibilities. 

The main determinant formula \eqref{mainIII} tells that 
\begin{align}
&\det \left(\omega(\z_i,\xi _j;\ell ^{(n)})\right)_{i,j=1,\cdots k}
\label{iii}\\
&\qquad\qquad=
\frac {\langle\beta _1,\cdots,\beta _{2n};\ell^{(n)}|
\bb^*(\zeta _1)\cdots\bb^*(\z_k)\cb^*(\xi _k)\cdots \cb ^*(\xi_1)
q^{2\al S(0)}|\mathrm{vac}\rangle}
{\langle\beta _1,\cdots,\beta _{2n};\ell^{(n)}|
q^{2\al S(0)}|\mathrm{vac}\rangle}.\nn
\end{align}
On the other hand, if $\ell ^{(n)}$ is a pure wedge product 
$\ell_0\wedge\ell_1\wedge \cdots\wedge\ell_{n-1}$
of linear factors $\ell_i$, then 
from the determinant formula \eqref{mainomega}--\eqref{defLZX}
one deduces by a simple linear algebra that
\begin{align}
\det \left(\omega(\z_i,\xi _j;\ell ^{(n)})\right)_{i,j=1,\cdots k}
=\frac{(\ell ^{(n)}\ ,\ 
\psi ^*_0(Z_1)\cdots \psi ^*_0(Z_k)\chi _0^*(X_k)\cdots \chi _0^*(X_1)M^{(n)}_0)_\al}
{(\ell ^{(n)}\ ,\ 
M^{(n)}_0)_\al}\,.\label{hhh}
\end{align}
Here, as usual $Z_j=\z ^{\frac 1 \nu}_j$, $X_j=\xi ^{\frac 1 \nu}_j$, and 
we used the definitions \eqref{psi0}, \eqref{chi0}.
This hints at the following postulate. 

\vskip .3cm

\noindent{\bf Postulate.} 
The polynomials $\ell ^{(n)}$
corresponding to the eigenvectors of the Matsubara transfer matrix for 
$\mathbf{n}\to\infty$ belong to the Grassmannian
$\mathrm{Gr}(n,V)\subset\mathbb{P}\bigl(\bigwedge^nV\bigr)$. 

\vskip .3cm

We present below an argument in favour of this postulate, by 
invoking  HLBA.

For fixed $\mathbf{n}$, 
the Bethe roots are either real or complex. 
According to \cite{DestriLow}, in
the limit $\mathbf{n}\to\infty$ the real
roots fill densely $\mathbb{R}_+$ with holes which correspond to solitons
with rapidities $\beta_j$, and the positions of the complex roots
are defined by the Bethe roots of the transfer-matrix constructed from the
physical S-matrix. Namely, to the vector
$\langle\beta _1,\cdots,\beta _{2n};\ell^{(n)}|$ there corresponds a Bethe vector
in $\(\mathbb{C}^2\)^{\otimes 2n}$. Let us be more precise. The two soliton S-matrix
is given by
\begin{align}
S_{i,j}(\beta _i-\beta _j)=S_0(\beta _i-\beta _j)\widetilde{S}_{i,j}(\mathfrak{b}_i/
\mathfrak{b}_j)\,,\label{Smatrix}
\end{align}
\begin{align}
&S_0(\beta)=\exp\(-i\int\limits _{0}^{\infty}
\frac{\sin(2k\nu\beta)\sinh((2\nu-1)\pi k)}{k\cosh(\pi \nu k)\sinh (\pi(1-\nu) k)}dk\)
\nn\,,
\end{align}
and
\begin{align}
&\widetilde{S}_{i,j}(\mathfrak{b}_i/
\mathfrak{b}_j)=
{\textstyle \frac 1 2 } (I_i\otimes I_j+\sigma ^3_i\otimes \sigma ^3_j)
+\frac{\mathfrak{b}_i-
\mathfrak{b}_j}
{\mathfrak{b}_i\mathfrak{q}^{-1}
-\mathfrak{b}_j\mathfrak{q}} \cdot
{\textstyle \frac 1 2 }(I_i\otimes I_j-\sigma ^3_i\otimes \sigma ^3_j)\nn\\&+
\sqrt{\mathfrak{b}_i
\mathfrak{b}_j}\frac{
\mathfrak{q}^{-1}-\mathfrak{q}}{\mathfrak{b}_i\mathfrak{q}^{-1}
-\mathfrak{b}_j\mathfrak{q}} \cdot
(\sigma ^+_i\otimes \sigma ^-_j+\sigma ^-_i\otimes \sigma ^+_j)\,.\nn
\end{align}
We define 
\begin{align}
\begin{pmatrix}A(\mathfrak{t})&B(\mathfrak{t})\\
C(\mathfrak{t})&D(\mathfrak{t})\end{pmatrix}_a=
\widetilde{S}_{a,2n}(\mathfrak{t}/
\mathfrak{b}_{2n})\cdots
\widetilde{S}_{a,1}(\mathfrak{t}/
\mathfrak{b}_1) \,.
\end{align}
Then we are interested in constructing the eigen-covectors of the 
HLBA transfer-matrix
$$T^{\mathrm{HLBA}}(\mathfrak{t})
=xA(\mathfrak{t})+x^{-1}D(\mathfrak{t})\,,$$
where $x$ is such that $|x|=1$, it is defined by $\theta$ from Section \ref{infinite}. 
The Bethe vectors in the weight zero sector
are given by the algebraic Bethe Ansatz \cite{FST}.
So, the correspondence between
the eigen-covector in the $\mathbf{n}\to\infty$ limit and the Bethe vector in HLBA reads as
\begin{align*}
\langle\beta _1,\cdots,\beta _{2n};\ell^{(n)}|\leftrightarrow
\langle\ \uparrow |\ \prod_{j=1}^nC(\mathfrak{u}_j)\,,
\end{align*}
where $\langle\ \uparrow |$ is the covector with all spins up, and $\mathfrak{u}=(u_1,\ldots,,u_n)$
satisfy the Bethe Ansatz equations. In this correspondence the eigen-covectors are parametrised by
$\mathfrak{u}$.
Suppose that $f_{\mathcal O_\al}(\beta _1,\cdots,\beta _{2n})$ is the form factor of $\mathcal O_\al$.
Then, we have the equality
\begin{align}
\langle\beta _1,\cdots,\beta _{2n};\ell^{(n)}|\mathcal O_\al|{\rm vac}\rangle=
\langle\ \uparrow |\prod_{j=1}^n
C(\mathfrak{u}_j)\cdot f(\beta _1,\cdots,\beta _{2n})
\label{fHLBA}\,.
\end{align}

Let us compute the right hand side by using \eqref{FF}. The definition of $w^{\epsilon_1,\cdots,\epsilon_{2n}}(\beta_1,\cdots,\beta _{2n})$ is basically
the same as in   \cite{book}, but it has  to be transposed since we consider the matrix elements
between the excited state and the vacuum and not 
vice versa:
$$w^{\epsilon_1,\cdots,\epsilon_{2n}}(\beta_1,\cdots,\beta _{2n})=
\prod_{j:\ \epsilon_j=+}C(\mathfrak{b}_j)|\downarrow\ \rangle\,.$$
The overall multiplier and the integral transformation involved in the
formulae for the form factors are independent of the
partitions. So, it is easy to see that we have to compute 
\begin{align}
&\ell ^{(n)}_{\{u\}}(\mathfrak{s}_1,\cdots \mathfrak{s}_n)\nn
\\&=\sum _{\{1,\cdots,2n\}\atop=I^-\cup I^+}
\ell ^{(n)}_{I^-\sqcup I^+}(\mathfrak{s}_1,\cdots \mathfrak{s}_n)
\frac 1{\prod\limits_{i\in I^-\atop j\in I^+}(\mathfrak{b}_i-\mathfrak{b}_j)}\prod
_{i\in I^-}\sqrt{\mathfrak{b}_i}
\langle\ {\uparrow}  
|\prod_{j=1}^n
C(\mathfrak{u}_j)\prod_{j:\ \epsilon_j=+}C(\mathfrak{b}_j)
|{\downarrow} 
\ \rangle\,.\nn
\end{align}
The scalar product in the right hand side is the domain wall partition
function given by the Izergin determinant \cite{Izergin}, 
so, we face a difficult
but clearly stated combinatorial problem. Surprisingly, for any set
$\mathfrak{u}_1,\cdots, \mathfrak{u}_n$ (not necessarily satisfying the Bethe
equations) we find that $\ell ^{(n)}_{\{u\}}(\mathfrak{s}_1,\cdots \mathfrak{s}_n)$
belongs to the Grassmanian, namely,
\begin{align}
&\ell ^{(n)}_{\{u\}}(\mathfrak{s}_1,\cdots \mathfrak{s}_n)\label{rrr}\\&=c(\mathfrak{q})\frac{\prod
\limits_{i,j=1}^n
(\mathfrak{u}_i-\mathfrak{u}_j\mathfrak{q}^2)}
{\prod\limits _{i=1}^{n}\prod\limits _{j=1}^{2n}
(\mathfrak{u}_i-\mathfrak{b}_j\mathfrak{q}^2)}\ \ \ell _{\{u\},0}\wedge
\ell _{\{u\},1}\wedge\cdots\wedge \ell _{\{u\},n-1}(\mathfrak{s}_1,\cdots \mathfrak{s}_n)\,,\nn
\end{align}
where $c(\mathfrak{q})$ is an irrelevant constant depending only on $\mathfrak{q}$, and
\begin{align}
&\ell _{\{u\},j}(\mathfrak{s})=
c(\mathfrak{q})^{-1}\frac
{\prod\limits _{i=1}^{n}\prod\limits _{j=1}^{2n}
(\mathfrak{u}_i-\mathfrak{b}_j\mathfrak{q}^2)}
{\prod
\limits_{i,j=1}^n
(\mathfrak{u}_i-\mathfrak{u}_j\mathfrak{q}^2)}
\label{kkkk}\\&\times\sum _{\{1,\cdots,2n\}\atop=I^-\cup I^+}
\ell _{I^-\sqcup I^+,j}(\mathfrak{s})
\frac 1{\prod\limits_{i\in I^-\atop j\in I^+}(\mathfrak{b}_i-\mathfrak{b}_j)}\prod
_{i\in I^-}\sqrt{\mathfrak{b}_i}
\langle\  {\uparrow}  
|\prod_{j=1}^n
C(\mathfrak{u}_j)\prod_{j:\ \epsilon_j=+}C(\mathfrak{b}_j)
|{\downarrow}
\ \rangle\,.\nn
\end{align}
Notice that 
$$\ell _{I^-\sqcup I^+,j}=c_{j,0}\mathbf{s}^{n+j}+\cdots c_{j,j}\mathbf{s}^{n}+
d_{j,1}\mathbf{s}^{n-1}+\cdots +d_{j,n}\,,$$
where $d_{j,k}$ depend on the partition $I^+\cup I^-$ and $c_{j,k}$ do not.
From this fact
it is easy to see that the  identity \eqref{kkkk} is a necessary condition for 
\eqref{rrr} to hold being a part of it. The
rest of equations in  the identity \eqref{rrr} can be viewed as Pl\"ucker relations.

Now we see that our assumption about $\ell ^{(n)}$ being in the Grassmannian
fits completely with the fact that the function $\omega (\z,\xi;\ell^{(n)})$ must
describe the ratio of the component of the form factor for the descendant
to the one for the primary field in the basis of Bethe vectors.
The formulae  \eqref{iii} and \eqref{hhh} allow us
to compute the form factors of
any operator for homogeneous or inhomogeneous XXZ chain. 
Here we
shall not consider the form factors for the 
lattice model
and proceed 
directly to the sG case.

As we have seen, for the description of the sG form factors, 
it is natural to slightly modify 
the operators $\psi_0^*(\z)$, $\chi^*_0 (\xi)$ by 
introducing a nontrivial block $\mathcal{A}$ in \eqref{genferm1}. 
On the other hand, when we take
the scaling limit from the inhomogeneous spin chain,  
the lattice fermions $\bb^*(\z)$, $\cb^*(\z)$ are also
modified by a Bogolubov transformation involving 
the function 
$$
\omega _0(\z,\xi)
=-\delta ^-_\z\delta^-_\xi\Delta ^{-1}_\z\psi _0(\z/\xi,\al)\,.
$$
If we are working with the asymptotic operators 
$\betab^*(\z)$, $\gammab^*(\xi)$,
$\bar\betab^*(\z)$, 
$\bar\gammab^*(\xi)$ this function contributes
only to the pairings of $\betab^*(\z)$ with $\bar\gammab^*(\xi)$, and
of $\bar\betab^*(\z)$ with $\gammab^*(\xi)$. 
In these two cases 
we consider
$\z\to\infty, \xi\to 0$ and $\z\to 0,\xi\to\infty$ respectively, and keep in the
asymptotics \eqref{asDinv}
of $\psi _0(\z/\xi,\al)$ only the part going in $(\xi/\z)^{\frac{j}\nu}$ or
$(\z/\xi)^{\frac{j}\nu}$. So, effectively $\omega _0(\z,\xi)$
is replaced either by
\begin{align*}
&\omega_{0,+}(\z,\xi)=-\delta ^-_\z\delta^-_\xi 
{\textstyle\frac i{4\nu}}
\sum_{j=0}^{\infty}(-1)^j(\xi/\z) ^{\frac j {\nu}}\cot {\textstyle\frac \pi{2\nu}}
(\al\nu+j),\nn
\end{align*}
or by
\begin{align*}
&\omega_{0,-}(\z,\xi)=\delta ^-_\z\delta^-_\xi 
{\textstyle\frac i{4\nu}}
\sum_{j=1}^{\infty}(-1)^j(\z/\xi )^{\frac j {\nu}}\cot {\textstyle\frac \pi{2\nu}}
(\al\nu-j)\,.\nn
\end{align*}
Now it is easy to see that 
$$\omega_{0,\pm}(\z,\xi)=\frac{C_\pm (Z,X)}{P(-Z)P(-X)}\,.$$
Due to this identity the effect of modifying 
$\bb^*(\z)$ and $\cb^*(\z)$ is exactly the same as the effect of
modifying $\psi_0^*(\z)$ , $\chi^*_0 (\xi)$. 

Putting together  all pieces of the puzzle 
we come to the main conclusion of this work:
\begin{align}
&L^{(\star)}_{\betabs^*(\mubs\z_1)\cdots \betabs^*(\mubs\z_p)
\bar\betabs^*(\z_{p+1}/\mubs)
\cdots \bar\betabs^*(\z_{k}/\mubs)
\bar\gammabs^*(\xi_k/\mubs)\cdots \bar\gammabs^*(\xi_{q+1}/\mubs)
\gammabs^*(\mubs\xi_q)\cdots \gammabs^*(\mubs\xi_1)\Phi _\al}\label{TheMain}\\&=
\psi ^*(Z_1)\cdots \psi ^*(Z_p)
\bar{\psi} ^*(Z_{p+1})\cdots \bar{\psi} ^*(Z_{k}) \bar{\chi }^*(X_k)\cdots 
\bar{\chi }^*(X_{q+1})\chi ^*(X_q)\cdots \chi ^*(X_1)
M^{(\star)}_0\,.\nn
\end{align}
So, the action of the BJMS fermions on local operators coincides with the action of
the BBS fermions on towers. 
This exact identification of two things introduced originally
for completely different reasons
is one more evidence of deep self-consistency of 
integrable two-dimensional quantum field theory.

\section{ BBS construction of null vectors}\label{null}

The main achievement of the paper \cite{BBS} consists in the fermionic
description of null vectors. 
In this section we shall discuss this issue. 
Throughout this section we fix the number of the parameters 
$\beta_j$ to be $2n$.  

Up to now we were interested in generic $\al$. 
However, as it is clear 
from the discussion of the regularised integrals, 
something special happens at the points of resonance, 
i.e., at the points where the pairing \eqref{MAINDEF} has poles. 
In this paper we shall consider only
the resonances occurring at $\al =\xin m$ with 
$m\in\Z_{\geq0}$. In the
CFT language
this corresponds to considering the degenerate fields which 
in the standard notation are 
$$\Phi_{1,m+1}=\Phi _{\xin m}\,.$$
We shall use both of these symbols. 
The polynomials $a^n\ell^{(n)}_{I^-\sqcup I^+}$ 
depend on $\al$ only through $a^2$, i.e., $\xin$-periodically.
So, 
the form factors of the descendants
$\mathcal{O}_{\al +m\xin}$ are
expressible via the pairings $(\ ,\ )_\al$:
\begin{align}
(\ell^{(n)}_{I^-\sqcup I^+}(\mathfrak{s}),L^{(n)}(S))_{\al+m\xin}
=(-)^{mn}
(\ell^{(n)}_{I^-\sqcup I^+}(\mathfrak{s}),L^{(n)}(S)\prod_{j=1}^nS_j^m)_\al\,.\label{shiftal}
\end{align}
Having this in mind we shall consider all of them together.

In the present paper we shall consider only
the case of the fields $\Phi _{1,2m}$ which we call even. 
According to our logic all these primary fields 
and their Virasoro descendants should
be considered as the fermionic descendants of 
$$
\Phi_{1,2}=\Phi _{\xin}\,.
$$
We
denote this space by $\mathcal{H}_{\frac{1-\nu}\nu}=\oplus_{c\in\Z}\mathcal{H}_{\frac{1-\nu}\nu,c}$
where $\mathcal{H}_{\frac{1-\nu}\nu,c}$ denotes the sector of charge $c$.

Similarly, 
the case of the fields $\Phi _{1,2m+1}$ which we call odd
is reduced to the case $\al=0$. 
While the point $\al=\xin$ belongs to the fundamental domain \eqref{fundomain}, 
$\al =0$ is on its boundary, and this is a source of many complications. 
For that reason we decided to postpone the consideration 
of $\al=0$ till a future work.
Still some technical points explained in the next 
subsection will be common to both even and odd cases.

\subsection{Peculiar properties of form factors at $a^2=1$}
\definecolor{5/2}{rgb}{0.8,0.2,0}
Due to \eqref{shiftal}, in order to treat the case
$a^2=1$ it is sufficient to
study
$$(\ell^{(n)}_{I^-\sqcup I^+}, N^{(n)}(S))_\al\bigl|_{\al=0}$$
for all Laurent polynomials $N^{(n)}(S)$. Let us do that assuming for the moment that the
coefficients of $N^{(n)}(S)$ are independent of $\al$.

Motivated by the formulae \eqref{reslm+}, \eqref{reslm-}, \eqref{MAINDEF},
let us introduce the following 
definition of residues for any Laurent polynomials $N(S)$, 
$\ell(\mathfrak{s})$:
\begin{align}
&\mathrm{Res}_{+}[N]=\res_{S=\infty}\Bigl(X^+(S)S^{-2n}N(S)\frac {dS}{S}\Bigr),
\ \frak{res}_{+}[\ell]=\res _{\mathfrak{s}=\infty}\Bigl(x^+(\mathfrak{s})
\mathfrak{s}^{-n}\ell(\mathfrak{s})\frac{d\mathfrak{s}}{\mathfrak{s}}\Bigr)\,,\nn\\
&\mathrm{Res}_{-}[N]=
\res_{S=0}\Bigl(X^-(S)N(S)\frac {dS}{S}\Bigr),
\ \qquad\ \mathfrak{res}_{-}[\ell]=
\res _{\mathfrak{s}=0}
\Bigl(x^-(\mathfrak{s})\ell(\mathfrak{s})\frac{d\mathfrak{s}}{\mathfrak{s}}\Bigr)\,.\nn
\end{align}
Extending the definition of $\mathrm{Res}_\pm$ as
\begin{align*}
&\mathrm{Res}_\pm\[N_1\wedge\cdots\wedge N_k\](S_1,\cdots ,S_{k-1})\\
&=
\sum_{j=1}^k
(-1)^{j-1}\mathrm{Res}_{\pm}[N_j]\ \(N_1\wedge\cdots\wedge \widehat{N_j}\wedge\cdots\wedge N_k
\)(S_1,\cdots ,S_{k-1})\,,
\end{align*}
we have
\begin{align}
&\mathrm{Res}_{\pm}^2=0,
\label{2=0}
\\
&\mathrm{Res}_{\pm}[D_1[Z]]=0\,,
\label{res[Z]}
\end{align}
where $D_1[Z]=D_A[Z]\bigl|_{A=1}$. 
Obviously the same relations hold for 
$\mathfrak{res}_{\pm}$ 
defined similarly as above.
Note that in the definition above, 
no multiple poles occur by antisymmetry.

Then using \eqref{reslm+}, \eqref{reslm-} and
\eqref{MAINDEF}, we can write the residue of the
pairing as 
\begin{align}
\res_{\al=0}(\ell^{(k)},N^{(k)})_\al
\,d\al=-\xinn \bigl(\bigl(\frak{res}_+[\ell^{(k)}]
,\mathrm{Res}_{+}[N^{(k)}]\bigr)_0-\bigl(\frak{res}_-[\ell^{(k)}] ,\mathrm{Res}_{-}[N^{(k)}]\bigr)_0\bigr)\,,
\label{nice}
\end{align}
for any $\ell^{(k)}$ and $N^{(k)}$ 
whose coefficients are independent of $\al$.  

Consider the polynomials $\ell_{I^-\sqcup I^+,i}$ ($i=0,1,\cdots ,n-1$).
We have
$$
\ell_{I^-\sqcup I^+,0}(\mathfrak{s})=(a^{-1}-a)
\mathfrak{q}^{-2n}p_{I^+}(\mathfrak{s}\mathfrak{q}^2)\,.
$$
The rest of them have the properties ($1\le i\le n-1$)
\begin{align}
&\ell_{I^-\sqcup I^+,i}(\mathfrak{s})=a
D_a[\mathfrak{s}^{i-n}]+O(\mathfrak{s}^{n-1})
\,,\quad \mathfrak{s}\to \infty\,,\label{1ton-1}\\
&\ell_{I^-\sqcup I^+,i}(\mathfrak{s})
=(a^{-1}-a)C_i(\mathfrak{b})+O(\mathfrak{s})\,, 
\quad\mathfrak{s}\to 0\,,\nn
\end{align}
where $C_i(\mathfrak{b})$ are irrelevant constants depending only on $\mathfrak{b}_j$'s.
These follow easily from \eqref{GENERATING}.
The formulae \eqref{1ton-1} clearly imply that
$$\mathfrak{res}_\pm [\ell_{I^-\sqcup I^+,i}(\mathfrak{s})]\Big|_{\al=0}=0,\quad i=1,\cdots n-1\,.$$

Altogether we come to the formula
\begin{align}
&\lim_{\al\to 0}(\ell^{(n)}_{I^-\sqcup I^+}\ ,\ N^{(n)})_\al\label{PAIRING}\\&=
\left.-2\pi i\Bigl(
(\tilde\ell_{I^-\sqcup I^+}^{(n-1)}\ ,\ \mathrm{Res}_+[N^{(n)}])_\al
+(-\mathfrak{q})^n\prod_{i\in I^-}\mathfrak{b}_i^{-\frac 1 2}\prod_{i\in I^+}\mathfrak{b}_i^{\frac 1 2}
(\tilde\ell_{I^-\sqcup I^+}^{(n-1)}\ ,\ \mathrm{Res}_-[N^{(n)}])_\al\Bigr)\right|_{\al=0}\,,
\nn
\end{align}
where 
$$\tilde\ell_{I^-\sqcup I^+}^{(n-1)}
=\ell_{I^-\sqcup I^+,1}\wedge\cdots\wedge \ell_{I^-\sqcup I^+,n-1}\,.$$

When $\al=0$, the monomial $S^{2n}$ cannot be reduced to lower degree because
it drops from the exact form $P(S)-AP(-S)$. So one has to relax the
degree restriction \eqref{deg-rest} to
\begin{align}
0\le\mathrm{deg}_{S_i}N^{(n)}(S_1,\cdots ,S_n)\le 2n\,
\label{kkkkkk}
\end{align}
An important point about $\al=0$ 
is that for certain polynomials satisfying \eqref{kkkkkk}  
the pairing \eqref{PAIRING} vanishes.  
This happens if one of the following three conditions is met:
\medskip

\noindent (i) 
Vanishing of both residues. 
The pairing \eqref{PAIRING} vanishes if 
\begin{align}
\mathrm{Res}_{+}[N^{(n)}]=\mathrm{Res}_{-}[N^{(n)}]=0\,.
\label{res=0}
\end{align}
\noindent (ii) 
Exact form. The pairing vanishes if 
\begin{align}
N^{(n)}=D[1]\wedge N^{(n-1)}\,,\label{P-P}
\end{align}
where 
$$
D[1]=D_A[1]|_{A=1}=P(S)-P(-S)\,.
$$
\noindent (iii) 
Quantum Riemann bilinear identity.
The pairing vanishes if
\begin{align}
N^{(n)}=C^{(2)}\wedge N^{(n-2)}\,,\label{C2}
\end{align}
where 
$$
C^{(2)}(S_1,S_2)=C (S_1,S_2)-C(S_2,S_1)\,.
$$

The last property (iii)
needs some comments. It is a consequence of the
quantum Riemann bilinear identity \cite{smiriemann}.
Consider the antisymmetric polynomial
$$
c^{(2)}(\mathfrak{s}_1,\mathfrak{s}_2)=p(\mathfrak{s}_1)\frac{\mathfrak{s}_2
\mathfrak{q}^2}
{\mathfrak{s}_1-\mathfrak{q}^2\mathfrak{s}_2}-
p(\mathfrak{q}^2\mathfrak{s}_1)
\frac{\mathfrak{s}_2}
{\mathfrak{s}_1\mathfrak{q}^2-\mathfrak{s}_2}
-p(\mathfrak{s}_2)
\frac{\mathfrak{s}_1\mathfrak{q}^2}
{\mathfrak{s}_2-\mathfrak{q}^2\mathfrak{s}_1}+
p(\mathfrak{q}^2\mathfrak{s}_2)
\frac{\mathfrak{s}_1}
{\mathfrak{s}_2\mathfrak{q}^2-\mathfrak{s}_1}\,.
$$
Suppose we find polynomials $r_i(\mathfrak{s})$, $s_i(\mathfrak{s})$
($i=1,\cdots ,n-1$) in the kernel of the operators
\begin{align*}
\mathfrak{res}_\pm:\C[\mathfrak{b}]^{2n-1}\simeq\oplus_{j=1}^{2n-1}
\C[\mathfrak{b}]\mathfrak{s}^j\rightarrow\C[\mathfrak{b}]
\end{align*}
such that 
\begin{align}
c^{(2)}(\mathfrak{s}_1,\mathfrak{s}_2)=
\sum_{j=1}^{n-1}
\(r_i(\mathfrak{s}_2)s_i(\mathfrak{s}_1)
-r_i(\mathfrak{s}_1)s_i(\mathfrak{s}_2)\)\,.
\label{c=rs}
\end{align}
Then a pairing  in 
${\rm Ker}\,\mathfrak{res}_+\cap {\rm Ker}\,\mathfrak{res}_-$ 
is defined by
$$r_i\circ r_j=s_i\circ s_j=0,\quad r_i\circ s_j=\delta _{i,j}\,.$$
The sets $\{r_i\}_{i=1}^{n-1}$ 
and $\{s_i\}_{i=1}^{n-1}$ are called half-bases. Obviously 
there is an action of the  symplectic group $Sp(2n-2)$.

The quantum Riemann bilinear identity (see \cite{smiriemann}) states
\begin{align}
(m_1\wedge m_2, C^{(2)})_0=2\pi i (m_1\circ m_2)\,.\label{riemann}
\end{align}
It is clear from \eqref{GENERATING} 
that for any partition $\{1,\cdots,2n\}=I^-\sqcup I^+$ 
the following polynomials satisfy \eqref{c=rs} for $a=1$: 
\begin{align}
r_i(\mathfrak{s})= (\mathfrak{q}^2\mathfrak{s})^{n-i}\,,\quad
s_i(\mathfrak{s})=\ell_{I^-\sqcup I^+,i}(\mathfrak{s})\,.\nn
\end{align}
So,
\begin{align}
\ell_{I^-\sqcup I^+,i}\circ
\ell_{I^-\sqcup I^+,j}=0\,,\quad i,j=1,\cdots, n-1\,,\label{ll=0}
\end{align}
which ensures that 
the pairing of
$\ell_{I^-\sqcup I^+,i}\wedge \ell_{I^-\sqcup I^+,j}$ 
and $C^{(2)}$ vanishes. 
The pairing of 
$\ell _{I^-\sqcup I^+,0}\wedge\ell _{I^-\sqcup I^+,j}$ 
and $C^{(2)}$ also vanishes due to 
\begin{align*}
\mathrm{Res}_{+,S_1}C^{(2)}(S_1,S_2)=0,
\quad \mathrm{Res}_{-,S_1}C^{(2)}(S_1,S_2)=-{\textstyle\frac 1 {2\nu}}(P(S_2)-P(-S_2))\,
\end{align*}
and \eqref{P-P}.

Before going further let us make one remark,
 which is supposed to be well known but
probably worth being repeated.
For generic $\al$ the polynomials
$\ell_{I^-\sqcup I^+,0}\wedge\cdots\wedge \ell_{I^-\sqcup I^+,n-1}$ 
with different partitions of
$\mathfrak{b}_j$'s span the space $\bigwedge^n\mathbb{C}^{2n}$ 
whose dimension 
$\binom{2n} {n}$ coincides with that of 
the weight $0$ subspace
$(\mathbb{C}^{\otimes 2n})_0$.
The latter can be interpreted in terms of the
$U(1)$-symmetry of the sG model:
soliton-antisoliton provide a two-dimensional representation of $U(1)$.
When $\al=m \xin$ 
the polynomials 
$\ell_{I^-\sqcup I^+,1}(\mathfrak{s}),\cdots,
\ell_{I^-\sqcup I^+,n-1}(\mathfrak{s})$
which enter the form factor formulas 
are divisible by $\mathfrak{s}$, 
and $\mathfrak{res}_+$ vanishes on them.
So, they are linear combinations of $2n-2$ different monomials.
Moreover, the form factors vanish if $\tilde\ell^{(n-1)}$ 
contains $c^{(2)}$ as a multiplier. 
Altogether, instead of 
$\bigwedge^n\mathbb{C}^{2n}$, we have a
smaller space $\bigwedge^{n-1}\mathbb{C}^{2n-2}$ 
isomorphic to the space of the maximal irreducible 
representation of $Sp(2n-2)$. 
The dimension of this space equals
$$\binom{2n-2}{n-1}-\binom{2n-2}{n-3}=\binom{2n}{n}-\binom{2n}{n-1}\,.$$
In the right hand side we have the Catalan number,
the multiplicity of the singlet representation of 
the group $SL(2)$ in $\mathbb{C}^{\otimes 2n}$.
This fact agrees with the quantum group reduction \cite{ReshSmi} which states
that the form factors of $\Phi _{\xin m}$ are invariant under the action
of the quantum group $U_{\mathfrak{q}}(\slt)$. 
Certainly,
for generic $\mathfrak{q}$, 
this multiplicity is the same as in the classical case.

\subsection{Null vectors for the fields $\Phi_{1,2m}$.}

In
the case $\al=\xin$ and generic $Q$ we can safely go from the 
odd representatives for $L^{(n)}$ to those of restricted degree and 
vice versa.
So, following the remark at the end of
Section \ref{BBSfermions} we consider the fermions defined by \eqref{modify-fermion}, 
which act on the space of towers with restricted degrees:
$$0\le\mathrm{deg}_{S_i}L^{(n)}(S_1,\cdots,S_n)\le 2n-1\,.$$
The identification with the previous subsection goes through
\begin{align}
N^{(n)}(S)=\prod_{j=1}^nS_j\cdot 
L^{(n)}(S)\,.
\label{N=SL}
\end{align}
As already noted, we have 
$(\ell^{(k)},L^{(n)})_\xin =(\ell^{(k)},N^{(n)})_0$.

First, observe that
$$
C(Z,S)=\frac 1 {2\nu}Z\cdot \frac {P(S)-P(-S)}{S}+O(Z^3)\,,
$$
and
$$
C_-(Z,X)=Z\cdot {\textstyle\frac i \nu}
\cot{\textstyle\frac \pi 2}(\al -{\textstyle\frac 1\nu})
\sigma _{2n}(B)X^{-1}(P(X)+P(-X))\bigl|_{\al=\xin}
+O(Z^3)=O(Z^3)
\,.
$$
It means that,  
under the identification \eqref{N=SL}, 
the operator $\bar\psi ^*_1$ acts as 
wedge product by an exact form. 
Hence  by \eqref{P-P} we find 
a set of null vectors:
\begin{align}
\bar{\psi }_1^*\ \mathcal{H}_{\frac{1-\nu}\nu,-1}= 0\,.\label{psibar*1}
\end{align}

We have 
$$
1\le\mathrm{deg}_{S_i}
\Bigl(\prod _{j=1}^nS_j\cdot L^{(n)}(S_1,\cdots,S_n)\Bigr)
\le 2n\,,
$$
which implies that 
\begin{align}
&\mathrm{Res}_-\(
\prod _{j=1}^nS_j\cdot L^{(n)}(S_1,\cdots,S_n)\)=0\,,\nn\\
&\mathrm{Res}_+\(
\prod _{j=1}^nS_j\cdot L^{(n)}(S_1,\cdots,S_n)\)=
\prod _{j=1}^{n-1}S_j\cdot\(\chi ^*_1
L^{(\star)}\)^{(n-1,n)}(S_1,\cdots,S_{n-1})\,.\nn
\end{align}
This gives due to $(\chi ^*_1)^2=0$ another set of
null vectors:
\begin{align}
\chi _1^*\ \mathcal{H}_{\xin,1}\simeq 0\,.\label{chi*1}
\end{align}

Now we turn to the Riemann bilinear identity. 
Introduce the operators
\begin{align}
\mathcal{C}_\mathrm{even}=\oint\psi ^*(D)\chi (D)
\frac {dD}{2\pi iD^3}
\,,
\quad 
\overline{\mathcal{C}}_\mathrm{even}
=\oint\bar{\psi }^*(D)\bar{\chi} (D)\frac {dD}{2\pi iD^3}
\,.\nn
\end{align}
We have for $L^{(\star)}\in \mathcal{H}_{\xin, -s-2}$
\begin{align}
\(
(\mathcal{C}_\mathrm{even}+\overline{\mathcal{C}}_\mathrm{even})
L\)^{(n-s,n)}(S_1,\cdots,S_{n-s})=-
{C}^{(2)}_\mathrm{even}\wedge L^{(n-s-2,n)}(S_1,\cdots,S_{n-s})\,,
\label{Ceven}
\end{align}
where
$$C^{(2)}_{\mathrm{even}}(S_1,S_2)=C (S_1,S_2)S_1^{-2}-C (S_2,S_1)S_2^{-2}\,.$$
The only non-trivial part of this computation is to make sure that no cross-terms
occur due to the Bogolubov transform \eqref{def-psi1}, \eqref{def-psi2}. This follows from
$$(x\tau _+(x)-x ^{-1}\tau_-(x^{-1}))\bigr|_{\al=\frac{1-\nu}\nu}=0\,.$$
Since
$$
S_1S_2C^{(2)}_{\mathrm{even}}(S_1,S_2)=-C
^{(2)}(S_1,S_2)\,,
$$
this together with \eqref{C2} and \eqref{Ceven}
gives rise to new null-vectors
\begin{align}
\(\mathcal{C}_\mathrm{even}+\overline{\mathcal{C}}_\mathrm{even}\)
\mathcal{H}_{\xin,-2}\simeq 0\,.\label{nullCeven}
\end{align}

In what follows we shall be  interested in right-chiral null vectors. 
For the operator $\Phi _{1,2}$ it is easy to identify them:
\begin{align}
&\chi ^*_1\ \psi ^*_{I^+}\chi ^*_{I^-}M^{(\star)}_0\simeq 0,\quad\ \ \ \#(I^+)=\#(I^-)+1\nn\\
&\mathcal{C}_\mathrm{even}\ \psi ^*_{I^+}\chi ^*_{I^-}M^{(\star)}_0\simeq 0,
\quad {\#(I^+)=\#(I^-)-2},\ \ 1\notin I^-\nn\,,\end{align}
where we introduced the notation
\begin{align}
&\psi _I^*=\psi^*_{a_1}\cdots \psi _{a_p}^*,\quad \chi _I^*=\chi^*_{a_p}\cdots \chi_{a_1}^*,
\quad
\bar\psi _I^*=\bar\psi^*_{a_1}\cdots \bar\psi _{a_p}^*,\quad \bar\chi _I^*=\bar\chi^*_{a_p}\cdots \bar\chi_{a_1}^*, \label{multiindex}\\
&\mathrm{for}\ I=\{a_1,\cdots ,a_p\},\quad a_1<a_2<\cdots <a_p\,.\nn
\end{align}
We added
the condition $1\notin I^-$ in the second formula since we do not want to count the same null-vector twice.
Notice that
$$[\mathcal{C}_\mathrm{even},\chi ^*_1]=0\,.$$

Consider the right-chiral descendants of $\Phi _{1,2m+2}$:
$$
\psi^*_{I^+}\chi _{I^-}^*\bar\chi ^*_{I_\mathrm{odd}(m)}M^{(\star)}\,, \quad \#(I^+)=\#(I^-)+m\,,
$$
where we
recall that $I_\mathrm{odd}(m)=\{1,3,\cdots, 2m-1\}$.
At first glance 
it is not clear how to construct the right-chiral null vectors with 
$\mathcal{C}_{\mathrm{even}}+\overline{\mathcal{C}}_{\mathrm{even}}$
because 
$\bar\chi _{2j-1}$ 
present in $\overline{\mathcal{C}}_{\mathrm{even}}$ 
might
spoil the product 
$\bar{\chi} ^*_{I_\mathrm{odd}(m)}$.
The solution was found in \cite{BBS}. Consider
$$
\mathcal{C}_{\mathrm{even}}^{m+1}
\psi^*_{I^+}\chi _{I^-}^*\bar\chi ^*_{I_\mathrm{odd}(m)}M^{(\star)}
\,, \quad {\#(I^+)=\#(I^-)-m-2}\,.
$$
According to the above considerations this can be transformed 
\begin{align}
\mathcal{C}_{\mathrm{even}}^{m+1}
\psi^*_{I^+}\chi _{I^-}^*\bar\chi ^*_{I_\mathrm{odd}(m)}M^{(\star)}
\simeq(-\overline{\mathcal{C}}_{\mathrm{even}}
)^{m+1}\psi^*_{I^+}\chi _{I^-}^*\bar\chi ^*_{I_\mathrm{odd}(m)}M^{(\star)}=0\,.\nn
\end{align}
The latter identity is due to the fact that
acting on $\bar\chi ^*_{I_\mathrm{odd}(m)}M^{(\star)}$ every 
$\overline{\mathcal{C}}_{\mathrm{even}}$ can be replaced
by the finite sum
$\sum_{j=1}^{m}\bar\psi _{2j+1}^*\bar\chi _{2j-1}$.

Let us summarise. 
The following right chiral null vectors exist for $\Phi _{(2m+1)\xin}$:
\begin{align}
&\chi^*_1\ \psi ^*_{I^+}\chi ^*_{I^-}
\bar{\chi} _{I_\mathrm{odd}(m)}^*M^{(\star)}_{0}\,,
\quad\ \  \#(I^+)=\#(I^-)+m+1\,,\nn
\\
&\mathcal{C}_{\mathrm{even}}^{m+1}\ 
\psi _{I^+}^*\chi ^*_{I^-}
\bar{\chi} _{I_\mathrm{odd}(m)}^*M^{(\star)}_{0}\,,
\quad {\#(I^+)=\#(I^-)-m-2},\ \ 1\notin I^-\,.
\nn
\end{align}
Introduce locally the notation
$$
H_{k}=\mathrm{Span}\{\, \psi _{I^+}^*\chi ^*_{I^-}
M^{(\star)}_{0}\,\mid
{\#(I^+)=\#(I^-)-k}\, ,\ \ 1\notin I^-\}\,.
$$
It is easy to see that 
\begin{align}
H_{k}\rightarrow\mathcal{C}_{\mathrm{even}}^{k}
H_{-k}\quad\mathrm{is\ an\ isomorphism}\,.\label{Ck}
\end{align}
Using \eqref{Ck} for $k=m+1$ 
we can combine
the two kinds of  the null-vectors into
\begin{align}
&\mathcal{C}_{\mathrm{even}}^{m+1}\ 
\psi _{I^+}^*\chi ^*_{I^-}
\bar{\chi} _{I_\mathrm{odd}(m)}^*M^{(\star)}_{0}\,,
\quad {\#(I^+)=\#(I^-)-m-2}\,,
\label{Cmeven}
\end{align}
if $1\in I^-$ we get the first kind of null-vectors, and if $1\notin I^-$ we get the second.

For the left descendants the situation is even simpler 
because $\mathcal{C}_{\mathrm{even}}$ commutes with $\psi ^*_{2j-1}$:
\begin{align}
&\bar{\psi}^*_1\ \bar\psi _{I^+}^* \bar\chi ^*_{I^-}\psi ^*
_{I_\mathrm{odd}(m)}
M^{(\star)}_{0}\,,
\quad\ \ \ \ {\#(I^+)=\#(I^-)-m-1}\,,\nn
\\
&\overline{\mathcal{C}}_{\mathrm{even}}
\ \bar\psi _{I^+}^*\bar\chi ^*_{I^-}\psi ^*
_{I_\mathrm{odd}(m)}
M^{(\star)}_{0}\,,
\quad \ \ {\#(I^+)=\#(I^-)-m-2},\ \ 1\notin I^+\,,\nn
\end{align}
which again can be put together as
\begin{align}
&\overline{\mathcal{C}}_{\mathrm{even}}
\ \bar\psi _{I^+}^*\bar\chi ^*_{I^-}\psi ^*
_{I_\mathrm{odd}(m)}
M^{(\star)}_{0}\,,
\quad \ \ {\#(I^+)=\#(I^-)-m-2}\,\label{Cevenleft}
\end{align}

Actually, we have a symmetry between the left and the
right chiral null vectors under
the isomorphism
\begin{align}
\bar\chi ^*_{2j-1}\mapsto \psi^*_{2j-1}\,,
\quad
\bar\chi_{2j-1}\mapsto \psi_{2j-1}\,,
\quad 
\bar\psi^*_{2j-1}\mapsto \chi ^*_{2j-1}\,,
\quad
\bar\psi_{2j-1}\mapsto \chi_{2j-1}\,.
\label{lr}
\end{align}
If $1\in I^-$ in \eqref{Cmeven} and $1\in I^+$ in \eqref{Cevenleft} the identification simply follows from \eqref{Ck} with $k=m+1$. 
If $1\notin I^-$ in \eqref{Cmeven} and $1\notin I^+$ in \eqref{Cevenleft}
the proof goes as follows. Denote by $\mathcal{X}$ the image of 
$\overline{\mathcal{C}}_{\mathrm{even}}$ 
under \eqref{lr}. 
The operators
$(\mathcal{X},\mathcal{C}_{\mathrm{even}}, 
[\mathcal{X}, \mathcal{C}_{\mathrm{even}}])$ 
constitute an $\mathfrak{sl}_2$-triple, 
we have the isomorphism \eqref{Ck} for $k=m+2$, and 
$\mathcal{C}_{\mathrm{even}}\mathcal{X}$
is invertible on $H_{-m-2}$. 
So,
$$\mathcal{X}H_{-m-2}=\mathcal{X}
(\mathcal{C}_{\mathrm{even}}\mathcal{X})^{-1}
\mathcal{C}_{\mathrm{even}}^{m+2}H_{m+2}=
\mathcal{C}_{\mathrm{even}}^{m+1}H_{m+2}
\,.
$$

\section{Comparison of null-vectors with CFT results}

If we specialise to $\al=\al_{1,2m+2}$ $(m=0,1,2,\ldots)$ where
\begin{align*}
\al_{1,2m+2}=(2m+1){\textstyle\frac{1-\nu}\nu}, 
\end{align*}
the Verma module $\mathcal{V}_\al$ 
has a singular vector at level $2m+2$. 
Let $w_{2m+2}$ be the singular vector, 
and $\mathcal W_{2m+2}$ the submodule generated by $w_{2m+2}$. 
One expects that
the null vectors in the sG model in the previous section 
should give rise to a fermionic description of
the singular vector $w_{2m+2}$ and the space $\mathcal W_{2m+2}$. 
By using the results of \cite{HGSIV, Boos}, 
in this section we will check that it is indeed so 
(up to level $8$ and modulo local integrals of motion). 
We plan to give some more details including the case of 
both chiralities in a separate publication.

From \eqref{IDENDESC} one derives 
\begin{align}
&\betab^*_{J^+-2m}
\gammab^*_{J^-+2m}\gammab^*_{\Io}
\Phi_{\alpha+2m\frac{1-\nu}{\nu}}(0)
 \cong C'_m(\al)\betab^*_{J^+ }
\gammab^*_{J^-}\bar\gammab^*_{\Io}
\Phi^{(m)}_{\alpha}(0),\label{translate}
\end{align}
where $\#(J^+)=\#(J^-)+m$, $C'_m(\al)=(-)^{m(m+1)/2}\prod_{j=1}^{m}t_{2j-1}(\al)\, C_m(\al)$, and for negative indices
\begin{align}\betab^*_{-a}=t_a(2-\al)\gammab_a\,.\label{conv}\end{align}
Using \eqref{translate} and the results of the previous section we obtain 
for the space
$$\mathcal W_{2m+2}^\mathrm{quo}
=\mathcal W_{2m+2}\ /\ \sum_{k=1}^\infty\mathbf{i}_{2k-1}\mathcal W_{2m+2}\,,$$
the following fermionic basis
\begin{align}
(\mathcal{C}_{m,\mathrm{even}})^{m+1}\betab^*_{I^+}\gammab^*_{I^-}\Phi _{(2m+1)\frac {1-\nu}{\nu}}\,,\quad \#(I^+)=\#(I^-)-2m-2\nn\,,
\end{align}
where
\begin{align}
\mathcal{C}_{m,\mathrm{even}}=\sum\limits _{j=1}^{\infty}\betab^*_{2j+1-2m}\gammab _{2j-1+2m}\,,\nn
\end{align}
with the  convention \eqref{conv}.

For making comparison with CFT results, 
it is convenient to work with 
fermions 
$\betab^{\rm CFT*}_a$, $\gammab^{\rm CFT*}_a$ \cite{HGSIV} 
normalised as
\begin{align*}
&\beta^*_a=D_a(\al)\beta^{\rm CFT*}_a,\
\gamma^*_a=D_a(2-\al)\gamma^{\rm CFT*}_a,\\
&D_a(\al)=
-\sqrt{\frac{i}{\nu}}
\Gamma(\nu)^{-a/\nu}(1-\nu)^{a/2}\frac1{(\frac{a-1}2)!}     
\frac{\Gamma(\frac\al2+\frac a{2\nu})}
{\Gamma(\frac\al2+\frac{a(1-\nu)}{2\nu})}.
\end{align*}

As it has been mentioned several times, at this moment 
we can identify the action of $\betab^*_{2j-1}$, $\gammab^*_{2j-1}$
only 
in the quotient space
$$\mathcal{V}^\mathrm{quo}_\al=
\mathcal{V}_\al\ /\ \sum_{k=1}^\infty\mathbf{i}_{2k-1}\mathcal{V}_\al\,.
$$
As a basis in this space we shall take 
$P(\mathbf{l}_{-2},\mathbf{l}_{-4},\cdots)\Phi_\alpha$
where the $P$ are monomials in the even Virasoro generators
ordered lexicographically. 
Another basis in the same space is provided by monomials in the
fermions, 
$$
\Psi _{a_1,\cdots ,a_p,b_1,\cdots ,b_p}
=\betab^{\rm CFT*}_{a_1}\cdots \betab^{\rm CFT*}_{a_p}
\gammab ^{\rm CFT*}_{b_1}\cdots\gammab^{\rm CFT*}_{b_p}\Phi_\al\,.
$$
Identification of these monomials
with the Virasoro basis (in $\mathcal{V}^{quo}_\al$) 
has been given 
up to level $6$ in \cite{HGSIV}, and at level $8$ in \cite{Boos}.
The general structure is as follows
$$\Psi _{\star}\equiv\(P_{\star}^\mathrm{even}(\mathbf{l}_{-2},\mathbf{l}_{-4},\cdots)
+d_\al P_{\star}^\mathrm{odd}(\mathbf{l}_{-2},\mathbf{l}_{-4},\cdots)\)\Phi_\alpha\,,$$
where $d_\al=\frac 1 6 \sqrt{(25-c)(24\Delta _\al +1-c)}$, and the coefficients of 
$P_{\star}^\mathrm{even}$, $P_{\star}^\mathrm{odd}$ depend only on the central charge $c$ and on the scaling dimension $\Delta_{\al}$. The dependence on $c$ is polynomial while the 
dependence on $\Delta _\al$ is generally
rational, simple poles at certain negative integers may appear.

The operator $\mathcal{C}_{m,\mathrm{even}}$ reads now as
\begin{align}
&\mathcal{C}_{m,\mathrm{even}}
=\sum\limits _{j=1}^{\infty}
\frac {D_{2j-1}((2m-1)\xin)}{D_{4m+2j-3}(2-(2m-1)\xin)}
\betab^{\mathrm{CFT}*}_{2j-1}\gammab ^{\mathrm{CFT}} _{2j+4m-3}\nn\\&-\sum\limits _{j=1}^{m-1}
\frac{t_{2j-1}(2-(2m-1)\xin)}
{D_{2j-1}(2-(2m-1)\xin)
D_{4m-2j-1}(2-(2m-1)\xin)}\gammab ^{\mathrm{CFT}}_{2j-1}\gammab ^{\mathrm{CFT}} _{4m-2j-1}\,,\nn
\end{align}
Actually, after the common multiplier
$\Gamma (\nu)^{\frac {4m-2}\nu}$ is extracted, the remaining combinations of $\Gamma$-functions 
in the right hand side collapse to
rational functions of $\nu$.

The following are the fermionic 
null vectors in 
$W^\mathrm{quo}_{2m+2}$ 
up to level 8: 
\begin{align*}
\begin{matrix}
\mathcal W^\mathrm{quo}_2&\text{level $2$}&\Psi_{1,1}\hfill\\
&\text{level $4$}&\Psi_{3,1}\hfill\\
&\text{level $6$}&\Psi_{5,1},&\Psi_{1,5}-\textstyle{\frac{4(\nu^2-4)}{\nu^2-16}}\Psi_{3,3}\hfill\\
&\text{level $8$}&\Psi_{7,1},&\Psi_{1,7}-\textstyle{\frac{9(\nu^2-4)}{\nu^2-36}}\Psi_{5,3},&\Psi_{1,3,3,1}\\
\mathcal W^\mathrm{quo}_4&\text{level $4$}&\Psi_{1,3}\hfill\\
&\text{level $6$}&\Psi_{3,3}\hfill\\
&\text{level $8$}&\Psi_{5,3},&\Psi_{1,3,3,1}\hfill\\
\mathcal W^\mathrm{quo}_6&\text{level $6$}&\Psi_{1,5}\hfill\\
&\text{level $8$}&\Psi_{3,5}\hfill\\
\mathcal W^\mathrm{quo}_8&\text{level $8$}&\Psi_{1,7}\hfill
\end{matrix}
\end{align*}

In every module the null-vectors of lowest possible 
degree are singular vectors. 
Modulo the integrals of motion, 
the simplest singular vector is 
\begin{align*}
w_2&=\Bigl(\mathbf{l}_{-2}-\frac 1{1-\nu}\mathbf{l}_{-1}^2\Bigr)\Phi _{1,2}
\equiv \mathbf{l}_{-2}\Phi _{1,2}\,, 
\end{align*}
which agrees with the result \cite{HGSIV}
$$
\Psi _{1,1}
\equiv \mathbf{l}_{-2}\Phi _{1,2}\,.
$$
The next 
singular vector $w_4$ is less trivial. 
$$
w_4
\equiv 
\Bigl(\frac12\lb_{-2}^2
-\frac {6\nu^2-16\nu +11} {3(1-\nu)}\lb_{-4}\Bigr)\Phi_{1,4}\,.
$$
On the other hand we have for generic $\al$ \cite{HGSIV}
\begin{align}
\Psi_{1,3}
\equiv
\Bigl(\frac12\lb_{-2}^2
+\frac{c-16-3d_\al}9\lb_{-4}\Bigr)\Phi_{\al}.\label{Psi13}
\end{align}
Substituting $\al=\al_{1,4}$ 
we 
indeed find that $\Psi_{1,3}\equiv w_4$.

The formulae for 
$w_6$ and 
$w_8$ on one hand, and for
$\Psi_{1,5}$, $\Psi_{1,7}$ \cite{HGSIV,Boos} 
on the other, are much more complicated.
Still, 
 substituting $\al =\al_{1,6}$, $\al =\al_{1,8}$ into the latter 
we find perfect agreement. 

Consider the first null-vector 
which is not the singular vector in 
$\mathcal{W}_2$:
$$
\mathbf{l}_{-2}w_2
\equiv \Bigl(\mathbf{l}_{-2}^2-
\frac 2 {1-\nu}\mathbf{l}_{-4}\Bigr)\Phi _{1,2}\,.
$$
It is easy to check that 
$(1/2)\mathbf{l}_{-2}w_2$ 
coincides with the formula from \cite{HGSIV}
\begin{align*}
\Psi_{3,1}
\equiv
\Bigl(\frac12\lb_{-2}^2+\frac{c-16+3d_\al}9\lb_{-4}\Bigr)\Phi_{\al}\,,
\end{align*}
specialised to  $\al =\al _{1,2}$.

Proceeding in the same way we have checked
that all the vectors from the above table indeed coincide with
the Virasoro null-vectors. 
In every particular case we have 
to solve an overdetermined system of linear equations, 
and eventually find solutions.
We consider this fact as 
a strong support to the statement that the 
BJMS fermions when they are identified with the BBS fermions
create the null vectors in the Verma modules.

\bigskip

{\it Acknowledgements.}\quad
\medskip

Research of MJ is supported by the Grant-in-Aid for Scientific 
Research B-23340039.
Research of TM is supported by the Grant-in-Aid for Scientific Research
B-22340031.
Research of FS is supported by   SFI
under Walton Professorship scheme, by
RFBR-CNRS grant 09-02-93106
and DIADEMS program (ANR) contract number BLAN012004.
MJ and TM 
would like to thank for the hospitality extended by the Hamilton Mathematical Institute 
where a part of this work was begun.



\begin{thebibliography}{10}

\bibitem{HGSIV}
H.~Boos, M.~Jimbo, T.~Miwa, and F.~Smirnov.
\newblock Hidden {Grassmann} structure in the {XXZ} model {IV}: {CFT} limit.
\newblock {\em Commun. Math. Phys.}, {\bf 299}:825--866, 2010.

\bibitem{OP}
M.~Jimbo, T.~Miwa, and F.~Smirnov.
\newblock On one-point functions of descendants in {sine-Gordon} model.
\newblock {\em New Trends in Quantum Integrable Systems: Proceedings of the
  Infinite Analysis 09, World Scientific Publishing, Singapore}, pages
  117--137, 2010.

\bibitem{HGSV}
M.~Jimbo, T.~Miwa, and F.~Smirnov.
\newblock Hidden {Grassmann} structure in the {XXZ} model {V}: {sine-Gordon}
  model.
\newblock {\em Lett. Math. Phys.}, {\bf 96}:325--365, 2011.

\bibitem{zamS}
A.~Zamolodchikov.
\newblock Exact {S}-matrix of quantum sine-{G}ordon solitons.
\newblock {\em JETP Lett.}, {\bf 25}:468--481, 1977.

\bibitem{smi86}
F.~Smirnov.
\newblock The general formula for solitons form factors in {sine-Gordon} model.
\newblock {\em J. Phys.}, {\bf A19}:L575, 1986.

\bibitem{smikir}
A.~Kirillov and F.~Smirnov.
\newblock A representation of the current algebra connected with the
  {SU(2)}-invariant Thirring model.
\newblock {\em Phys. Lett.}, {\bf B 198}:506--510, 1987.
\bibitem{book}
F.~Smirnov.
\newblock {\em Form {Factors} in {Completely} {Integrable} {Models} of
  {Quantum} {Field} {Theory}}.
\newblock World Scientific, 1992.

\bibitem{alzam}
Al. Zamolodchikov.
\newblock Two point correlation function in scaling {Lee-Yang} model.
\newblock {\em Nucl. Phys.}, {\bf B348}:619--641, 1991.


\bibitem{Lukalpha}
S.~Lukyanov.
\newblock Form-factors of exponential fields in the {sine-Gordon} model.
\newblock {\em Mod. Phys. Lett.}, {\bf A12}:2543--2550, 1997.

\bibitem{BPZ}
A.A. Belavin, A.M. Polyakov, and A.B. Zamolodchikov.
\newblock Infinite conformal symmetry in two-dimensional quantum field theory.
\newblock {\em Nucl. Phys.}, B241:333--380, 1984.

\bibitem{smi90}
F.~Smirnov.
\newblock Reductions of the {sine-Gordon} model as a perturbation of minimal
  models of conformal field theory.
\newblock {\em Nucl. Phys.}, {\bf B 337}:156--180, 1990.

\bibitem{ReshSmi}
N.~Reshetikhin {and}~F. Smirnov.
\newblock Hidden quantum group symmetry and integrable perturbations of
  conformal field theories.
\newblock {\em Commun. Math. Phys.}, {\bf 131}:157--177, 1990.

\bibitem{LukZam}
S.~Lukyanov and A.~Zamolodchikov.
\newblock Exact expectation values of local fields in quantum {sine-Gordon}
  model.
\newblock {\em Nucl.Phys.}, {\bf B493}:571--587, 1997.

\bibitem{JM}
M.~Jimbo and T.~Miwa.
\newblock {\em Algebraic {Analysis} of {Solvable} {Lattice} {Models}},
  volume~{\bf 85}.
\newblock AMS, 1995.

\bibitem{smiabel}
F.~Smirnov.
\newblock Form-factors, deformed {Knizhnik-Zamolodchikov} equations and finite
  gap integration.
\newblock {\em Commun. Math. Phys.}, {\bf 155}:459--487, 1993.

\bibitem{smiriemann}
F.~Smirnov.
\newblock On the deformation of {Abelian} integrals.
\newblock {\em Lett. Math. Phys.}, {\bf 36}:267--275, 1996.

\bibitem{BBS}
O.~Babelon, D.~Bernard, and F.~Smirnov.
\newblock Null-vectors in integrable field theory.
\newblock {\em Commun. Math. Phys.}, {\bf 186}:601--648, 1997.

\bibitem{Fothers}
B.~Feigin, M.~Jimbo, M.~Kashiwara, T.~Miwa, and E.~Mukhin {and}~Y. Takeyama.
\newblock A functional model for the tensor product of level $1$ highest and
  level $-1$ lowest modules for the quantum affine algebra
  {$U_q(\widehat{\mathfrak{sl}}_2)$}.
\newblock {\em Eur. J. Combinatorics}, {\bf 25}:1197--1229, 2004.

\bibitem{FST}
E.~Sklyanin, L.~Takhtajan, and L.~Faddeev.
\newblock Quantum inverse problem method.
\newblock {\em Teor. Mat. Fiz}, {\bf 40}:194--220, 1979.

\bibitem{HGS}
H.~Boos, M.~Jimbo, T.~Miwa, F.~Smirnov, and Y.~Takeyama.
\newblock Hidden {Grassmann} structure in the {XXZ} model.
\newblock {\em Commun. Math. Phys.}, {\bf 272}:263--281, 2007.

\bibitem{HGSII}
H.~Boos, M.~Jimbo, T.~Miwa, F.~Smirnov, and Y.~Takeyama.
\newblock Hidden {Grassmann} structure in the {XXZ} model {II : Creation}
  operators.
\newblock {\em Commun. Math. Phys.}, {\bf 286}:875--932, 2009.

\bibitem{Drinfeld}
V.~Drinfeld.
\newblock Quantum {Groups}.
\newblock In {\em Proceedings of the {International} {Congress} of
  {Mathematicians}}, pages 798--820, Berkeley, 1990.

\bibitem{Jimbo}
M.~Jimbo.
\newblock A $q$-difference analogue of {$U(\mathfrak{g})$} and the
  {Yang-Baxter} equation.
\newblock {\em Lett. Math. Phys.}, {\bf 10}:63--69, 1985.

\bibitem{BLZIII} V.~Bazhanov, S.~Lukyanov and A.~Zamolodchikov.
\newblock 
Integrable structure of conformal field theory {III}:
	the {Yang-Baxter} relation.
\newblock 
{\em Commun. Math. Phys.}, \textbf{200}:297--324, 1999. 

\bibitem{HGSIII}
M.~Jimbo, T.~Miwa, and F.~Smirnov.
\newblock Hidden {Grassmann} structure in the {XXZ} model {III}: {Introducing
  Matsubara} direction.
\newblock {\em J. Phys. A:Math.Theor.}, \textbf{42}:304018, 2009.

\bibitem{BG}
H.~Boos and F.~G{\"o}hmann.
\newblock On the physical part of the factorized correlation functions of the
  {XXZ} chain.
\newblock {\em J. Phys.}, {\bf A42}:1--27, 2009.

\bibitem{DestriDeVega}
C.~Destri and H.~de~Vega.
\newblock Unified approach to thermodynamic {Bethe Ansatz} and finite size
  corrections for lattice models and field theories.
\newblock {\em Nucl. Phys.}, {\bf B438}:413--454, 1995.

\bibitem{Boos}
H.~Boos.
\newblock Fermionic basis in conformal field theory and thermodynamic {Bethe
  Ansatz} for excited states.
\newblock {\em SIGMA}, {\bf 7}:007, 36p, 2011.

\bibitem{FFLZZ}
V.~Fateev, D.~Fradkin, S.~Lukyanov, A.~Zamolodchikov, and Al. Zamolodchikov.
\newblock Expectation values of descendent fields in the {sine-Gordon} model.
\newblock {\em Nucl. Phys.}, {\bf B540}:587--609, 1999.

\bibitem{compl}
H.~Boos, M.~Jimbo, T.~Miwa, and F.~Smirnov.
\newblock Completeness of a fermionic basis in the homogeneous {$XXZ$} model.
\newblock {\em J. Math. Phys.}, {\bf 50}:095206 (online), 2009.

\bibitem{DestriLow}
C.~Destri {and} J.~H. Lowenstein.
\newblock Analysis of the {Bethe}-ansatz equations of the chiral-invariant
  {Gross-Neveu} model.
\newblock {\em Nucl. Phys.}, {\bf 205B}:369--385, 1982.

\bibitem{Izergin}
A.~Izergin.
\newblock Partition function of a six-vertex model in a finite volume.
\newblock {\em Sov. Phys. Dokl.}, \textbf{32}:878--879, 1987.

\end{thebibliography}
\end{document}